\documentclass[10pt,journal,compsoc]{IEEEtran}

\usepackage{balance}
\usepackage{amsmath,amssymb,bm,bbm}
\usepackage[ruled,vlined,linesnumbered]{algorithm2e}
\DontPrintSemicolon
\renewcommand{\CommentSty}[1]{\textnormal{\textsf{\textit{#1}}}\unskip}
\SetKwProg{fnDef}{def}{}{}
\usepackage{txfonts}
\SetKwComment{tcp}{$\rhd$\, }{}

\usepackage{etoolbox}
\makeatletter
\patchcmd{\@makecaption}
  {\scshape}
  {}
  {}
  {}
  \makeatother

\usepackage{color}
\usepackage{array}
\usepackage{graphicx}
\usepackage{epstopdf}
\usepackage{xspace}
\usepackage{booktabs}
\usepackage{url}
\usepackage[amssymb]{ktmath}
\usepackage{multirow}

\usepackage[percent]{overpic}

\usepackage{enumitem}

\usepackage[nocompress]{cite}

\usepackage[caption=false,font=footnotesize,labelfont=sf,textfont=sf]{subfig}

\newcommand*{\eps}{\vareps}
\newtheorem{observation}{Observation}
\newtheorem{problem}{Problem}

\newcommand{\remove}[1]{}
\newcommand{\edist}[2]{\Arrowvert {#1} - {#2} \Arrowvert}

\newcommand{\stream}{{\cal S}}

\newcommand{\coreset}{\texttt{coreset}}
\newcommand{\km}{\texttt{$k$-means}\xspace}

\newcommand{\kmpp}{\texttt{$k$-means++}\xspace}
\newcommand{\skmpp}{\texttt{streamkm++}\xspace}

\newcommand{\major}{\texttt{major}\xspace}
\newcommand{\minor}{\texttt{minor}\xspace}
\newcommand{\prefixsum}{\texttt{prefixsum}\xspace}

\newcommand{\clusterupdate}{\texttt{StreamCluster-Update}\xspace}
\newcommand{\clusterquery}{\texttt{StreamCluster-Query}\xspace}
\newcommand{\object}{\mathcal{H}}

\newcommand{\cstree}{\texttt{CT}\xspace}
\newcommand{\ct}{\texttt{CT}\xspace}
\newcommand{\ctinit}{\texttt{CT-Init}\xspace}
\newcommand{\ctupdate}{\texttt{CT-Update}\xspace}
\newcommand{\ctcoreset}{\texttt{CT-Coreset}\xspace}

\newcommand{\cc}{\texttt{CC}\xspace}

\newcommand{\cctree}{\texttt{CC}\xspace}
\newcommand{\ccupdate}{\texttt{CC-Update}\xspace}
\newcommand{\ccinit}{\texttt{CC-Init}\xspace}
\newcommand{\cccoreset}{\texttt{CC-Coreset}\xspace}

\newcommand{\rcc}{\texttt{RCC}\xspace}
\newcommand{\rccinit}{\texttt{RCC-Init}\xspace}
\newcommand{\rccupdate}{\texttt{RCC-Update}\xspace}
\newcommand{\rcccoreset}{\texttt{RCC-Coreset}\xspace}

\newcommand{\seqkm}{\texttt{Sequential $k$-means}\xspace}
\newcommand{\hybrid}{\texttt{OnlineCC}\xspace}
\newcommand{\hybridinit}{\texttt{\hybrid-Init}\xspace}
\newcommand{\hybridupdate}{\texttt{\hybrid-Update}\xspace}
\newcommand{\hybridquery}{\texttt{\hybrid-Query}\xspace}

\usepackage{microtype} 

\begin{document}

\title{Streaming $k$-Means Clustering with Fast Queries}

\author{Yu Zhang,
        Kanat Tangwongsan,
        Srikanta Tirthapura
\IEEEcompsocitemizethanks{
\IEEEcompsocthanksitem Yu Zhang is with the Department of Electrical and Computer Engineering, Iowa State University, Ames, IA, 50011.\protect\\
E-mail: yuz1988@iastate.edu
\IEEEcompsocthanksitem Kanat Tangwongsan is with the Compute Science Program, Mahidol University International College, Thailand.\protect\\
Email: kanat.tan@mahidol.edu
\IEEEcompsocthanksitem Srikanta Tirthapura is with the Department of Electrical and Computer Engineering, Iowa State University, Ames, IA, 50011.\protect\\
E-mail: snt@iastate.edu
}
}

\IEEEtitleabstractindextext{%
\begin{abstract}
  We present methods for $k$-means clustering on a stream with a focus on
  providing fast responses to clustering queries. Compared to the current
  state-of-the-art, our methods provide substantial improvement in the query
  time for cluster centers while retaining the desirable properties of provably
  small approximation error and low space usage.  Our algorithms rely on a novel
  idea of ``coreset caching'' that systematically reuses coresets (summaries of
  data) computed for recent queries in answering the current clustering
  query. We present both theoretical analysis and detailed experiments
  demonstrating their correctness and efficiency.
\end{abstract}

}

\maketitle

\IEEEdisplaynontitleabstractindextext
\IEEEpeerreviewmaketitle

\IEEEraisesectionheading{\section{Introduction}\label{sec:intro}}
\IEEEPARstart{C}{lustering} is a fundamental method for understanding and interpreting data that
seeks to partition input objects into groups, known as \emph{clusters}, such
that objects within a cluster are similar to each other, and objects in
different clusters are not.  A clustering formulation called $k$-means is
simple, intuitive, and widely used in practice.  Given a set of points $S$ in a
Euclidean space and a parameter $k$, the objective of $k$-means is to partition
$S$ into $k$ clusters in a way that minimizes the sum of the squared distance
from each point to its cluster center.

In many cases, the input data is not available all at once but arrives as a
continuous, possibly unending, sequence.  This variant, known as \emph{streaming
$k$-means clustering}, requires an algorithm to maintain enough state to be
able to incrementally update the clusters as more data arrive.  Furthermore,
when a query is posed, the algorithm is required to return $k$ cluster centers,
one for each cluster, for the data observed so far.

Prior work on streaming $k$-means (e.g.~\cite{AMR+12,AJM09,GMM+03,SWM11}) has
mainly focused on optimizing the memory requirements, leading to algorithms with
provable approximation quality that only use space polylogarithmic in the input
stream's size~\cite{AMR+12,AJM09}.  However, these algorithms require an
expensive computation to answer a query for cluster centers.  This can be a
serious problem for applications that need answers in (near) real-time, such as
in network monitoring and sensor data analysis.  Our work aims at improving the
clustering query-time while keeping other desirable properties of current
algorithms.

To understand why current solutions to streaming $k$-means clustering have a high query runtime, let us review the
framework they use. At a high level, an incoming data stream $\stream$ is divided into smaller ``chunks''
$\stream_1,\stream_2,\ldots,$. Each chunk is summarized into a compact
representation, known as a ``coreset'' (for example, see~\cite{HM04}). The
resulting coresets may still not all fit into the memory of the processor.
Hence, multiple coresets are further merged recursively into higher-level
coresets, forming a hierarchy of coresets, or a ``coreset tree''.  When a query
arrives, all active coresets in the coreset tree are merged together, and a
clustering algorithm such as $\kmpp$~\cite{AV07} is applied on the result,
outputting $k$ cluster centers. The query time is consequently proportional to
the number of coresets that need to be merged together. In prior algorithms, the
total size of all these coresets could be as large as the whole memory itself,
which causes the query time to often be prohibitively high.

\subsection{Our Contributions}
We present three algorithms (\cc{}, \rcc{}, and \hybrid{}) for
streaming $k$-means clustering that asymptotically and practically improve upon
prior algorithms' response time of a query while retaining guarantees on memory
efficiency and solution quality of the current state-of-the-art.  We provide
theoretical analysis, as well as extensive empirical evaluation, of the proposed
algorithms.

At the heart of our algorithms is the idea of ``coreset caching'' that to our
knowledge, has not been used before in streaming clustering. It works by reusing
coresets that have been computed during previous (recent) queries to speedup the
computation of a coreset for the current query. In this way, when a query
arrives, the algorithm has no need to combine all coresets currently in memory;
it only needs to merge a coreset from a recent query (stored in the coreset
cache) along with coresets of points that arrived after this
query.

\begin{table*}[ht]
{
\small
\centering
\begin{tabular}{|c| c| c| c| c|}
\hline
Name & Query Cost & Update Cost & Memory Used & Accuracy: Coreset level  \\
           & (per point)   & (per point)    & & returned at query  \\
\hline
Coreset Tree ($\ct$) & $O\left( \frac{kdm}{q} \cdot \frac{r \log N}{\log r}\right)$ &  $O(dm)$ & $O\left(dm \cdot \frac{r \log N}{\log r}\right)$ & $\log N / \log r$ \\
\hline
Cached Coreset Tree ($\cc$)&  $O\left(\frac{kdm}{q} \cdot r \right)$ & $O(dm)$ & $O\left(dm \cdot \frac{r\log N}{\log r} \right)$ & $2  \log N / \log r$\\
\hline
Recursive Cached Coreset Tree ($\rcc$) & $O\left(\frac{kdm}{q} \cdot \log \log N \right)$ & $O(dm \log \log N)$ & $O(dm N^{1/8})$ & $O(1)$ \\
\hline
Online Coreset Cache ($\hybrid$)  & usually $O(1)$ worst case $O\left(\frac{kdm}{q} \cdot r \right)$ & $O(dm)$ & $O\left(dm \cdot \frac{r\log N}{\log r} \right)$  & $2 \log N / \log r$  \\
\hline
\end{tabular}
\smallskip
}
\caption{The time and accuracy of different clustering algorithms. For algorithms other than $\ct$, we assume answering queries by using the coreset cache.}
\label{table:summary}
\end{table*}


Our theoretical results are summarized in Table~\ref{table:summary}.
Throughout, let $n$ denote the number of points observed in the stream so far.
We measure an algorithm's performance in terms of both running time (further
separated into query and update) and memory consumption.  The query cost, stated
in terms of $q$, represents the expected amortized cost per input point assuming
that the total number of queries does not exceed $n/q$---or that the average
interval between queries is $\Omega(q)$. The update cost is the average (i.e.,
amortized) per-point processing cost, taken over the entire stream.  The memory
cost is expressed in terms of words assuming that each point is $d$-dimensional
and can be stored in $O(d)$ words.  Furthermore, let $m$ denote a user-defined
parameter that determines the coreset size ($m$ is set independent of $n$ and is
often $O(k)$ in practice); $r$ denote a user-defined parameter that sets the
merge degree of the coreset tree; and $N = n/m$ be the number of ``base
buckets.''

In terms of accuracy, each of our algorithms provides a provable
$O(\log k)$-approximation to the optimal $k$-means solution---that is, the
quality is comparable, up to constants, to what we would obtain if we simply
stored all the points so far in memory and ran an (expensive) batch $\kmpp$
algorithm at query time.  Further scrutiny reveals that for the same target
accuracy (holding constants in the big-$O$ fixed), our simplest algorithm
``Cached Coreset Tree'' ($\cc$) improves the query runtime of a current
state-of-the-art, \ct{}\footnote{$\ct$ is essentially the \skmpp
  algorithm~\cite{AMR+12} and~\cite{AJM09} except it has a more flexible rule
  for merging coresets.}, by a factor of $O(\log N)$---at the expense of a
(small) constant factor increase in memory usage.  If more flexibility in
tradeoffs among the parameters is desired, coreset caching can be applied
recursively.  Using this scheme, the ``Recursive Cached Coreset Tree'' ($\rcc$)
algorithm can be configured so that it has a query runtime that is a factor of
$O(r)$ times faster than \ct, and yields better solution quality
than \ct{}, at the cost of a small polynomial factor increase in memory.

In practice, a simple sequential streaming clustering algorithm, due to
MacQueen~\cite{MacQueen67}, is known to be fast but lacks good theoretical
properties.  We derive an algorithm, called \hybrid{}, that combines ideas from
\cc and sequential streaming to further enhance clustering query runtime while
keeping the provable clustering quality of \rcc and \cc.

We also present an extensive empirical study of the proposed algorithms, in
comparison to the state-of-the-art $k$-means algorithms, both batch and
streaming.  The results show that our algorithms yield substantial speedups
(5--100x) in query runtime and total time, and match the accuracy of \skmpp for
a broad range of datasets and query arrival frequencies.

\subsection{Related Work}
\label{sec:related}

When all input is available at the start of computation (batch setting), Lloyd's
algorithm~\cite{Lloyd82}, also known as the \km algorithm, is a simple,
widely-used algorithm.  Although it has no quality guarantees, heuristics such
as \kmpp~\cite{AV07} can generate a starting configuration such that Llyod's
algorithm will produce provably-good clusters.

In the streaming setting, \cite{MacQueen67} is the earliest streaming
$k$-means method, which maintains the current cluster centers and applies one
iteration of Lloyd's algorithm for every new point received.  Because it is fast
and simple, this sequential algorithm is commonly used in practice (e.g., Apache Spark
mllib~\cite{MBY+15}).  However, it cannot provide any guarantees on the
quality~\cite{KMN+04}.  BIRCH~\cite{ZRL96} is a streaming clustering method
based on a data structure called the CF Tree; it produces cluster centers
through agglomerative hierarchical clustering on the leaf nodes of the
tree. CluStream\cite{AHW+03} constructs ``microclusters'' that summarize subsets
of the stream, and further applies a weighted \km algorithm on the
microclusters. STREAMLS~\cite{GMM+03} is a divide-and-conquer method based on
repeated application of a bicriteria approximation algorithm for clustering. A
similar divide-and-conquer algorithm based on \kmpp is presented
in~\cite{AJM09}. Invariably, these methods have high query-processing cost and
are not suitable for applications that require fast query response.  In
particular, at the time of query, they require merging multiple data structures,
followed by an extraction of cluster centers, which is expensive.

Har-Peled and Mazumdar~\cite{HM04} present coresets of size
$O(k \eps^{-d} \log n)$ for summarizing $n$ points \km, and also show how to
use the merge-and-reduce technique based on the Bentley-Saxe
decomposition~\cite{BS80} to derive a small-space streaming algorithm using
coresets. Further work~\cite{HK07, FL11, FSS13} reduced the size of a \km coreset
to $O(k/\eps^2)$. A close cousin to ours, \skmpp~\cite{AMR+12}
(essentially the \ct{} scheme) is a streaming $k$-means clustering algorithm
that uses the merge-and-reduce technique along with \kmpp to generate a
coreset. Our work improves on \skmpp w.r.t. query runtime.

\noindent\textbf{Roadmap:} We present preliminaries in Section~\ref{sec:prelim},
background for streaming clustering in Section~\ref{sec:background} and then the
algorithms $\cc$, $\rcc$, and $\hybrid$ in Section~\ref{sec:algorithm}, along
with their proofs of correctness and quality guarantees. We then present
experimental results in Section~\ref{sec:expts}.

\section{Preliminaries and Notation}
\label{sec:prelim}
We work with points from the $d$-dimensional Euclidean space $\R^d$ for integer
$d > 0$. Each point is associated with a positive integral weight ($1$ if
unspecified).  For points $x,y \in \R^d$, let $D(x,y) =\edist{x}{y}$ denote the
Euclidean distance between $x$ and $y$.  Extending this notation, the distance
from a point $x \in \R^d$ to a point set $\Psi \subset \R^d$ is
$D(x,\Psi) = \min_{\psi \in \Psi}\edist{x}{\psi}$.  In this notation, the
\emph{$k$-means clustering problem} is as follows:

\begin{problem}[$k$-means Clustering] 
  Given a set $P \subseteq \R^d$ with $n$ points and a weight function
  $w\!: P \to \Z^+$, find a point set $\Psi \subseteq \R^d$, $|\Psi| = k$, that
  minimizes the objective function
\[
\phi_{\Psi}(P) = \sum_{x \in P} w(x) \cdot D^2(x,\Psi) = \sum_{x \in
  P}\min_{\psi \in \Psi} \left (w(x) \cdot {\edist{x}{\psi}}^2 \right).
\]
\end{problem}

\noindent\textbf{Streams:} A stream $\stream = e_1, e_2, \ldots$ is an ordered
sequence of points, where $e_i$ is the $i$-th point observed by the algorithm.
For $t > 0$, let $\stream(t)$ denote the first $t$ entries of the stream:
$e_1, e_2, \ldots, e_t$. For $0 < i \le j$, let $\stream(i,j)$ denote the
substream $e_i, e_{i+1}, \ldots, e_j$.
Define $S = \stream(1,n)$ be the whole stream observed until $e_n$, where $n$
is, as before, the total number of points observed so far.


\noindent\textbf{\kmpp Algorithm:}
Our algorithms rely on a batch algorithm as a subroutine: \kmpp algorithm~\cite{AV07}, 
which has the following properties:
\begin{theorem}[Theorem 3.1 in \cite{AV07}]
\label{theo:kmeans++}
On an input set of $n$ points $P \subseteq \R^d$, the $k$-means++ algorithm returns a set
$\Psi$ of $k$ centers such that $\expct{\phi_{\Psi}(P)} \le 8(\ln k + 2) \cdot \phi_{OPT}(P)$ where $\phi_{OPT}(P)$ is
the optimal $k$-means clustering cost for $P$. The time complexity of the
algorithm is $O(kdn)$.
\end{theorem}

\noindent\textbf{Coresets and Their Properties:}
Our clustering method builds on the concept of a \emph{coreset}, a small-space
representation of a weighted point set that (approximately) preserves desirable
properties of the original point set.  A variant suitable for $k$-means
clustering is as follows:
\begin{definition}[$k$-means Coreset]
\label{defn:coreset}
For a weighted point set $P \subseteq \R^d$, integer $k > 0$, and parameter
$0 < \eps < 1$, a weighted set $C \subseteq \R^d$ is said to be a
$(k, \eps)$-coreset of $P$ for the $k$-means metric, if for any set $\Psi$
of $k$ points in $\R^d$, we have
\[(1-\eps) \cdot \phi_{\Psi}(P) \le \phi_{\Psi}(C) \le (1+\eps) \cdot \phi_{\Psi}(P)\]
\end{definition}

Throughout, we use the term ``coreset'' to always refer to a $k$-means coreset.
When $k$ is clear from the context, we simply say an $\eps$-coreset.  For
integer $k > 0$, parameter $0 < \eps < 1$, and weighted point set
$P \subseteq \R^d$, we use the notation $\coreset(k, \eps, P)$ to mean a
$(k,\eps)$-coreset of $P$. We use the following observations from~\cite{HM04}.

\begin{observation}[\cite{HM04}]
\label{obs:coreset1}
If $C_1$ and $C_2$ are each $(k,\eps)$-coresets for disjoint multi-sets $P_1$
and $P_2$ respectively, then $C_1 \cup C_2$ is a $(k, \eps)$-coreset for
$P_1 \cup P_2$.
\end{observation}

\begin{observation}[\cite{HM04}]
  \label{obs:coreset2}
  Let $k$ be fixed.  If $C_1$ is $\eps_1$-coreset for $C_2$, and $C_2$ is a
  $\eps_2$-coreset for $P$, then $C_1$ is a $((1+\eps_1)(1+\eps_2)-1)$-coreset
  for $P$.
\end{observation}

While our algorithms can work with any method for constructing coresets, 
an algorithm due to ~\cite{FSS13} by Feldman, Schimidt and Sohler provides the following
guarantees, which is best coreset construction algorithm from our knowledge:
\begin{theorem}[\cite{FSS13} Corollary $4.5$]
\label{theo:coreset-build}
Given a point set $P$ with $n$ points, there exists an algorithm to compute $\coreset(k, \epsilon, P)$ 
with size $O(k / \epsilon^2)$. 
Let the coreset size be denoted by $m$, then the time complexity of constructing the coreset is $O(dnm)$. 
\end{theorem}

\section{Streaming Clustering and Coreset Trees}
\label{sec:background}
To provide context for how algorithms in this paper will be used, we describe a
generic ``driver'' algorithm for streaming clustering.  We also discuss the
coreset tree (\ct) algorithm.  This is both an example of how the driver works
with a specific implementation and a quick review of an algorithm from prior
work that our algorithms build upon.

\begin{algorithm}[tb]
\caption{Stream Clustering Driver}\label{algo:cluster-init}
\label{algo:cluster-update}
\label{algo:cluster-query}
\fnDef{$\clusterupdate(\object, p)$}{
  \tcp{Insert new point $p$ from the stream into $\object$}
  Add $p$ to $\object.C$\;
  
  \If{$(|\object.C| = m)$}
  {
    $\object.\mathcal{D}$.Update($\object.C$)\;
    $\object.C \gets \emptyset$\;
  }
}
\fnDef{$\clusterquery()$}{
  $C_1 \gets \object.\mathcal{D}.Coreset()$\;
  \Return $\kmpp(k, C_1 \cup \object.C)$\;
}
\end{algorithm}

\subsection{Driver Algorithm}
The ``driver'' algorithm (presented in Algorithm~\ref{algo:cluster-update}) internally keeps state as an object $\object$. The state $\object$ involves a specific implementation of a clustering data structure $\mathcal{D}$ and an auxiliary point set $C$. The point set $C$ receives every new point from the stream and with maximum capacity $m$. The size $m$ is the coreset size where the value is determined by the coreset construction algorithm. Once the size of $C$ increases to $m$, the $m$ points will be inserted into the clustering data structure $\mathcal{D}$, as well $C$ will be emptied.  So $C$ groups arriving points into batches at the granularity of size $m$ and stores the current batch.  Subsequent algorithms in this paper, including \ct, are implementations for the clustering data structure $\mathcal{D}$.

\subsection{\ct: $r$-way Merging Coreset Tree}
\label{sec:cstree}

\newcommand*{\qtree}[0]{\ensuremath{Q}\xspace}
The \emph{$r$-way coreset tree} (\ct{}) turns a traditional batch algorithm for
coreset construction into a streaming algorithm that works in limited space.
Although the basic ideas are the same, our description of \ct{} generalizes the
coreset tree of Ackermann et~al.~\cite{AMR+12}, which is the special case when
$r = 2$.

\remove{
\begin{figure*}[tb]
  \centering
  \includegraphics[width=.99\textwidth]{figs/coreset-tree-smaller-examples}
  \caption{Illustration of a coreset tree with $r = 3$.}
\label{fig:ct-example}
\end{figure*}
}

\noindent\textbf{The Coreset Tree:}
A coreset tree $\qtree$ maintains \emph{buckets} at multiple levels.  The
buckets at level $0$ are called \emph{base buckets}, which contain the original
input points.  The size of each base bucket is specified by a parameter $m$.
Each bucket above that is a coreset summarizing a segment of the stream observed
so far.  In an $r$-way \ct, level $\ell$ has between $0$ and $r-1$ (inclusive)
buckets, each is a summary of $r^\ell$ base buckets.

Initially, the coreset tree is empty. After observing $n$ points in the stream,
there will be $N = \lfloor n/m \rfloor$ base buckets (level 0). Some of these
base buckets may have been merged into higher-level buckets.  The distribution
of buckets across levels obeys the following invariant:
\begin{quote}
  If $N$ is written in base $r$ as $N = (s_q, s_{q-1}, \dots, s_1, s_0)_r$, with
  $s_q$ being the most significant digit (i.e., $N = \sum_{i=0}^q s_ir^i$), then
  \emph{there are exactly $s_i$ buckets in level $i$}.
\end{quote}


\begin{algorithm}[t]
\caption{Coreset Tree Algorithm}\label{algo:ctfunctions}
\label{algo:ctinit}
\label{algo:ctupdate}
\label{algo:ctquery}
\tcp{\emph{Input:} bucket $b$}
\fnDef{$\ctupdate(b)$}{
  Append $b$ to $Q_0$\;
  $j \gets 0$\;
  \While{$|Q_j| \ge r$} {
    $U \gets \coreset(k, \eps, \cup_{B \in Q_j} B)$\;
    Append $U$ to $Q_{j+1}$\;
    $Q_j \gets \emptyset $\;
    $j \gets j+1$ \;
  }
}
\tcp{For query method $\clusterquery()$}
\fnDef{\ctcoreset()}{
  \Return{$\bigcup_{j} \{\bigcup_{B \in Q_j} B \}$}\;
}
\end{algorithm}

\noindent{}\emph{\textbf{How is a base bucket added?}}
The process to add a base bucket is reminiscent of incrementing a base-$r$
counter by one, where merging is the equivalent of transferring the carry from
one column to the next.  More specifically, \ct maintains a sequence of
sequences $\{Q_j\}$, where $Q_j$ is the buckets at level $j$. To incorporate a
new bucket into the coreset tree, \ctupdate{}, presented in
Algorithm~\ref{algo:ctupdate}, first adds it at level $0$. When the number of
buckets at any level $i$ of the tree reaches $r$, these buckets are merged,
using the coreset algorithm, to form a single bucket at level $(i+1)$, and the
process is repeated until there are fewer than $r$ buckets at all levels of the
tree. An example of how the coreset tree evolves after the addition of base buckets 
is shown in Figure~\ref{fig:coreset-tree}.

\begin{figure}
  \includegraphics[width=0.48\textwidth, height=6cm]{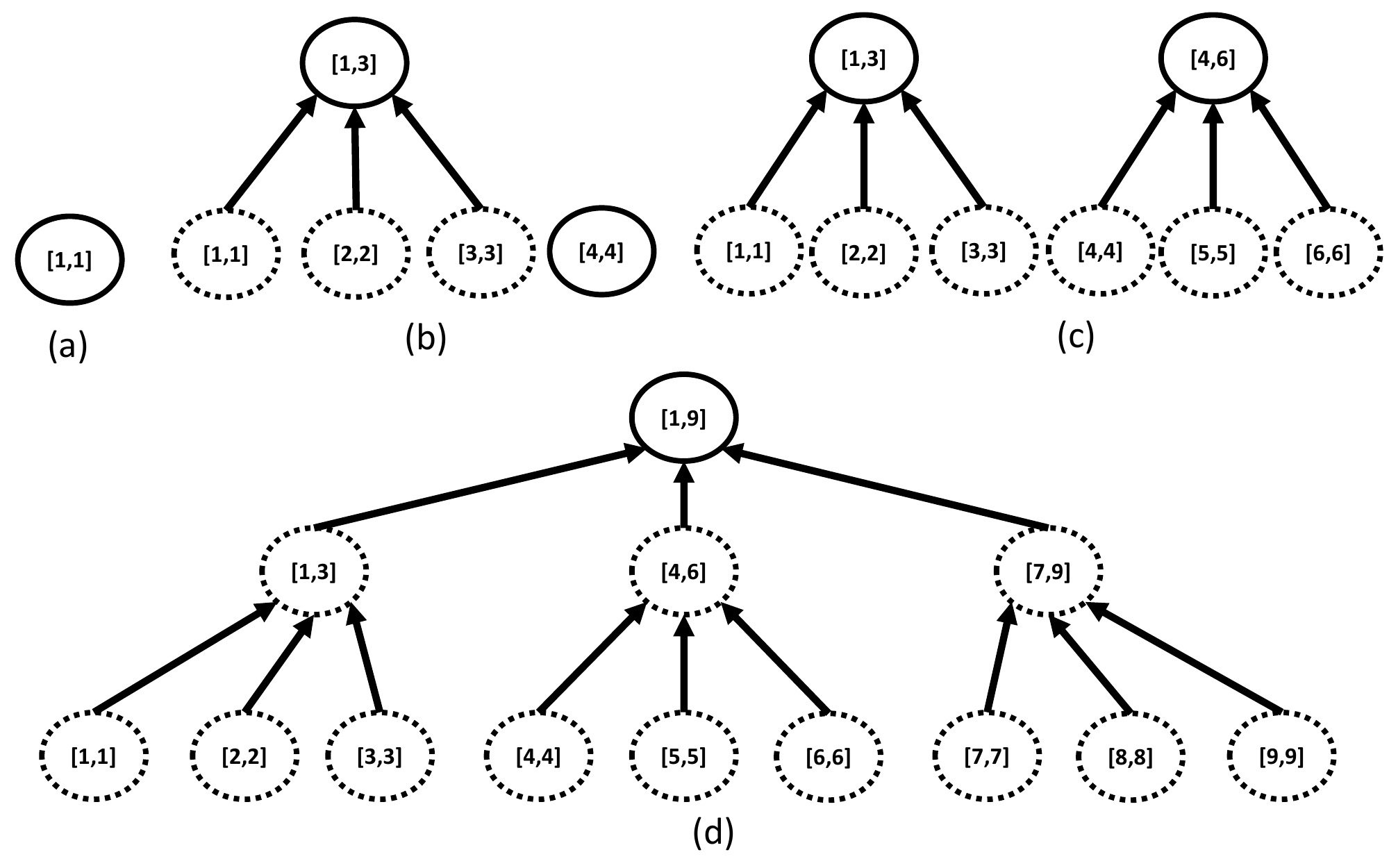}
  \caption{Illustration of a $3$-way coreset tree (\ct), showing the states after receiving number of base buckets 
  (a) $1$, (b) $4$, (c) $6$, (d) $9$. The notation $[l,r]$ denotes a coreset of all points in buckets number from $l$ to $r$, both endpoints inclusive. The coreset tree consists of coresets represented by ellipses, each is a base bucket (level $0$) or has been formed by merging multiple coresets from lower levels. A coreset becomes inactive when it is merged and represented by a dotted bordered ellipse.}
\label{fig:coreset-tree}
\end{figure}

\noindent{}\emph{\textbf{How to answer a query?}}  The algorithm simply unions all the active buckets together, specifically
$\bigcup_{j} \{\bigcup_{B \in Q_j} B \}$.  Notice that the driver will combine
this with a partial base bucket before deriving the $k$-means centers.

Following lemmas stating the properties of the \ct algorithm. We use the following definition in proving clustering guarantees.
\begin{definition}[Level-$\ell$ Coreset]
  For $\ell \in \Z_{\geq 0}$, \emph{a $(k,\eps, \ell)$-coreset} of a point set
  $P \subseteq \R^d$, denoted by $\coreset(k, \eps, \ell, P)$, is as follows:
The level-$0$ coreset of $P$ is $P$. 
For $\ell > 0$, a level-$\ell$ coreset of $P$ is a coreset of the union
of $C_i$'s (i.e., $\coreset(k, \eps, \cup_{i=1}^t C_i)$), where each
$C_i$ is a level-$\ell_i$ coreset, $\ell_i < \ell$, of $P_i$ such that
$\{P_j\}_{j=1}^t$ forms a partition of $P$.
\end{definition}

\begin{lemma}
\label{lemma:cstree-level2}
For a point set $P$, parameter $\eps > 0$ and integer $\ell \ge 0$, 
if $C = \coreset(k,\eps,\ell,P)$ is a level $\ell$-coreset of $P$, 
then $C = \coreset(k, \eps', P)$ where $\eps' = \left(1+\eps \right)^{\ell} - 1$.
\end{lemma}
\begin{proof}
	We prove this by induction, denote our lemma as proposition $\mathcal{P}$. 
	Consider the base case $\ell = 0$, by definition,
	the level $0$ coreset is in the base buckets where original input points inside, 
	and $\left(1+\epsilon \right)^{0} - 1 = 0$.
	
	Now consider level $\ell > 0$. Suppose that $\mathcal{P}(L - 1)$ is true, 
	the task is to prove ${\cal P}(L)$.
	Suppose $C$ is a level $L$ coreset, then $C$ is merged by $r$ level $(L - 1)$ coresets. 
	For $i=1 \ldots r$, let $C_i$ denote the level $(L - 1)$ coreset, 
	each summarizing a point set $P_i$. 
	Then $C$ must be of the form $C = \coreset(k, \epsilon, \cup_{i=1}^r C_i)$.

	By the inductive hypothesis, we know that $C_i = \coreset(k, \epsilon_i, P_i)$
	where $\epsilon_i = \left(1 + \epsilon \right)^{L-1} - 1$.  
	Let $C' = \cup_{i=1}^r C_i$. From Observation~\ref{obs:coreset1} and using
	$P = \cup_{i=1}^r P_i$, it must be true that $C'= \coreset(k, \epsilon_i, P_i)$.
	Since $C = \coreset(k, \epsilon, C')$ and using Observation~\ref{obs:coreset2}, we
	get $C=\coreset(k, \gamma, P)$ where
	$\gamma = (1 + \epsilon)(1 + \epsilon_i) - 1$. Simplifying, we get
	$\gamma = (1 + \epsilon)(1 + \left(1 + \epsilon \right)^{(L-1)} - 1)-1 =
	\left(1 + \epsilon \right )^L - 1$.
	This proves the inductive case, which completes the proof.
\end{proof}

\begin{fact}
\label{fact:cstree-fact}
After observing $N$ base buckets, the number of levels in the coreset tree \ct satisfies
$\ell = \max \{ j \mid Q_j \neq \emptyset\}$, is $\ell \leq \log_r N$.
\end{fact}

\remove{
We first determine the number of levels in the coreset tree after observing $N$ base buckets.  Let the maximum level of the tree be denoted by $\ell(Q) = \max\{j \,|\, Q_j \neq \emptyset \}$.
\begin{lemma}
\label{lemma:cstree-level}
After observing $N$ base buckets, $\ell(Q) \le \log_r N$.
\end{lemma}
\begin{proof}
  As was pointed out earlier, for each level $j \geq 0$, a bucket in $Q_j$ is a
  summary of $r^j$ base buckets.  Let $\ell^* = \ell(Q)$. After observing $N$
  base buckets, the coverage of a bucket at level $\ell^*$ cannot exceed $N$, so
  $r^{\ell^*} \leq N$, which means $\ell(Q) = \ell^* = \log_r N$.
\end{proof}
}

The accuracy of a coreset is given by the following lemma, since it is
clear that a level-$\ell$ bucket is a level-$\ell$ coreset of
its responsible range of base buckets.

\begin{lemma}
\label{lemma:cstree-accuracy}
Let $\eps=(c \log r) /\log N$ where $c$ is a small enough constant. 
After observing $N$ base buckets from the stream, a clustering query \clusterquery{} returns a set of $k$ centers $\Psi$ of $S$ whose clustering
cost is a $O(\log k)$-approximation to the optimal clustering for $S$.
\end{lemma}
\begin{proof}
	After observing $N$ base buckets, Fact~\ref{fact:cstree-fact} indicates
	that all coresets in the coreset tree are at level no greater than $\log_r N$. Using Lemma~\ref{lemma:cstree-level2}, the maximum level coreset is an $\epsilon'$-coreset where
	\[
	\epsilon' = 
	\left[ \left(1 + \frac{c \log r}{\log N}\right)^{\frac{\log N}{\log r}} - 1 \right] 
	\le 
	\left[ e^{\left(\frac{{c}\log r}{\log N}\right) \cdot \frac{\log N}{\log r}} - 1 \right] 
	< 0.1
	\]
	
	Consider that \clusterquery computes \kmpp on the union of two sets, 
	one of the result is \ctcoreset and the other is the partially-filled base bucket
	$\object.C$.  Hence, $\Theta = (\cup_{j} \cup_{B \in Q_j} B) \cup \object.C$
	is the coreset union that is given to \kmpp.  
	Using Observation~\ref{obs:coreset1}, the union set $\Theta$ is a $\epsilon'$-coreset of $S$. 
	Let $\Psi$ be the final $k$ centers generated by running $\kmpp$ on $\Theta$, 
	and let $\Psi_{OPT}$ be the set of $k$ centers which achieves optimal $\km$ clustering
	cost for $S$. From the definition of coreset, when $\epsilon'<0.1$, we have
	\begin{equation}
	\label{eq:accuracy eq1}
	0.9\phi_{\Psi}(S)  \le \phi_{\Psi}(\Theta) \le 1.1\phi_{\Psi}(S)
	\end{equation}
	\begin{equation}
	\label{eq:accuracy eq2}
	0.9\phi_{\Psi_{OPT}}(S)  \le \phi_{\Psi_{OPT}}(\Theta) \le 1.1\phi_{\Psi_{OPT}}(S)
	\end{equation}
	
	Let $\Psi_1$ denote the set of $k$ centers which achieves optimal $\km$
	clustering cost for the union coreset set $\Theta$.  
	Using Theorem~\ref{theo:kmeans++}, we have
	\begin{equation}
	\label{eq:accuracy eq3}
	\expct{ \phi_{\Psi}(\Theta)} \le 8(\ln k+2) \cdot \phi_{\Psi_1}(\Theta)
	\end{equation}
	
	Since $\Psi_1$ is the optimal $k$ centers for $\Theta$, we have 
	\begin{equation}
	\label{eq:accuracy eq4}
	\phi_{\Psi_1}(\Theta) \le \phi_{\Psi_{OPT}}(\Theta)
	\end{equation}
	
	Using Equations~\ref{eq:accuracy eq2}, \ref{eq:accuracy eq3} and \ref{eq:accuracy eq4} we get
	\begin{equation}
	\label{eq:accuracy eq5}
	\expct{\phi_{\Psi}(\Theta)} \le 9(\ln k+2)\cdot\phi_{\Psi_{OPT}}(S) 
	\end{equation}
	
	Using Equations~\ref{eq:accuracy eq1} and \ref{eq:accuracy eq5},
	\begin{equation}
	\label{eq:accuracy eq6}
	\expct{\phi_{\Psi}(S)} \le 10(\ln k+2)\cdot\phi_{\Psi_{OPT}}(S) 
	\end{equation}
	We conclude that $\Psi$ is a factor $O(\log k)$ clustering centers of $S$
	compared to the optimal.
\end{proof}

The following lemma quantifies the memory and time cost of \ct.
\begin{lemma}
\label{thm:cstree-time}
Let $N$ be the number of base buckets observed so far. 
Algorithm $\ct$, including the driver, takes amortized $O(dm)$ time per point, 
using $O\left(dm \cdot \frac{r \log N}{\log r} \right)$ memory. 
The amortized cost of answering a query is 
$O\left( \frac{k d m}{q} \cdot \frac{r \log N}{\log r} \right)$ per point.
\end{lemma}
\begin{proof}
  First, the cost of arranging $n$ points into level-$0$ buckets is trivially
  $O(n)$, resulting in $N = n/m$ buckets.  For $j \geq 1$, a level-$j$ bucket is
  created for every $r^j$ buckets, so the number of level-$j$ buckets ever
  created is $N/r^j$.  Hence, across all levels, the total number of buckets
  created is $\sum_{j=1}^\ell \tfrac{N}{r^j} = O(N/r)$. Furthermore, when a
  bucket is created, \ct merges $rm$ points into $m$ points.  By
  Theorem~\ref{theo:coreset-build}, the total cost of creating these buckets is
  $O(\tfrac{N}r \cdot d m^2 r) = O(dnm)$, hence $O(dm)$ amortized time per
  point. 
  In terms of space, each level must have fewer than $r$ buckets, each
  with $m$ points. Therefore, across $\ell \leq \log_r N$ levels,
  the space required is $O(dm \cdot \frac{r \log N}{\log r})$.  
  Finally, when answering a query, the union of all the buckets has at most
  $O(m \cdot \frac{r \log N}{\log r})$ points, computable in the same time as
  the size.  Therefore, \kmpp run on these points plus one base bucket,
  takes $O( \frac{k d m}{q} \cdot \frac{r \log N}{\log r})$.  The amortized bound immediately
  follows.  This proves the theorem.
\end{proof}

As evident from the above lemma, answering a query using \ct is expensive 
compared to the cost of adding a point.  More precisely, when queries
are made rather frequently---every $q$ points,
$q < O(k \cdot \frac{r \log N}{\log r})$---the cost of query
processing is asymptotically greater than the cost of handling point arrivals.
We address this issue in the next section.

\section{Clustering Algorithms with Fast Queries}
\label{sec:algorithm}

\newcommand{\leftep}{{\tt left}}
\newcommand{\rightep}{{\tt right}}
\newcommand{\level}{{\tt level}}
\newcommand{\cache}{\texttt{cache}\xspace}
\newcommand{\bspan}{{\tt span}}
\newcommand{\lookup}{{\tt lookup}}
\newcommand{\bunion}{{\tt union}}
\newcommand{\bweight}{{\tt weight}}
\newcommand{\epoch}{{\tt epoch}}



This section describes algorithms for streaming clustering with an emphasis on
query time. 

\newcommand*{\keyset}{\texttt{keySet}\xspace}
\subsection{Algorithm $\cctree$: Coreset Tree with  Caching}
\label{sec:cctree}
The $\cctree$ algorithm uses the idea of ``coreset caching'' to speed up
query processing by reusing coresets that were constructed during prior
queries. In this way, it can avoid merging a large number of coresets at query
time. When compared with $\cstree$, the $\cctree$ algorithm 
is with the same update process ($\ctupdate$), but apply caching
coreset during the query.

\begin{figure*}[tb]
  \includegraphics[width=0.98\textwidth]{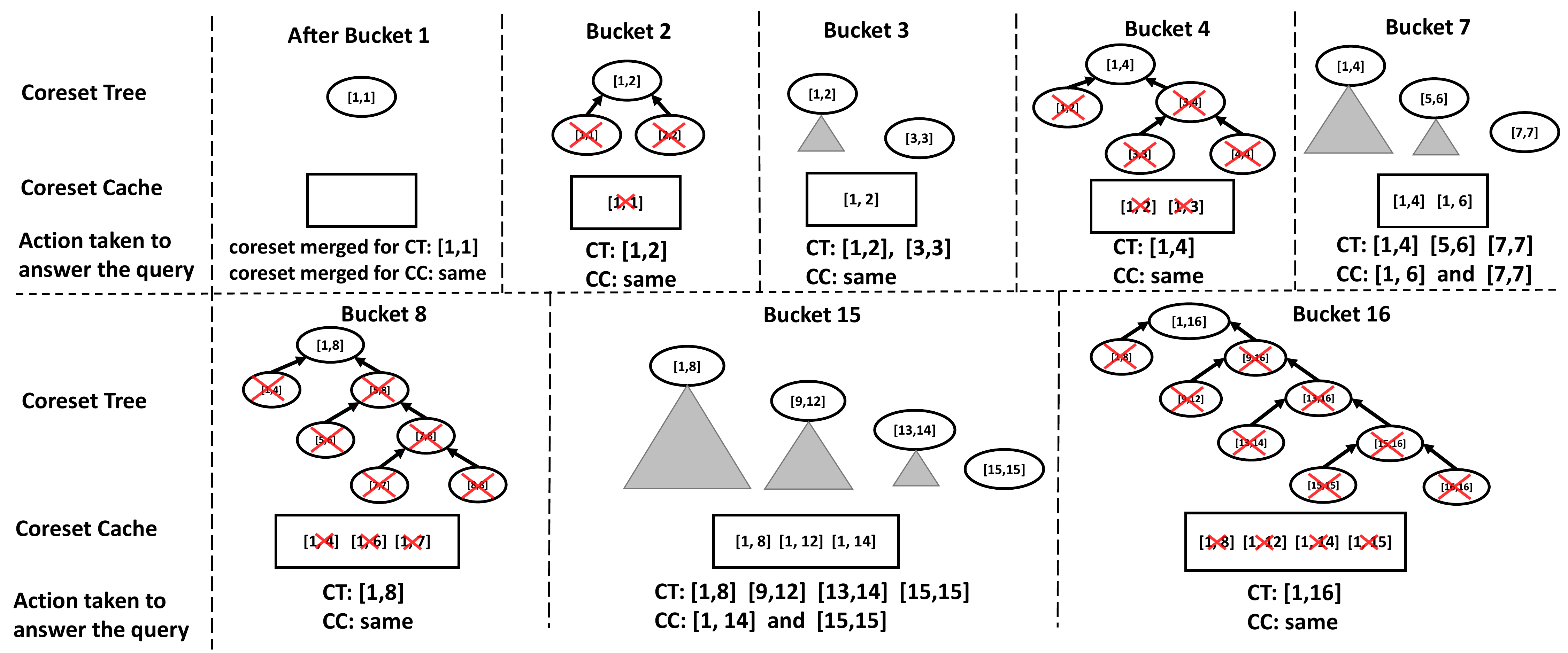}
  \caption{Illustration of Algorithm \cc, showing the states of coreset tree and cache after batch $1$, $2$, $3$, $4$, $7$, $8$, $15$ and $16$. The notation $[l,r]$ denotes a coreset of all points in buckets $l$ to $r$, both endpoints inclusive. The coreset tree consists of a set of coresets, each of which is a base bucket or has been formed by merging multiple coresets. Whenever a coreset is merged into a another coreset (in the tree) or discarded (in the cache), the coreset is marked with an ``X''. We suppose that a clustering query arrives after seeing each batch, and describe the actions taken to answer this query (1)~if only \ct was used, or (2)~if \cc was used along with \ct.}
\label{fig:algo-cc}
\end{figure*}

In addition to the coreset tree $\cstree$, the $\cctree$ algorithm
also has an additional \emph{coreset cache} denoted by $\cache$, that stores a
subset of coresets that were previously computed. When a new query has
to be answered, $\cctree$ avoids the cost of merging coresets from
multiple levels in the coreset tree. Instead, it reuses previously
cached coresets and retrieves a small number of additional coresets
which are the same level of the coreset tree, thus leading to less 
computation at query time.

However, the level of the resulting coreset increases linearly with
the number of merges a coreset is involved in. For instance, suppose 
we recursively merge the current coreset with the next arriving base bucket
of coreset to get a new coreset, and so on, for $N$ batches. The resulting 
coreset will have a level of $\Theta(N)$, which can lead to a poor clustering
accuracy. Additional care is needed to ensure that the level of a coreset
is controlled while caching is used.

\noindent\textbf{Details:} Each cached coreset is a summary of base
buckets $1$ through some number $u$. We call this number $u$ as the
\emph{right endpoint} of the coreset and use it as the key/index into
the cache. We call the interval $[1,u]$ as the ``span'' of the
bucket. To explain which coresets can be reused by the algorithm, 
we introduce the following definitions.

For integers $n > 0$ and $r > 0$, consider the unique decomposition of $n$
according to powers of $r$ as $n = \sum_{i=0}^j \beta_i r^{\alpha_i}$, where
$0 \leq \alpha_0 < \alpha_1 \ldots < \alpha_j$ and $0 < \beta_i < r$ for each
$i$. The $\beta_i$s can be viewed as the non-zero digits in the representation
of $n$ as a number in base $r$. Let $\minor(n,r) = \beta_0 r^{\alpha_0}$, the
smallest term in the decomposition, and $\major(n,r) = n - \minor(n,r)$. Note
that when $n$ is in the form of single term $\beta r^{\alpha}$ 
where $0 < \beta < r$ and $\alpha \geq 0$, $\major(n) = 0$.

For $\kappa =1 \ldots j$, let $n_\kappa = \sum_{i=\kappa}^{j} \beta_i r^{\alpha_i}$. 
$n_\kappa$ can be viewed as the number obtained by dropping the $\kappa$ smallest 
non-zero digits in the representation of $n$ as a number in base $r$. 
The set $\prefixsum(n, r)$ is defined as $\{n_{\kappa} \mid \kappa = 1 \ldots j \}$. 
When $n$ is of the form $\beta r^{\alpha}$ where $0 < \beta < r$, $\prefixsum(n, r)=\emptyset$. 

For instance, suppose $n=47$ and $r=3$. Since $47 = 1\cdot 3^3 + 2 \cdot 3^2 + 2 \cdot 3^0$, 
we have $\minor(47,3) = 2, \major(47,3) = 45$, and $\prefixsum(47,3) = \{27, 45\}$.

\cc caches every coreset whose right endpoint is in $\prefixsum(N,r)$.
When a query arrives when $N$ buckets received, the task is to compute a coreset
whose span is $[1,N]$. \cc partitions $[1,N]$ as $[1,N_1] \cup [N_1+1,N]$ 
where $N_1 = \major(N,r)$. 
Out of these two intervals, suppose the query comes after every new base bucket received, 
this guarantees that $[1,N_1]$ should be available in the cache. 
$[N_1+1,N]$ is retrieved from the coreset tree, through the union of no more than $(r-1)$ coresets. 
This needs a merge of no more than $r$ coresets. This is in contrast with \ct, 
which may need to merge as many as $(r-1)$ coresets at each level of the tree, 
resulting in a merge of up to $(r-1) \cdot \frac{\log N}{\log r}$ coresets for all levels at query time.

The algorithm for maintaining the cache and answering clustering queries 
is shown in Algorithm~\ref{algo:cctree-functions}. 
The caching process works along with the query process ($\cccoreset$), in a way that
making our algorithm be flexible with the queries by users. When the queries are frequent, 
our algorithm utilizes the cache to provide a faster query speed and a guarantee on the accuracy of 
clustering result. Otherwise in case of the queries are infrequent, we will show that the time complexity 
of updating the cache is at the same level of the query process without caching (algorithm \ct). 
This caching design helps the clustering system to adapt in the faces of both burst queries and occasional queries. 
Figure~\ref{fig:algo-cc} shows an example of how the \cc algorithms updates the cache and answers queries
using cached coresets.

Note that to keep the size of the cache small, as new base buckets arrive,
\ccupdate{} will ensure that ``stale'' or unnecessary coresets are removed. 
The following fact relates to what the cache should store.

\begin{algorithm}
  \caption{Coreset Tree with Caching}
  \label{algo:cctree-functions}
  \label{algo:cctree-init}
  \label{algo:cctree-update}
  \label{algo:cctree-coreset}
\fnDef{$\ccinit(r, k, \eps)$}{
  \tcp{The coreset tree}
  $Q \gets \ctinit{}(r, k, \eps)$ \;
  $\cache \gets \emptyset$
}
\fnDef{$\ccupdate(b,N)$}{
\tcp{$b$ is a new bucket and $N$ is the number of buckets received so far.}
  $Q.\ctupdate(b, N)$\;
}
\fnDef{$\cccoreset()$}{
\tcp{Return a coreset of points in buckets $1$ till $N$}
	\uIf{$N$ exists in \cache}
	{
		\Return coreset for buckets $[1, N]$ from the \cache \;
	}
   	$N_1 \gets \major(N,r)$ and $N_2 \gets \minor(N,r)$\;
	Let $N_2=\beta r^{\alpha}$ where $\alpha$ and $\beta < r$ are positive integers\;
	\tcp{coreset for buckets spanning $[1, N_1]$ is not in the \cache}
	\eIf{$N_1$ does not exist in \cache}
	{$U \gets \ctcoreset()$ \;}
	{
		\tcp{$A$ is the coreset for buckets $N_1+1, N_1+2, \ldots, (N_1+N_2)=N$ and is retrieved from the coreset tree}   
   		$a \gets \cup_{B \in Q_\alpha} B$\;
   		\tcp{$b$ is the coreset for buckets spanning $[1, N_1]$, retrieved from the cache} 
   		$b \gets \cache.\lookup(N_1)$ \;
   		$U \gets a \cup b$ \;
   		
	}	
   	\tcp{Store coreset into \cache}
   		$C \gets \coreset(k, \epsilon, U)$\; 
    	Add coreset $C$ to $\cache$ using key $N$\; 
    	Remove each bucket from $\cache$ whose key does not appear in $\prefixsum(N) \cup \{N\}$  \;
   	\Return {$C$} 
}
\end{algorithm}

\begin{fact}
  \label{fact:prefixextension}
  Let $r \geq 2$. For each $N \in \Z^+$,
  $\prefixsum(N+1, r) \subseteq \prefixsum(N,r) \cup \{ N \}$.
\end{fact}

Since $\major(N,r) \in \prefixsum(N,r)$ for each $N$, 
if the query comes after each new base bucket received,
we can always retrieve the bucket with span $[1,\major(N,r)]$ from $\cache$.

\begin{lemma}
\label{lemma:cache correctness}
Suppose query comes after receiving each new base bucket. 
Immediately before base bucket $N$ arrives, each $y \in \prefixsum(N,r)$ appears in the key set of \cache.
\end{lemma}
\begin{proof}
Proof is by induction on $N$. 
The base case $N=1$ is trivially true, since $\prefixsum(1,r)$ is empty set. 
For the inductive step, assume that before bucket $N$ arrives, 
each $y \in \prefixsum(N, r)$ appears in $\cache$. 
During the query after receiving bucket $N$, 
we store the coreset whose span is $[1, N]$ to the cache.
By Fact~\ref{fact:prefixextension}, we know that
$\prefixsum(N+1, r) \subseteq \prefixsum(N,r) \cup \{N\}$. 
Using this, every bucket with a right endpoint in $\prefixsum(N+1,r)$ 
is present in $\cache$ at the beginning of bucket $(N+1)$ arrives. 
Hence, the inductive step is proved.
\end{proof}

When the queries come less frequent, the cache is less frequently updated as well. 
Then it can not guarantee that the \major is in the cache, that is $N_1$ may not 
exist in the \cache. In this case, the $\cccoreset$ will switch back to the $\ctcoreset$ method in \ct. 
We analyze the time complexity of algorithm $\cc$ under the assumption that we can always use cache 
to accelerate the current query. In practice, we run experiments to show the result of 
algorithm performance when less frequent queries.

\begin{lemma}
\label{lemma:cctree-level}
When queried after inserting base bucket $N$, Algorithm~\cccoreset returns a coreset 
whose level in no more than $\left\lceil 2 \log_r N \right\rceil - 1$.
\end{lemma}
\begin{proof}
Let $\chi(N)$ denote the number of non-zero digits in the representation of $N$ as a number in base $r$. 
We show that the level of the coreset returned by Algorithm~\cccoreset is no more than 
$\left\lceil \log_r N \right\rceil + \chi(N)-1$. 
Since $\chi(N) \le \left\lceil \log_r N \right\rceil$, the lemma follows.

The proof is by induction on $\chi(N)$. If $\chi(N)=1$, then $\major(N,r)=0$,
and the coreset is retrieved directly from the coreset tree $Q$. 
By Fact~\ref{fact:cstree-fact}, each coreset in $Q$ is at a level 
no more than $\lceil \log_r N \rceil$, and the base case follows. 
Suppose the claim was true for all $N$ such that $\chi(N) = t$. Consider $N$ such that
$\chi(N)=(t+1)$. The algorithm computes $N_1 = \major (N,r)$, and retrieves the
coreset with span $[1,N_1]$ from the cache. Note that $\chi(N_1) = t$. By the
inductive hypothesis, the coreset for span $[1,N_1]$ is at a level
$\left\lceil \log_r N \right\rceil + t - 1$. The coresets for span
$[N_1+1,N]$ are retrieved from the coreset tree; note there are multiple such
coresets, but each of them is at a level no more than
$\left\lceil \log_r N \right\rceil$, using Fact~\ref{fact:cstree-fact}.  
The level of the union coreset for span $[1,N]$ is no more than
$\left\lceil \log_r N \right\rceil + t$, proving the inductive case.
\end{proof}

With the coreset level bounded, we can give the guarantee on the accuracy of clustering centers.
Let the accuracy parameter $\eps=\frac{c \log r}{2\log N}$, where $c < \ln{1.1}$.
\begin{lemma}
\label{lemma:cctree-accuracy}
After observing $N$ buckets from the stream, when using clustering data structure \cc, 
Algorithm~\clusterquery returns a set of $k$ points whose clustering cost 
is within a factor of $O(\log k)$ of the optimal $k$-means clustering cost.
\end{lemma}
\begin{proof}
The proof is similar as Lemma~\ref{lemma:cstree-accuracy}. From Lemma~\ref{lemma:cctree-level}, 
we know that the level of a coreset returned is no more than $\left\lceil 2 \log_r N \right\rceil - 1$. 
Using Lemma~\ref{lemma:cstree-level2}, the returned coreset, say $C$, is an $\eps'$-coreset where 
$\eps' = \left[ \left(1 + \frac{{c}\log r}{2\log N}\right)^{\frac{2\log N}{\log r}} - 1 \right] \le \left[ e^{\left(\frac{{c}\log r}{2\log N}\right) \cdot \frac{2\log N}{\log r}}-1 \right] < 0.1$.
Following an argument similar to that of Lemma~\ref{lemma:cstree-accuracy}, we arrive at the result.
\end{proof}

The following lemma shows the time and space complexity of the \cache. 
We show that the time on updating the \cache is at least at the same level of 
the time of answering a query in Algorithm~$\ct$, and can be better to be in linear
scale of $r$ instead of $\frac{r \log N}{\log r}$.
\begin{lemma}
\label{lemma:cctree-time}
Algorithm~\ref{algo:cctree-update} processes a stream of points 
using amortized time $O(dm)$ per point, 
using memory of $O\left(dm \cdot \frac{r \log N}{\log r} \right)$. 
The amortized cost of answering a query is $O\left(\frac{kdm}{q} \cdot r \right)$.
\end{lemma}

\begin{proof}
The runtime for Algorithm~\ccupdate~ is same as the Algorithm~\ctupdate. 
The update time for \ccupdate is $O(dm)$ per point.

From Lemma~\ref{lemma:cache correctness}, $N_1$ is always in the \cache. 
Algorithm~\cccoreset combines no more than $r$ buckets, 
out of which there is no more than one bucket from the cache, 
and no more than $(r-1)$ from the coreset tree. 
From Theorem~\ref{theo:coreset-build}, the time to compute coreset 
on $O(mr)$ points is $O(d m^2 r)$, and coreset size $m$ is $O(k)$.
The time to compute coreset $C$ is $O(kdmr)$.   
To compute the $k$ centers from $U$, 
it is necessary to run $\kmpp$ on $O(mr)$ points using time $O(kdmr)$. 
The amortized query time per point is $O\left( \frac{kdmr}{q} \right)$.

\remove {
In the worst case, when the major part is always not in the cache. 
Comparing to Algorithm~\ct, the additional time cost is updating the cache. 
The number of points in set $U$ is at most $(m \cdot \frac{r \log N}{\log r})$, 
computation time is $O(d m^2 r \cdot \log_r N)$. 
As the size of coreset is $O(k)$, 
the addition time for constructing coreset is $O(kdm \cdot \frac{r \log N}{\log r})$ 
at the same level of the query time by \ct.
}

The coreset tree $Q$ uses space $O\left( dm \cdot \frac{r \log N}{\log r} \right)$. 
After processing bucket $N$, $\cache$ only stores those buckets that are 
corresponding to $\prefixsum(N,r) \cup \{ N \}$. 
The number of such buckets possible is $O(\log_r N )$, 
so the space cost of $\cache$ is $O(dm \cdot \frac{\log N}{\log r})$. 
The space complexity follows.
\end{proof}

\subsection{Algorithm $\rcc$: Recursive Coreset Cache}
\newcommand{\rmax}{\rho}
\newcommand{\order}{\texttt{order}}
\newcommand{\mainds}{\mathcal{R}}

There are still a few issues with the $\cc$ data structure. First, the level of
the coreset finally generated is $O(\log_rN)$. Since theoretical guarantees on
the approximation quality of clustering worsen with the number of levels of the
coreset, it is natural to ask if the level can be further reduced to $O(1)$.
Moreover, the time taken to process a query is linearly proportional to $r$; we
wish to reduce the query time even more.  While it is natural to aim to
simultaneously reduce the level of the coreset as well as the query time, at
first glance, these two goals seem to be inversely related. It seems that if we
decreased the level of a coreset (better accuracy), then we will have to
increase the merge degree, which would in turn increase the query time. For
example, if we set $r=\sqrt{N}$, then the level of the resulting coreset is
$O(1)$, but the query time will be $O(\sqrt{N})$.

In the following, we present a solution $\rcc$ that uses the idea of coreset
caching in a recursive manner to achieve both a low level of the coreset, as
well as a small query time. In our approach, we keep the merge degree $r$ 
in a relatively high value, thus keeping the levels of coresets low. 
At the same time, we use coreset caching even within a single level of a coreset tree, 
so that there is no need to merge $r$ coresets at query time. 
Special care is required for coreset caching in this case, 
so that the level of the coreset does not increase significantly.

For instance, suppose we built another coreset tree with merge degree $2$ for
the $O(r)$ coresets within a single level of the current coreset tree, this
would lead to a inner tree with level of $\log r$. At query time, we aggregate
$O(\log r)$ coresets from this inner tree, in addition to a coreset from the \cc. So, this will lead to a level of $O\left(\max\left\{\frac{\log N}{\log r}, \log r\right \}\right)$ and a query time proportional to $O(\log r)$. This is an improvement from the coreset cache, which has a query time proportional to $r$ and a level of $O\left(\frac{\log N}{\log r}\right)$.

We can take this idea further by recursively applying the same idea to the
$O(r)$ buckets within a single level of the coreset tree. Instead of having a
coreset tree with merge degree $2$, we use a tree with a higher merge degree,
and then have a coreset cache for this tree to reduce the query time, and apply
this recursively within each tree. This way we can approach the ideal of a small
level and a small query time. We are able to achieve interesting tradeoffs, as
shown in Table~\ref{table:rcc}. To keep the level of the resulting coreset low,
along with the coreset cache for each level, we also maintain a list of coresets
at each level, like in the $\ct$ algorithm. To merge coresets to a higher level,
we use the list, rather than the recursive coreset cache.

More specifically, the $\rcc$ data structure is defined inductively as
follows. For integer $i \ge 0$, the $\rcc$ data structure of order $i$ is
denoted by $\rcc(i)$. $\rcc(0)$ is a $\cctree$ data structure with a merge
degree of $r_0 = 2$. For $i>0$, $\rcc(i)$ consists of:
\begin{itemize}
\item $\cache(i)$, a coreset cache storing previous coresets.

\item For each level $\ell = 0, 1, 2, \ldots$, there are two structures. 
One is a list of buckets $L_{\ell}$, 
similar to the structure $Q_{\ell}$ in a coreset tree. 
The maximum length of a list is $r_{i} = 2^{2^i}$. 
Another is an $\rcc_{\ell}$ structure which is a $\rcc$ 
structure of a lower order $(i-1)$, which stores the same information as list $L_\ell$, 
except in a way that can be quickly retrieved during a query.
\end{itemize}

The main data structure $\mainds$ is initialized as $\mainds = \rccinit(\iota)$, 
for a parameter $\iota$, to be chosen. 
Note that $\iota$ is the highest order of the recursive structure. 
This is also called the ``nesting depth" of the structure. 

\begin{algorithm}[t]
\DontPrintSemicolon
\caption{$\mainds.\rccinit(\iota)$
\label{algo:rcc-init}
}
$\mainds.\order \gets \iota$, 
$\mainds.\cache \gets \emptyset$, 
$\mainds.r \gets 2^{2^{\iota}}$\;

\tcp{$N$ is the number of buckets so far}
$\mainds.N \gets 0$\; 
\ForEach{$\ell=0,1,2, \ldots$}
{
   $\mainds.L_{\ell} \gets \emptyset$\;
   \uIf{$\mainds.\order > 0$}
   {$\mainds.\rcc_{\ell} \gets \mainds.\rccinit(\mainds.\order - 1)$}
}
\Return $\mainds$\;
\end{algorithm}

\begin{algorithm}[t]
\DontPrintSemicolon
\caption{$\mainds.\rccupdate(b)$
\label{algo:rcc-update}}
\tcp{$b$ is a new base bucket}
$\mainds.N \gets \mainds.N + 1$\;
\tcp{Insert $b$ into $\mainds.L_0$ and merge if needed}
Append $b$ to $\mainds.L_0$.\;
\If{$\mainds.\order > 0$} 
{
	recursively update $\mainds.\rcc_0$ by $\mainds.\rcc_0.\rccupdate(b)$ \;
}
\;
\tcp{Clear $\mainds.L$ and $\rcc$ if number of buckets reaches $r$}
$\ell \gets 0$\;
\While{$(|\mainds.L_\ell| = \mainds.r)$}
{
  $b' \gets \coreset(k, \eps, \cup_{B \in \mainds.L_\ell} B)$ \;
  Append $b'$ to $\mainds.L_{\ell+1}$ \;
  \If{$\mainds.\order > 0$}
  {recursively update $\mainds.\rcc_{\ell+1}$ by $\mainds.\rcc_{\ell+1}.\rccupdate(b)$} \;
  \tcp{Empty the list of coresets $\mainds.L$}
  $\mainds.L_{\ell} \gets \emptyset$ \;
  \tcp{Empty the \rcc structure}
  \If{$\mainds.\order > 0$} 
  {$\mainds.\rcc_{\ell} \gets \rccinit(\mainds.\order-1)$} \;
  $\ell \gets \ell + 1$ \;
}
\end{algorithm}

\begin{algorithm}[t]
\DontPrintSemicolon
\caption{
$\mainds.\rcccoreset()$
\label{algo:rcc-getbuckets}
\label{algo:rcc-coreset}
}
$U \gets \emptyset$ \;
$N_1 \gets \major(\mainds.N, \mainds.r)$ \;
\eIf{$N_1$ does not exist in $\mainds.\cache$} 
{
	$U \gets \cup_{\ell} \{ \mainds.\rcc_{\ell}.\rcccoreset() \}$ \;
}
{
	$b_1 \gets$ retrieve coreset with endpoint $N_1$ from $\mainds.\cache$\;
	Let $\ell^*$ be the lowest numbered non-empty level among $\mainds.L_i, i \ge 0$.\;
	\tcp{Apply $\rcc$ data structure to retrieve the coresets from level $\ell^*$}
	\eIf{$\mainds.\order > 0$}
	{$b_2 \gets \mainds.\rcc_{\ell^*}.\rcccoreset()$}
	{$b_2 \gets \mainds.L_{\ell^*}$}
	$U \gets b_1 \cup b_2$ 
}

\tcp{Store coreset into \cache}
	$b' \gets \coreset(k, \epsilon, U)$\;
	Add $b'$ to $\mainds.\cache$ with right endpoint $\mainds.N$\;
	From $\mainds.\cache$, remove all buckets $b''$ such that 
	$\rightep(b'') \notin \prefixsum(\mainds.N) \cup \{ N \}$ \;
\Return $b'$
\end{algorithm}

\begin{table}[ht]
{
\footnotesize
\setlength{\tabcolsep}{2pt}
\begin{tabular}{ c c c c c}
\toprule
$\iota$ & coreset level & Query cost  & update cost & Memory \\
        & at query      & (per point) & per point   & \\
\midrule
$\log \log N - 3$ & $O(1)$ & $O\left( \frac{kdm}{q} \log \log N \right)$ & $O(dm \log \log N)$ & $O\left( dmN^{1/8} \right)$ \\
\midrule
$\log \log N / 2$ & $O(\sqrt{\log N})$ & $O \left( \frac{kdm}{q} \log \log N \right)$ & $O(dm \log \log N)$ & $O \left( dm 2^{\sqrt{\log N}} \right)$ \\
\bottomrule
\end{tabular}
\smallskip
\caption{Possible tradeoffs for the $\rcc(\iota)$ algorithm, based on the parameter $\iota$, the nesting depth of the structure.
\label{table:rcc}}
}
\end{table}

\begin{lemma}
\label{lemma:rcc-level}
When queried after inserting $N$ buckets, 
Algorithm~\ref{algo:rcc-coreset} using $\rcc(\iota)$ returns a coreset 
whose level is $O\left(\frac{\log N}{2^{\iota}}\right)$. 
The amortized time cost of answering a clustering query is 
$O\left(\frac{kdm}{q} \cdot \log \log N \right)$ per point. 
\end{lemma}
\begin{proof}
Algorithm~\ref{algo:rcc-coreset} retrieves a few coresets from $\rcc$ of different orders. 
From the outermost structure $\rcc(\iota)$, 
it retrieves one coreset $c$ from $\cache(\iota)$. 
Using an analysis similar to Lemma~\ref{lemma:cctree-level}, 
the level of $b_{1}$ is no more than $\frac{2 \log N}{\log r_{\iota}}$. 

Note that for $i < \iota$, the maximum number of coresets that will be inserted
into $\rcc(i)$ is $r_{i+1} = r_i^2$. The reason is that inserting $r_{i+1}$
buckets into $\rcc(i)$ will lead to the corresponding list structure for
$\rcc(i)$ to become full. At this point, the list and the $\rcc(i)$ structure
will be emptied out in Algorithm~\ref{algo:rcc-update}. From each recursive call
to $\rcc(i)$, it can be similarly seen that the level of a coreset retrieved
from the cache is at level $\frac{2 \log{r_i}}{\log r_{i-1}}$, which is
$O(1)$. The algorithm returns a coreset formed by the union of all the coresets,
followed by a further merge step. Thus, the coreset level is one more than the
maximum of the levels of all the coresets returned, which is
$O\left( \frac{\log N}{\log r_{\iota}} \right)$.

For the query cost, similar to our analysis in $\cc$, we assume that 
for each order of $\rcc(i)$, we can always use the cache in coreset queries.
Comparing to Algorithm~\cccoreset, the minor part of coreset is retrieved from the inner 
\rcc data structure with lower order. Thus, for each order of $\rcc$, the number 
of coresets merged is $2$. The number of coresets merged at query time 
is equal to two times the nesting depth of the structure, that is $2 \cdot \iota$. 
The query time equals the cost of running $\kmpp$ on the union of all these coresets, 
for a total time of $O(kd m \log \log N)$. The amortized per-point cost of a query follows.

\remove{
Consider the case that the coreset of major does not exist in the cache. 
Suppose all the caches at different order are not available, our algorithm 
uses \ctcoreset to retrieve coresets at each order of \rcc. 
Note that the number of levels for each $\rcc(i)$ is at most $2$ except the $\rcc(\iota)$. 
The reason is as follows. 
It is necessary to insert $r_{i+1} = r_i^2$ buckets into $\rcc(i)$. 
Since this is the maximum number of buckets that will be inserted into $\rcc(i)$, 
there are at most two levels of lists within each $\rcc(i)$ for $i < \iota$.
The number of coresets be merged for $\rcc(\iota-1)$ is $2^{\iota-1}$, which is $O(\log r_{\iota})$.
The total number of coresets be merged is 
$O\left( \frac{\log N}{\log r_{\iota}} \cdot \log r_{\iota} \right) = O(\log N)$, 
and query time is $O(kd m \log N)$.
}
\end{proof}

\begin{lemma}
\label{lemma:rcc-performance}
The memory consumed by $\rcc(\iota)$ is $O(dm r_{\iota})$. 
The amortized processing time is $O(dm \log \log N)$ per point.
\end{lemma}
\begin{proof}
First, as stated in the proof of Lemma~\ref{lemma:rcc-level}, 
in $\rcc(i)$ for $i < \iota$, there are at most two level of lists $L_{\ell}$. 


We prove by induction on $i$ that $\rcc(i)$ has no more than $6r_i$ buckets.  
For the base case, $i=0$, and we have $r_0=2$. 
In this case, $\rcc(0)$ has two levels, each with no more than $2$ buckets. 
So that the total memory is no more than $6$ buckets, due to the lists in two levels, 
and no more than $2$ buckets in the cache, for a total of $6=3r_0 < 6r_0$ buckets. 
For $i=1$,  $r_1=4$, the two lists have at most $8$ buckets and cache has no more than $2$ buckets. 
The recursive structures $\rcc(0)$ has $6$ buckets and there are two recursive structures, one for each level. 
Thus in total $\rcc(1)$ has no more than $22$ buckets, which is less than $24 = 6r_1$ buckets. 

For the inductive case, consider that $\rcc(i)$, the list at each level has no more than $r_i$ buckets. 
The recursive structures $\rcc_{\ell}$ within $\rcc(i)$ themselves have no more than $6r_{i-1}$ buckets. 
Adding the constant number of buckets within the cache, 
we get the total number of buckets within $\rcc(i)$ to be 
$2r_i + 2 \cdot 6r_{i-1} + 2 = 2r_i + 12 \sqrt{r_i} + 2 \le 6 r_i$, 
for $r_i \ge 16$, i.e. $i \ge 2$. 
Thus if $\iota$ is the nesting depth of the structure, 
the total memory consumed is $O(dm r_{\iota})$, since each bucket requires $O(dm)$ space.

For the updating process time cost, when a bucket is inserted into $\mainds = \rcc(\iota)$, 
it is added to list $L_0$ within $\mainds$. 
The cost of maintaining these lists, 
that is the cost of merging into higher level lists, 
is amortized $O(dm)$ per point, 
similar to the analysis in Lemma~\ref{lemma:cctree-time}. 
The bucket is also recursively inserted into a $\rcc(\iota-1)$ structure, 
and a further structure within, 
and the amortized time for each such structure is $O(dm)$ per point. 
The total time cost is $O(dm\iota)$ per point which is equal to $O(dm \log \log N)$.
\end{proof}

Different tradeoffs are possible by setting $\iota$ to specific values. Some examples are shown in the Table~\ref{table:rcc}.

\newcommand{\fallbackcost}{\texttt{$\phi_{prev}$}\xspace}
\newcommand{\estcost}{\texttt{$\phi_{now}$}\xspace}

\subsection{Online Coreset Cache: a Hybrid of \cc and \seqkm}
\label{sec:hybrid}

\begin{figure}[t]
  	\includegraphics[width=0.49\textwidth]{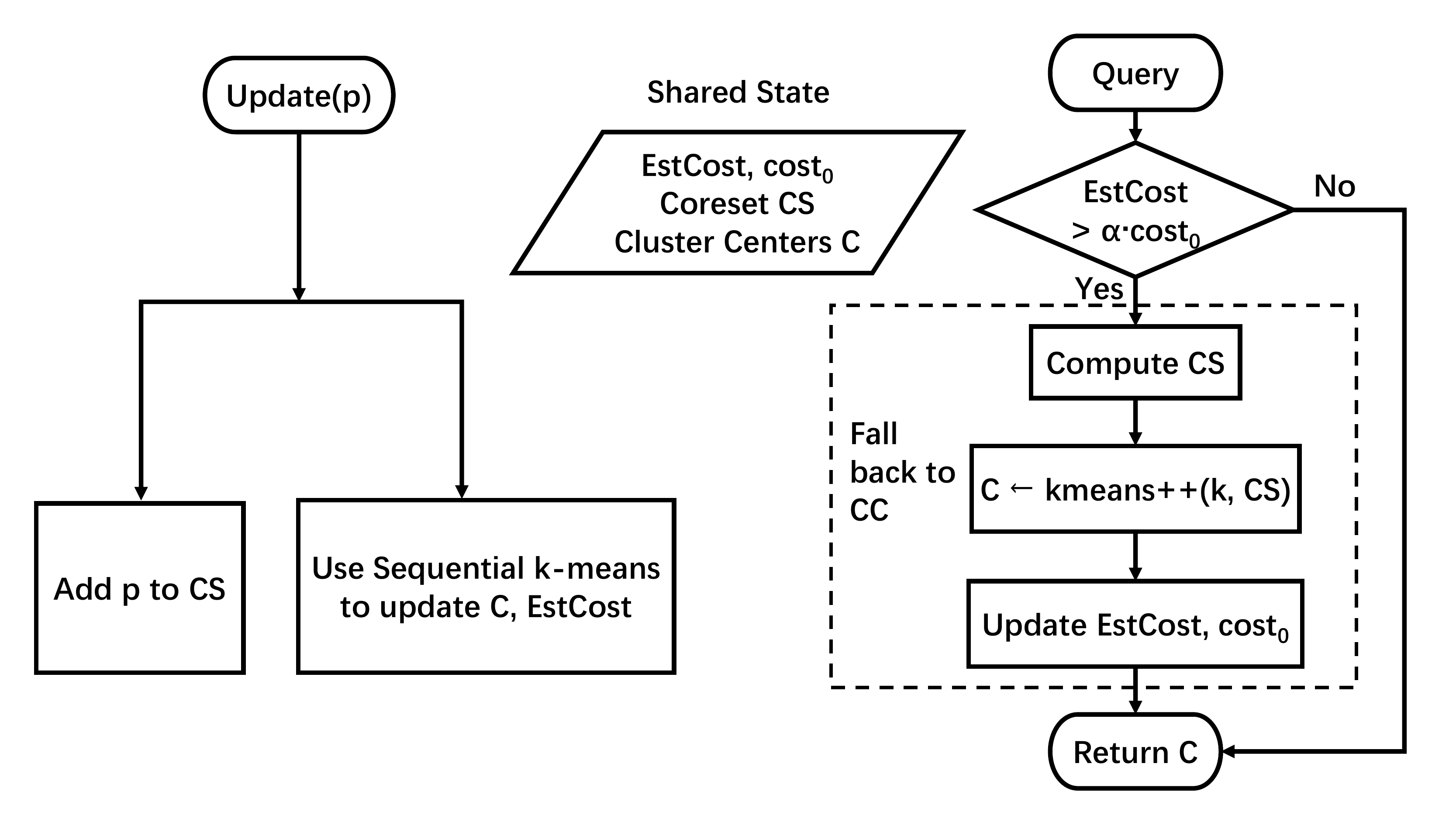}
  	\caption{Illustration of Algorithm \hybrid.}
\label{fig:algo-onlinecc}
\end{figure}

If we break down the query runtime of the algorithms considered so far, we
observe two major components: (1)~the construction of the coreset of all points
seen so far, through merging stored coresets; and (2)~the \kmpp algorithm
applied on the resulting coreset. The focus of the algorithms discussed so far (
\cc and \rcc) is on decreasing the runtime of the first component, coreset
construction, by reducing the number of coresets to be merged at query time. But
they still have to pay the cost of the second component \kmpp, which is
substantial in itself, since the runtime of \kmpp is $O(kdm)$, where $m$ is the
size of the coreset. To make further progress, we have to reduce this
component. However, the difficulty in eliminating \kmpp at query time is that
without an approximation algorithm such as \kmpp, we have no way to guarantee
that the returned clustering is an approximation to the optimal.

This section presents an algorithm, \hybrid, which only occasionally runs \kmpp
at query time, and most of the time, uses a much cheaper method that costs
$O(1)$ to compute the clustering centers. \hybrid uses a combination of \cc and
the \seqkm algorithm~\cite{MacQueen67} (aka. Online Lloyd's algorithm) to
maintain the cluster centers quickly while also providing a guarantee on the
quality of clustering. Like \seqkm, it incrementall updates the current set of
cluster centers for each arriving point. However, while \seqkm can process
incoming points (and answer queries) extremely quickly, it cannot provide any
guarantees on the quality of answers, and in some cases, the clustering quality
can be very poor when compared with say, \kmpp. To prevent such deterioration in
clustering quality, our algorithm (1)~occasionally falls back to \cc, which is
provably accurate, and (2)~runs \seqkm so long as the clustering cost does not
get much larger than the previous time \cc was used. This ensures that our
clusters always have a provable quality with respect to the optimal.

To accomplish this, \hybrid also processes incoming points using \cc, thus
maintaining coresets of substreams of data seen so far. When a query arrives, it
typically answers them in $O(1)$ time using the centers maintained using
\seqkm. If, however, the clustering cost is significantly higher (by more than a
factor of $\alpha$ for a parameter $\alpha > 1$) than the previous time that the
algorithm fell back to \cc, then the query processing again returns to \cc to
regenerate a coreset. One difficulty in implementing this idea is that
(efficiently) maintaining an estimate of the current clustering cost is not
easy, since each change in cluster centers can affect the contribution of a
number of points to the clustering cost. To reduce the cost of maintenance, our
algorithm keeps an upper bound on the clustering cost; as we show further, this
is sufficient to give a provable guarantee on the quality of clustering. Further
details on how the upper bound on the clustering cost is maintained, and how
algorithms \seqkm and \cc interact are shown in Algorithm~\ref{algo:hybrid}, 
with a schematic illustration in Figure~\ref{fig:algo-onlinecc}.

\begin{algorithm}[t]
\label{algo:hybrid}
\caption{The Online Coreset Cache: A hybrid of \cc and \seqkm algorithms}

\fnDef{$\hybridinit(k, \eps, \alpha)$}{

\tcp{$C$ is the current set of cluster centers}
Initialize $C$ by running $\kmpp$ on set $\stream_0$ consisting of the first $O(k)$ points of the stream\;

\tcp{\fallbackcost is the clustering cost during the previous ``fallback'' to \cc; 
\estcost is an estimate of the clustering cost of $C$ on the stream so far}
\fallbackcost , \estcost $\gets$ clustering cost of $C$ on $\stream_0$\;
$Q \gets \ccinit(r,k,\eps)$\;
}

\tcp{On receiving a new point $p$ from the stream}
\fnDef{$\hybridupdate(p)$}{
	Assign $p$ to the nearest center $c_p$ in $C$  \;

    $\estcost \gets \estcost + \edist{p}{c_p}^2$ \;
	
	\tcp{Compute new centroid of $c_p$ and $p$ where $w$ is the weight of $c_p$}
	$c_p'  \gets (w \cdot c_p + p) / (w + 1)$   \; 
	Assign the position of $c_p'$ to $c_p$  \;
    Add $p$ to the current bucket $b$. If $|b|=m$, then $Q.\ccupdate(b)$ \;
}

\fnDef{$\hybridquery()$}{
	\If{$\estcost > \alpha \cdot \fallbackcost$}
	{
	     $CS \gets Q.\cccoreset() \cup b$, where $b$ is the current bucket that not yet inserted into $Q$ \;

             $C \gets \kmpp(k, CS)$   \;
             $\fallbackcost \gets \phi_C(CS)$, the \km cost of coreset $CS$ on centers $C$  \;
             $\estcost \gets \fallbackcost / (1-\eps)$   \;
	}
  \Return $C$  \;
}
\end{algorithm}

\begin{lemma}
\label{lemma:ecost}
In Algorithm~\ref{algo:hybrid}, after observing point set $P$ from the stream, if $C$ is the current set of cluster centers, then $\estcost$ is an upper bound on $\phi_C(P)$.
\end{lemma}
\begin{proof}
Consider the value of $\estcost$ between every two consecutive switches to \cc. Without loss of generality, suppose there is one switch happens at time $0$, let $P_0$ denote the points observed until time $0$ (including the points received at $0$). We will do induction on the number of points received after time $0$, we denote this number as $i$. Then $P_i$ is $P_0$ union the $i$ points received after time $0$. Let $\estcost (i)$ denote the cost $\estcost$ at time $i$. 

When $i$ is $0$, we compute $C$ from the coreset $CS$, from the coreset definition
\[
\fallbackcost = \phi_{C}(CS) \geq (1-\epsilon) \cdot \phi_{C}(P_0)
\] 
where $\epsilon$ is the approximation factor of coreset $CS$. So for dataset $P_0$, 
the estimation cost $\estcost(0) = \fallbackcost / (1-\eps) \geq \phi_{C}(P_0)$.

At time $i$, denote $C_i$ as the cluster centers maintained and $\estcost_i$ as the estimation of \km cost. Assume the statement is true such that $\estcost (i) > \phi_{C_i}(P_i)$. 

Consider when a new point $p$ comes, $c_p$ is the nearest center in $C_i$ to $p$. 
We compute $c_p'$, the new position of $c_p$, let $C_{i+1}$ denote the new center set 
where $C_{i+1}=C_i \setminus \{ c_p \} \cup \{c_p' \}$. 

Based on the $\hybridupdate(p)$ in Algorithm\ref{algo:hybrid}, 
\[
\estcost (i+1) = \estcost (i) + \edist{p}{c_p}^2
\]

From the assumption of inductive step, 
\[
\estcost (i) \geq \phi_{C_i}(P_i)
\]

As $c_p'$ is the centroid of $c_p$ and $p$, we have
\[
\edist{p}{c_p} > \edist{p}{c_p'} 
\]

Because $\phi_{C_{i+1}}(P_{i+1})$ is the true cost of point set $P_{i+1}$ on centers $C_{i+1}$,
\[
\phi_{C_i}(P_i) + \edist{p}{c_p'}^2 \geq \phi_{C_{i+1}}(P_{i+1})
\] 

Adding up together, we get:
\[
\estcost (i+1) \geq \phi_{C_{i+1}}(P_{i+1})
\]
Thus the inductive step is proved.
\end{proof}

\begin{lemma}
\label{lemma:online}
When queried after observing point set $P$, the \hybrid algorithm returns a set of $k$ points $C$ whose clustering cost is within $O(\log k)$ of the optimal $k$-means clustering cost of $P$, in expectation.
\end{lemma}

\begin{proof}
  Let $\phi^*(P)$ denote the optimal \km cost for $P$. We will show that
  $\phi_C(P) = O(\log k) \cdot \phi^*(P)$. There are two cases:

  \noindent{}\underline{Case I}: When $C$ is directly retrieved from \cc, 
  Lemma~\ref{lemma:cctree-accuracy} implies that
  $\expct{\phi_{C}(P)} \leq O(\log k) \cdot \phi^* (P)$. This case is handled
  through the correctness of \cc.

  \noindent{}\underline{Case II}:
  The query algorithm does not fall back to \cc. We first note from
  Lemma~\ref{lemma:ecost} that $\phi_C(P) \le \estcost$. Since the algorithm did
  not fall back to \cc, we have $\estcost \leq \alpha \cdot \fallbackcost$. Since
  \fallbackcost was the result of applying $\cc$ to the $P_0$ which is the point set received when last recent fall back, 
  we have from Lemma~\ref{lemma:cctree-accuracy} that
  $\fallbackcost \leq O(\log k) \cdot \phi^*(P_0)$. Since $P_0 \subseteq P$, we know
  that $\phi^*(P_0) \leq \phi^*(P)$. Putting together the above four
  inequalities, we have $\phi_C(P) = O(\log k) \cdot \phi^*(P)$.  
\end{proof}

\remove{
\begin{algorithm}
\label{algo:cluster}
\DontPrintSemicolon
\caption{{\tt StreamCluster}($\mathcal{D}, m$)}
\tcc{Framework for stream clustering}
\tcc{$\mathcal{D}$ is the clustering data structure, and $m$ is the desired size of a coreset.}

{\bf Init.}
$n \gets 0$\;
$\mathcal{D}.Init()$\;

$C \gets \emptyset$\;

\tcp{Insert points into $\mathcal{D}$ in batches of size $m$}
\If{point $p$ arrives}
{
  $n \gets (n+1)$\;
  Add $p$ to $C$\;
  
  \If{$(|C| = m)$}
  {
     $\mathcal{D}$.Update($C,n/m$)\;
     $C \gets \emptyset$\;
  }
}

\If{a clustering query arrives}{Return $\mathcal{D}$.Query()\;}
\end{algorithm}
}


\section{Experimental Evaluation}
\label{sec:expts}
\newcommand{\census}{\texttt{USCensus1990}\xspace}
\newcommand{\uscensus}{\texttt{USCensus1990}\xspace}
\newcommand{\tower}{\texttt{Tower}\xspace}
\newcommand{\covtype}{\texttt{Covtype}\xspace}
\newcommand{\power}{\texttt{Power}\xspace}
\newcommand{\intrusion}{\texttt{Intrusion}\xspace}
\newcommand{\drift}{\texttt{Drift}\xspace}

This section describes an empirical study of the proposed algorithms, in
comparison to the state-of-the-art clustering algorithms.  Our goals are
twofold: to understand the clustering accuracy and the running time of
different algorithms in the context of continuous queries, and to investigate
how they behave under different settings of algorithm parameters.

\subsection{Datasets}
Our experiments use a number of real-world or semi-synthetic datasets, based on
data from the UCI Machine Learning Repositories~\cite{Lichman13}.  These are
commonly used datasets in benchmarking clustering
algorithms. Table~\ref{table:dataset} provides an overview of the datasets used.

The \textit{\covtype} dataset models the forest cover type prediction problem
from cartographic variables. The dataset contains $581,012$ instances and $54$
integer attributes. The \textit{\power} dataset measures electric power
consumption in one household with a one-minute sampling rate over a period of
almost four years. We remove entries with missing values, resulting in a dataset
with $2,049,280$ instances and $7$ floating-point attributes. The
\textit{\intrusion} dataset is a $10\%$-subset of the \emph{KDD Cup 1999}
data. The task was to build a predictive model capable of distinguishing between
normal network connections and intrusions. After ignoring symbolic attributes,
we have a dataset with $494,021$ instances and $34$ floating-point
attributes. To erase any potential special ordering within data, we randomly
shuffle each dataset before using it as a data stream.

However, each of the above datasets, as well as most datasets used in previous
works on streaming clustering, is not originally a streaming dataset; the
entries are only read in some order and consumed as a stream.  
To better model the evolving nature of data streams and the drifting of center locations, 
we generate a semi-synthetic dataset, called \drift, which we derive from the \census
dataset~\cite{Lichman13} as follows: The method is inspired by
Barddal~\cite{BGE+16}.  The first step is to cluster the \census dataset to
compute $20$ cluster centers and for each cluster, the standard deviation of the
distances to the cluster center.  Following that, the synthetic dataset is
generated using the Radial Basis Function (RBF) data generator from the MOA
stream mining framework~\cite{BHK+10}. The RBF generator moves the drifting
centers with a user-given direction and speed. For each time step, the RBF
generator creates $100$ random points around each center using a Gaussian
distribution with the cluster standard deviation. In total, the synthetic
dataset contains $200,000$ and $68$ floating-point attributes.

\begin{table}
{ 
\centering
\setlength{\tabcolsep}{3pt}
\begin{tabular}
{l  c  c  l}
\toprule
Dataset & Number of Points & Dimension & Description\\ 
\midrule
\covtype & $581,012$ & $54$ & Forest cover type\\
\power & $2,049,280$ & $7$ & Household  power consumption \\
\intrusion & $494,021$ & $34$ & KDD Cup 1999\\
\drift & $200,000$ & $68$ & Derived from US Census 1990\\
\bottomrule
\end{tabular}\par
\smallskip
\caption{An overview of the datasets used in the experiments.}
\label{table:dataset}
}
\end{table}

\subsection{Experimental Setup and Implementation Details}
We implemented all the clustering algorithms in Java, and ran experiments on a
desktop with Intel Core i5-4460 $3.2$GHz processor and $8$GB main memory.

\noindent\textbf{Algorithms Implementation:} Our baseline algorithms are two prominent streaming clustering algorithms. (1)~the \seqkm algorithm due to
MacQueen~\cite{MacQueen67}, which is frequently implemented in clustering
packages today. For \seqkm clustering, we use the implementation in Apache Spark
MLLib~\cite{MBY+15}, though ours is modified to run sequentially. Furthermore,
the initial centers are set by the first $k$ points in the stream instead of
setting by random Gaussians, to ensure no clusters are empty. (2)~We also
implemented \skmpp~\cite{AMR+12}, a current state-of-the-art algorithm that has
good practical performance.  \skmpp can be viewed as a special case of \ct where
the merge degree $r$ is $2$. The bucket size is $20 \cdot k$, where $k$ is the number
of centers. \footnote{A larger bucket size such as $200k$ can yield slightly 
better clustering quality. But this led to a high runtime for \skmpp, 
especially when queries are frequent, so we use a smaller bucket size.}

For \cc, we set the merge degree to $2$, in line with \skmpp. For \rcc, we use a
maximum nesting depth of $3$, so the merge degrees for different structures are
$N^{\frac{1}{2}}, N^{\frac{1}{4}}$ and $N^{\frac{1}{8}}$, respectively. For
\hybrid, the threshold $\alpha$ is set to $1.2$ by default. 

We use the batch \kmpp algorithm as the baseline for clustering accuracy. This
is expected to outperform any streaming algorithm, as with the scope of all the points. The \kmpp algorithm, similar to~\cite{AMR+12,AJM09}, is used to derive coresets. We also use \kmpp as the final step to construct $k$ centers from the coreset, and take the best clustering out of five independent runs of \kmpp; 
each run of \kmpp is followed by up to $20$ iterations of Lloyd's algorithm to 
further improve clustering quality. 
Finally for each statistic, we report the median from nine independent runs of each algorithm to improve robustness. The queries on cluster centers present with interval of $q$ points. 
Hence, from the beginning of the stream, there is one query per $q$ input points received. By default, the number of clusters is set to $30$, the query interval is set to $100$ points. To present queries in a practical way, we also generate the queries in a poisson process. Let $\lambda$ be the arrival rate of the poisson process, the inter arrival time between query events should be an exponential distribution variable with mean of $1 / \lambda$. We set the inter arrival of queries in range of $\{50, 100, 200, 400, 800, 1600, 3200 \}$ points.

\noindent\textbf{Metrics:} 
We use three metrics: clustering accuracy, runtime and memory cost. 
The clustering accuracy is measured using the \km cost, also known as the
within cluster sum of squares (SSQ).  We measure the average runtime of the
algorithm per point, as well as the total runtime over the entire stream. 
The runtime contains two parts, (1)~update time, the time required to update internal
data structures upon receiving new point, and (2)~query time, the time required
to answer clustering queries.  Finally, we consider the memory consumption
through measuring the number of points stored by the internal data structure, 
including both the coreset tree and coreset cache. From the number of points, 
we estimate the number of bytes used, assuming that each dimension of 
a data point consumes $8$ bytes (size of a double-precision floating-point number).



\begin{figure*}
  \centering
  \subfloat[\covtype]{
  		\includegraphics[width=0.24\textwidth]{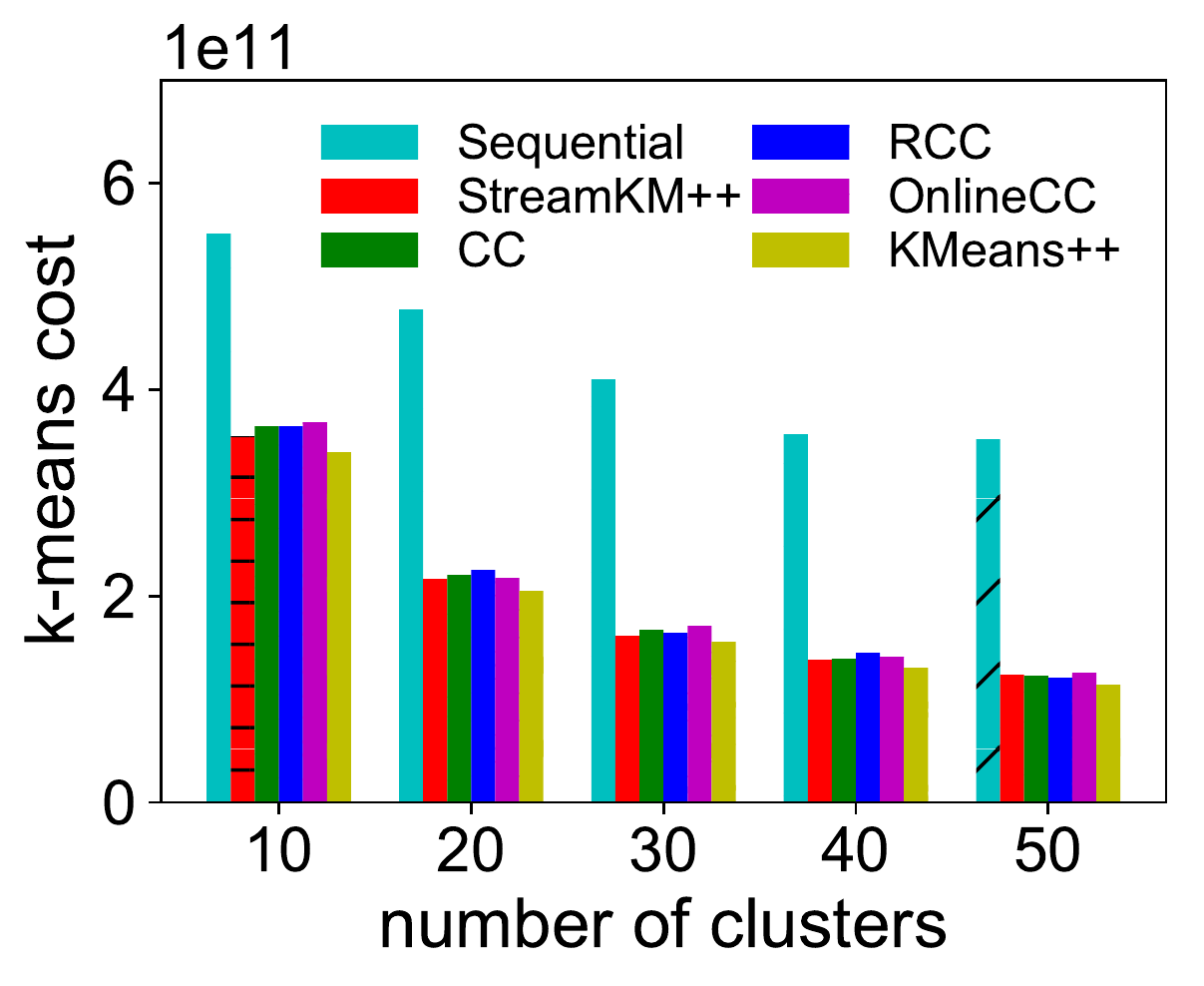}
  }
  \subfloat[\power]{
  		\includegraphics[width=0.24\textwidth]{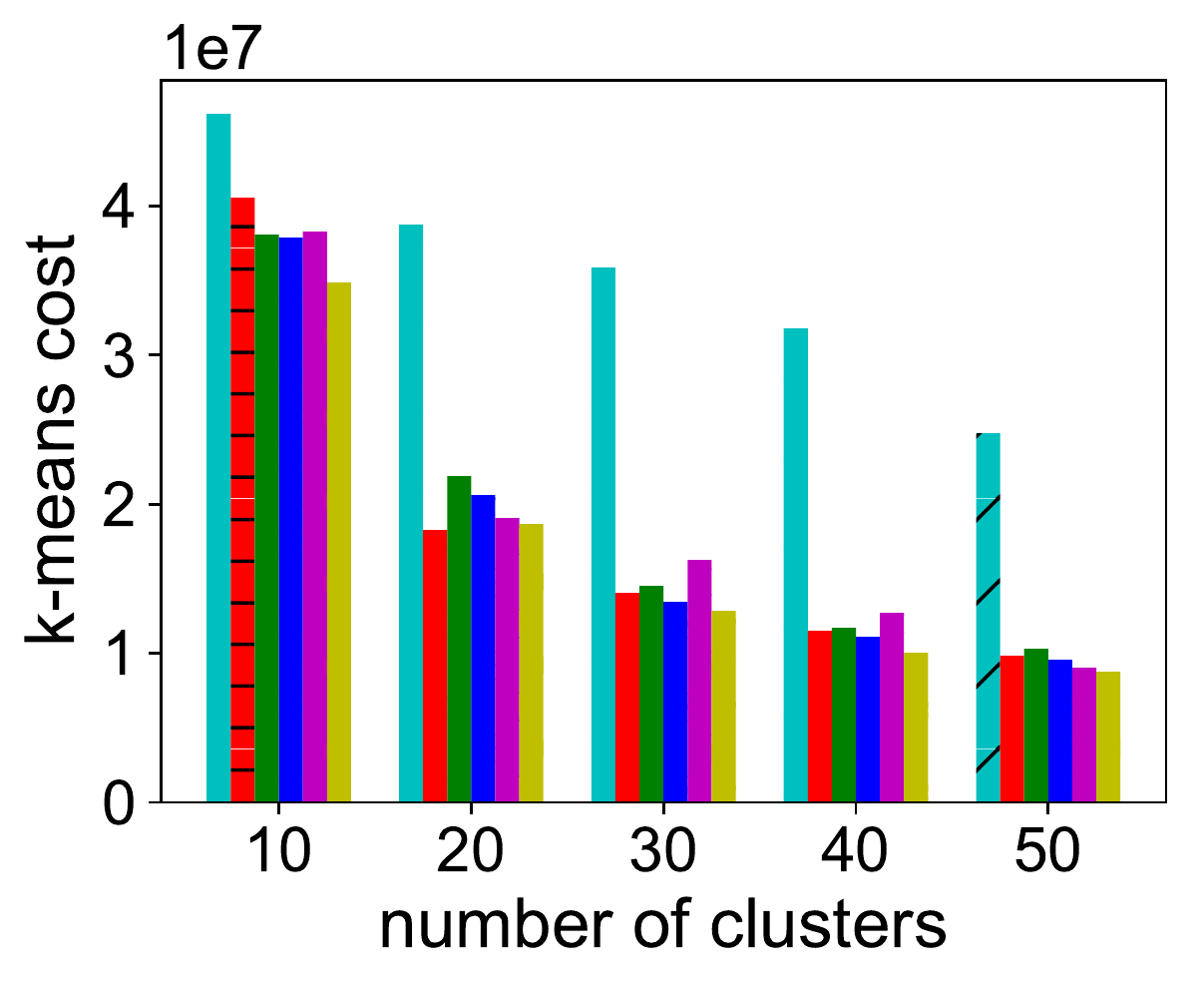}
  }
  \subfloat[\intrusion (\km cost of \seqkm is not shown)]{
		\includegraphics[width=0.24\textwidth]{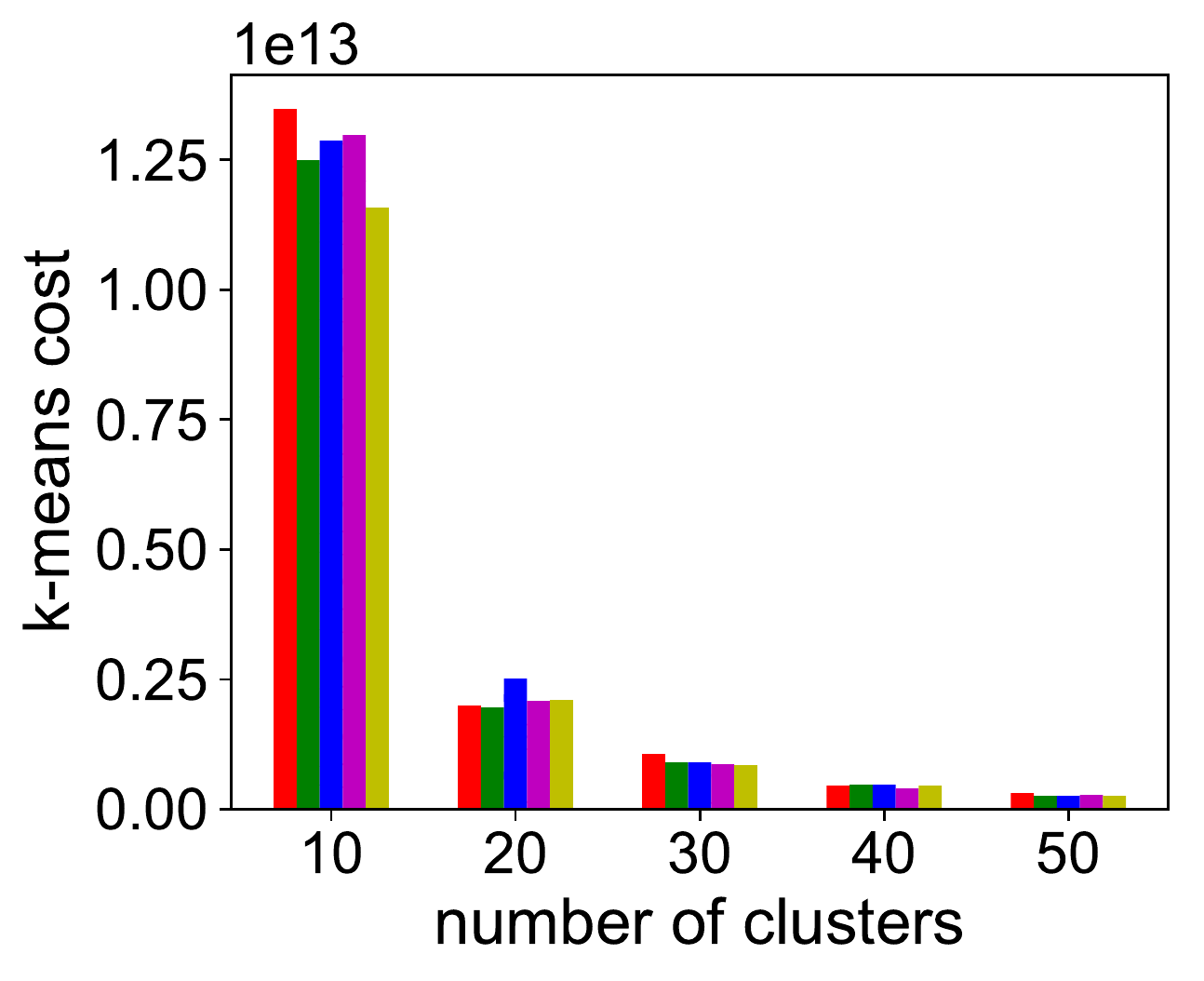}
  }
  \subfloat[\drift]{
  		\includegraphics[width=0.24\textwidth]{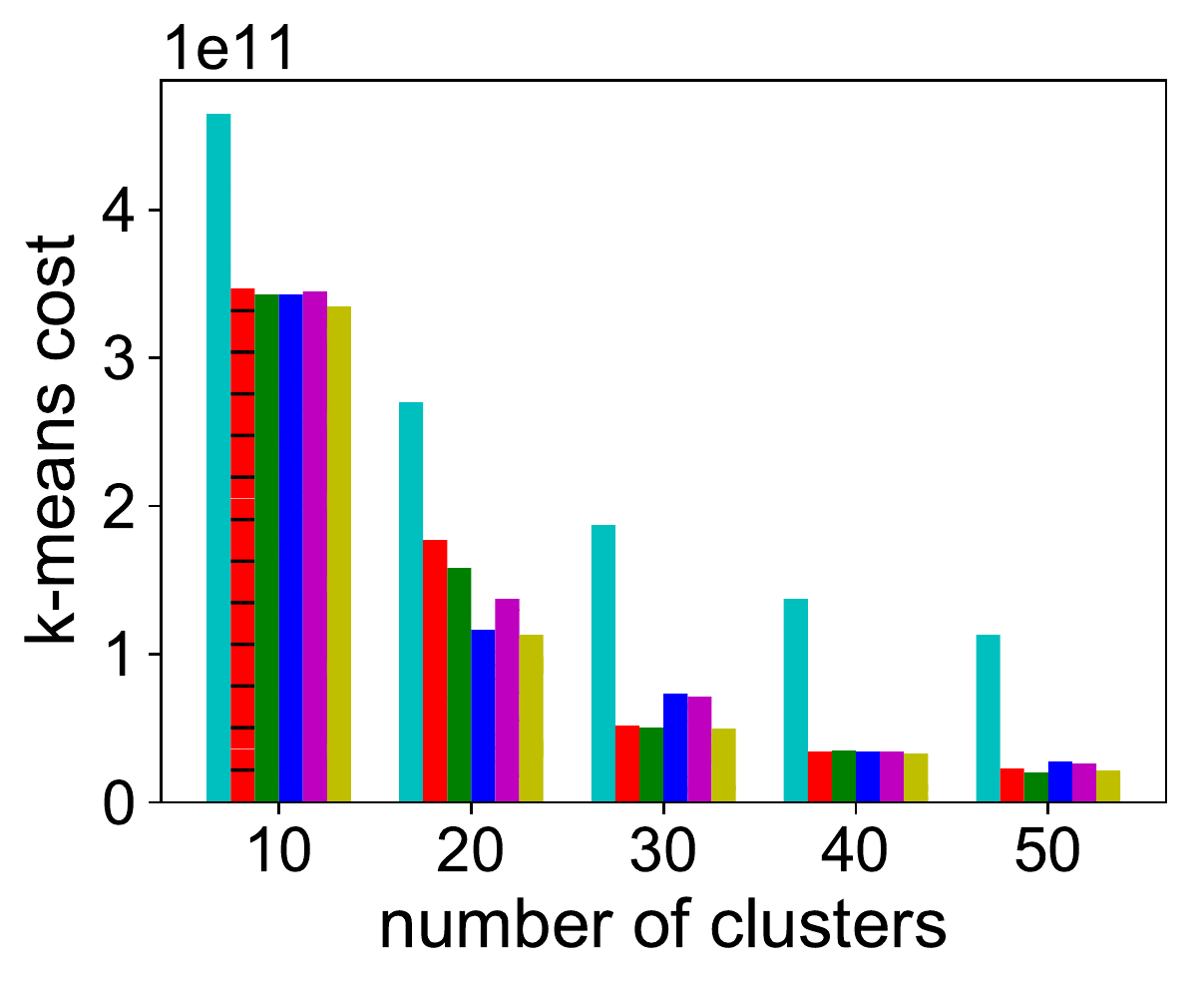}
  }
  \caption{\km cost vs. number of clusters $k$. The cost is computed at the end of observing all the points.  
\km cost of \seqkm on \intrusion dataset is not shown in Figure (c), since it was orders of magnitude larger ($10^4$) than the other algorithms.}
\label{fig:cost-versus-k}
\end{figure*}
\begin{figure*}
  \centering
  \subfloat[\covtype]{
  \centering \includegraphics[width=0.23\textwidth]{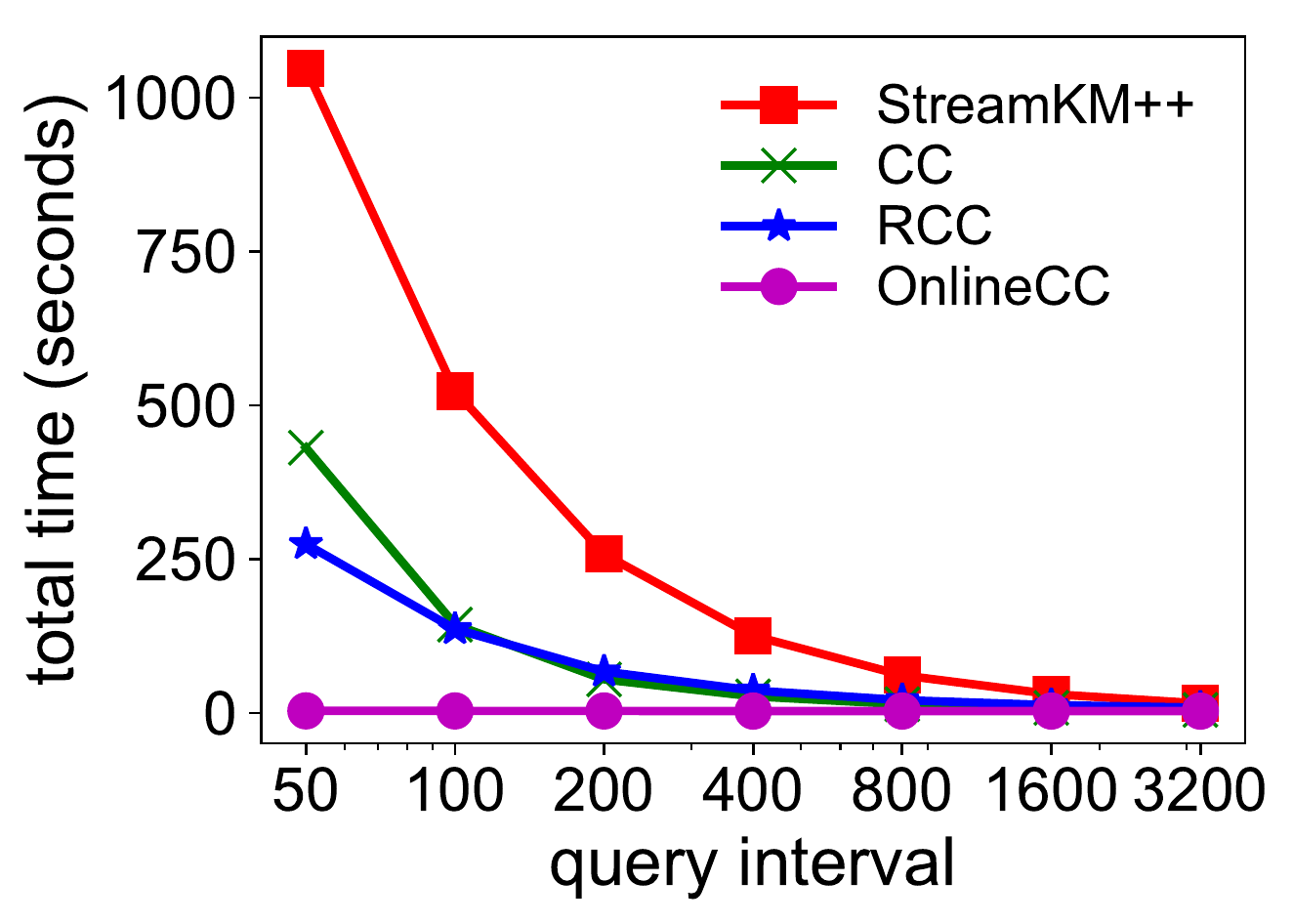}
  }
  \subfloat[\power]{
  \centering \includegraphics[width=0.23\textwidth]{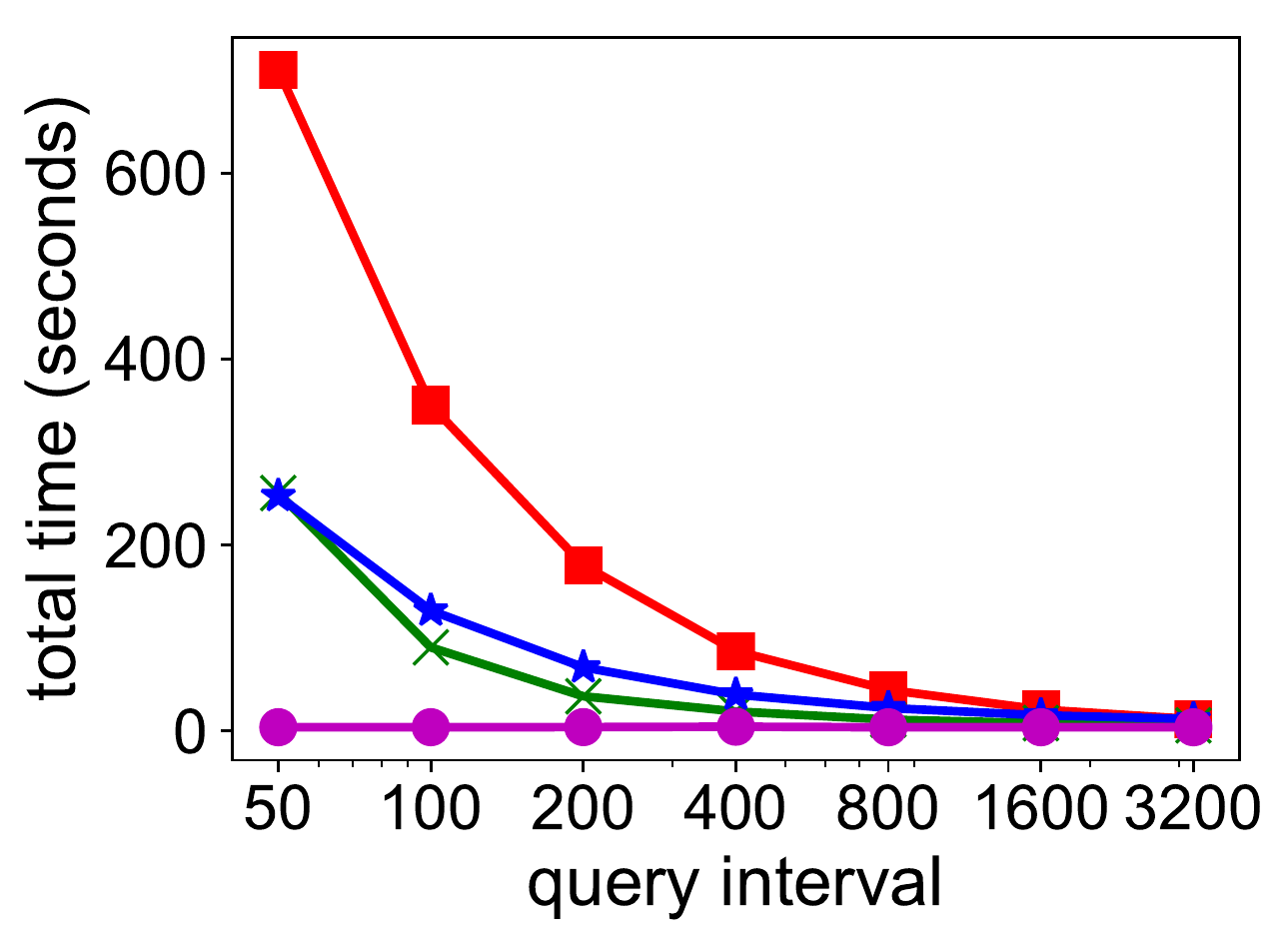}
  }
  \subfloat[\intrusion]{
  \centering \includegraphics[width=0.23\textwidth]{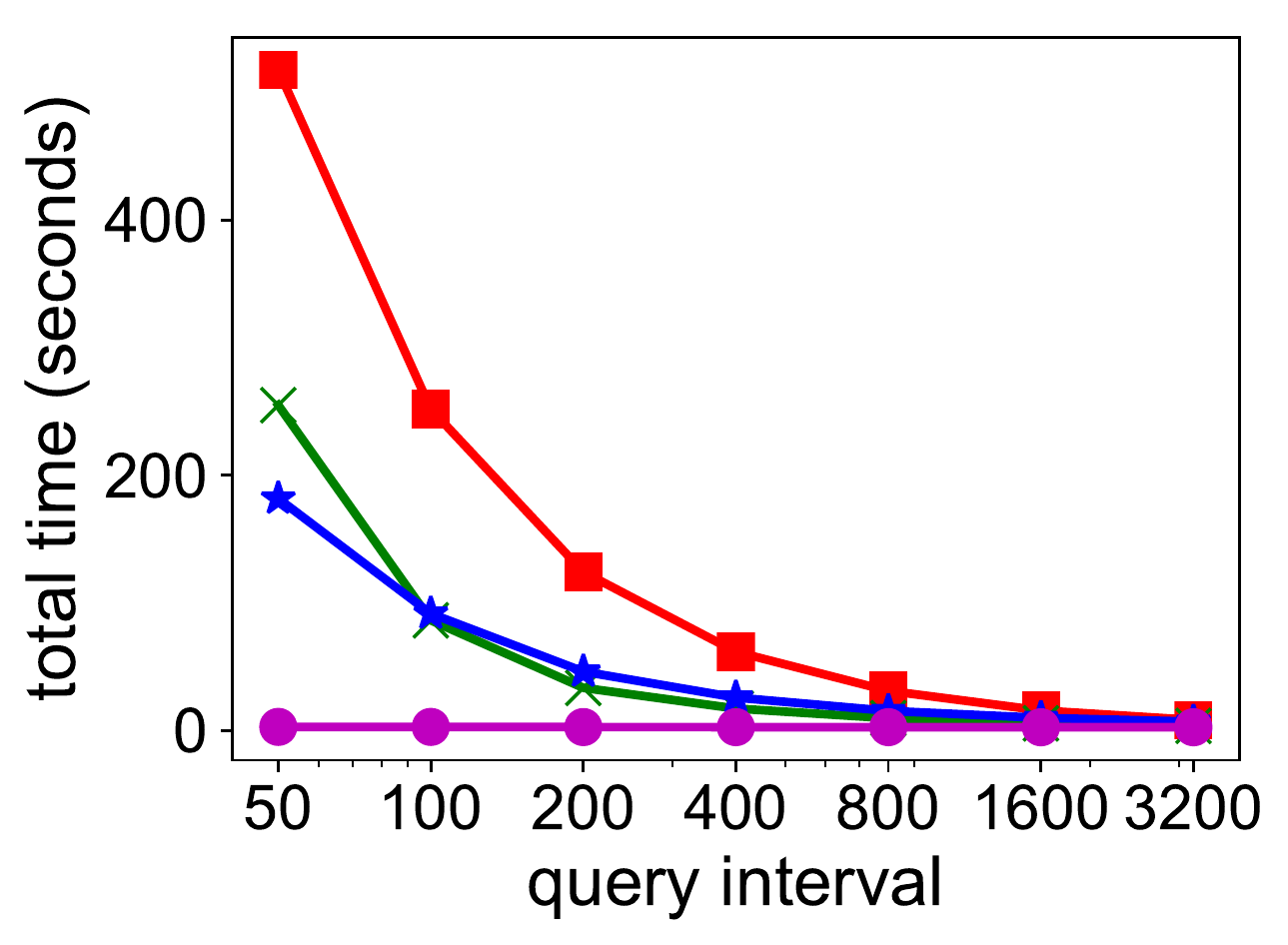}
  }
  \subfloat[\drift]{
  \centering \includegraphics[width=0.23\textwidth]{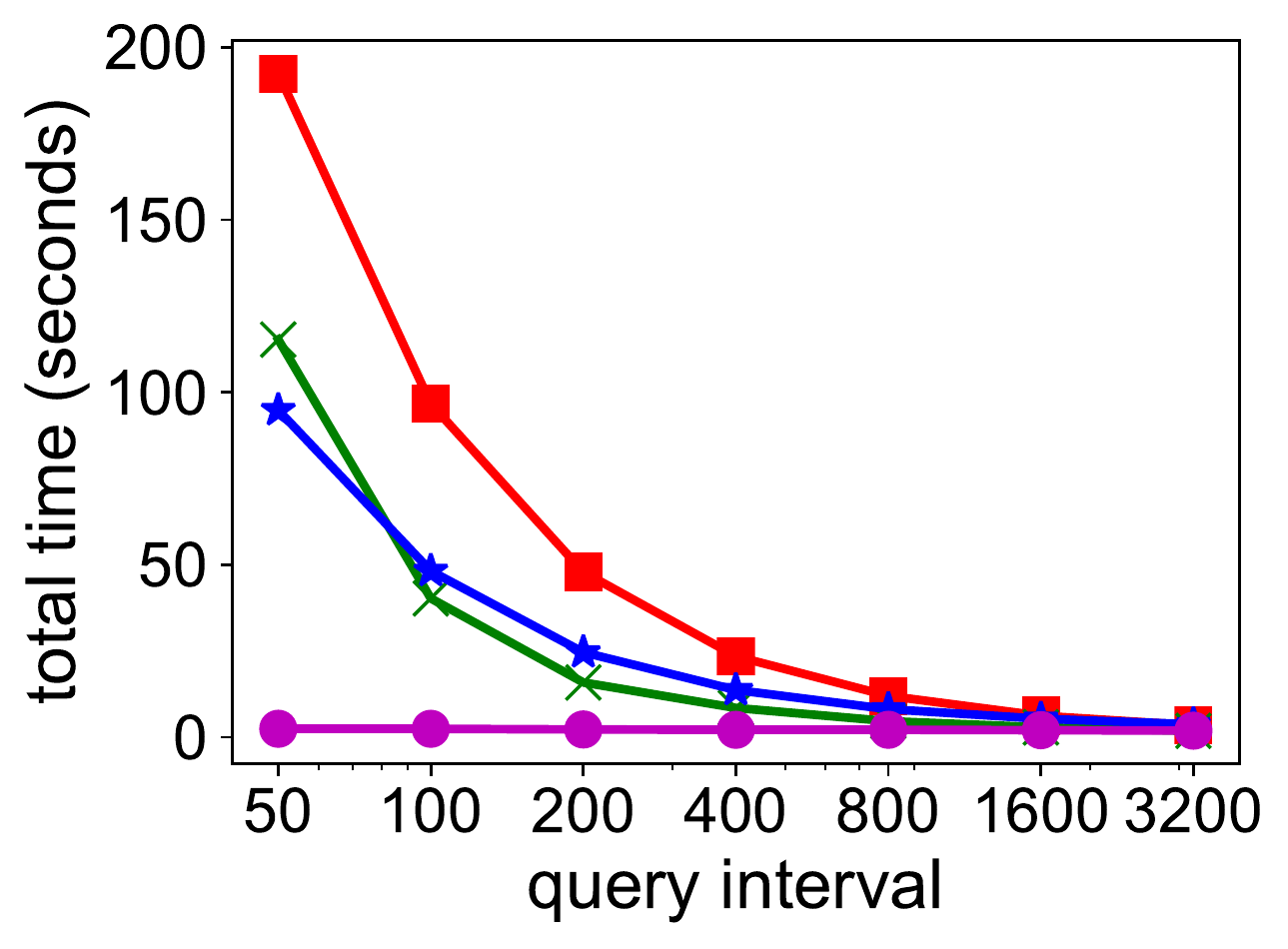}
  }
  \caption{Total time (seconds) vs. query interval $q$. The total time is for the entire dataset overall stream. For every $q$ points, there is a query for the cluster centers. The number of centers $k$ is $30$.}
 \label{fig:time-versus-q}
\end{figure*}

\begin{figure*}
  \centering
  \subfloat[\covtype]{
  \centering \includegraphics[width=0.23\textwidth]{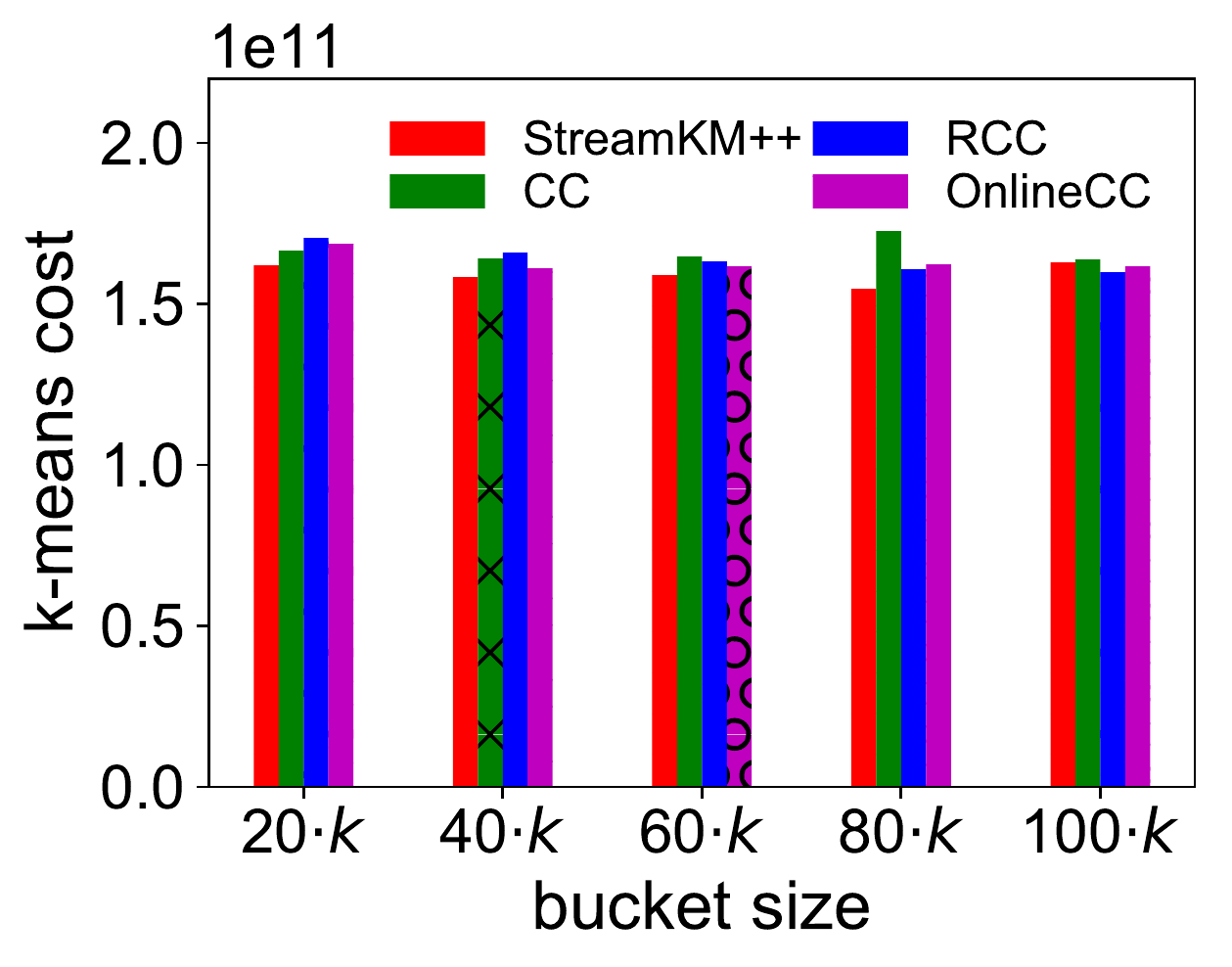}
  }
  \subfloat[\power]{
  \centering \includegraphics[width=0.23\textwidth]{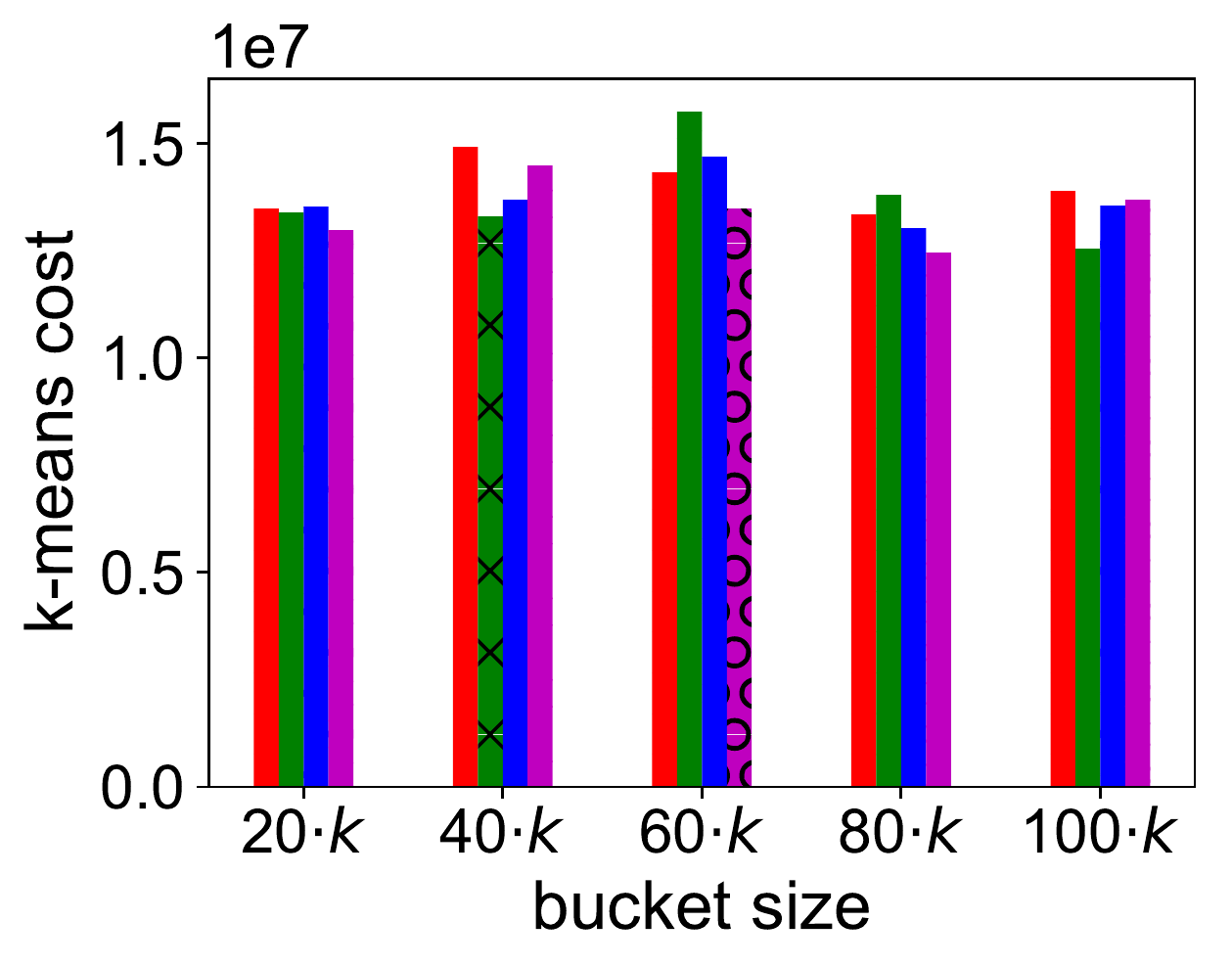}
  }
  \subfloat[\intrusion]{
  \centering \includegraphics[width=0.23\textwidth]{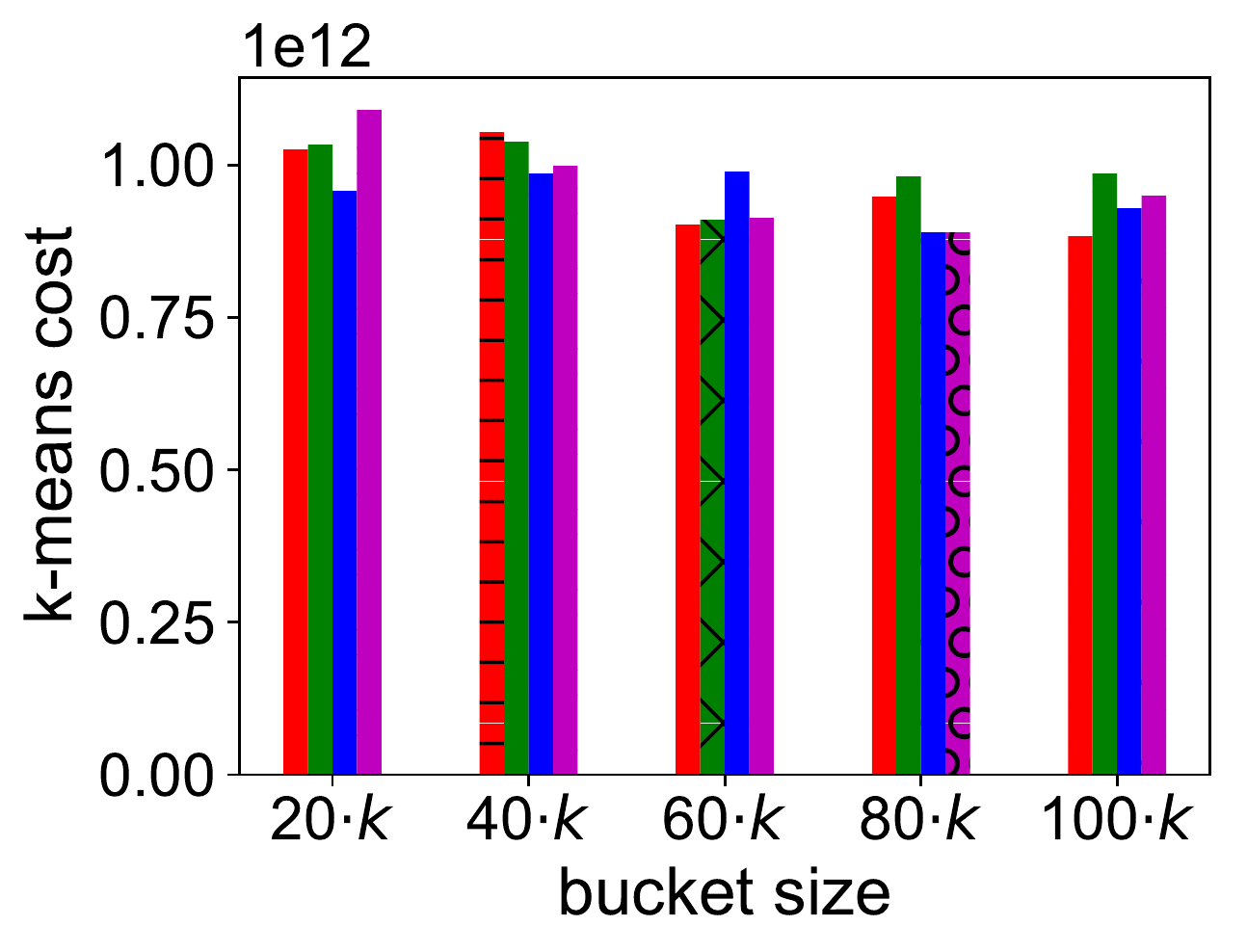}
  }
  \subfloat[\drift]{
  \centering \includegraphics[width=0.23\textwidth]{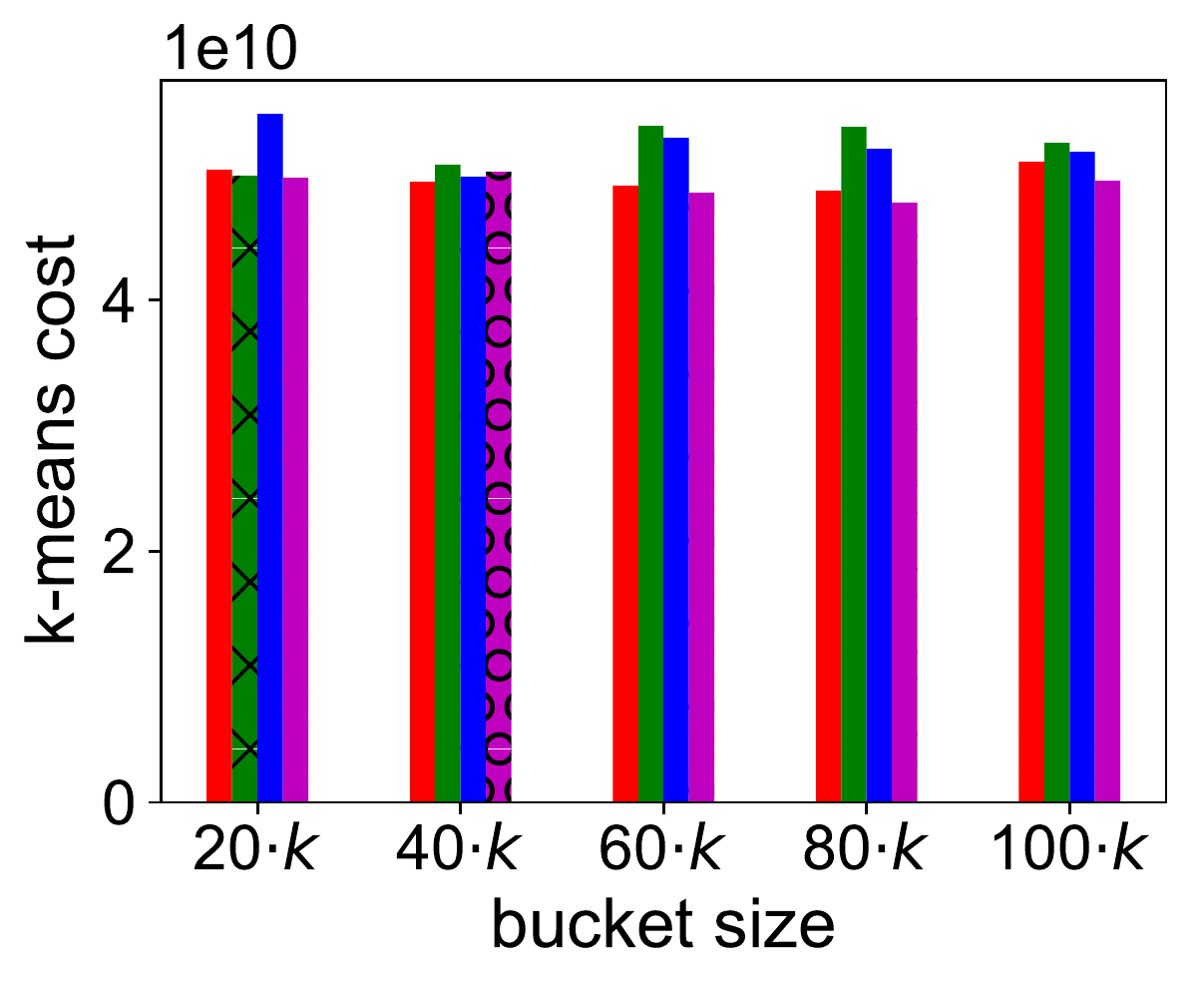}
  }
  \caption{\km cost vs. bucket size $m$. The cost is computed at the end of observing all the points. The number of clusters $k=30$, query interval $q=100$.  }
 \label{fig:cost-versus-m}
\end{figure*}
\begin{figure*}
  \centering
  \subfloat[\covtype]{
  \centering \includegraphics[width=0.23\textwidth]{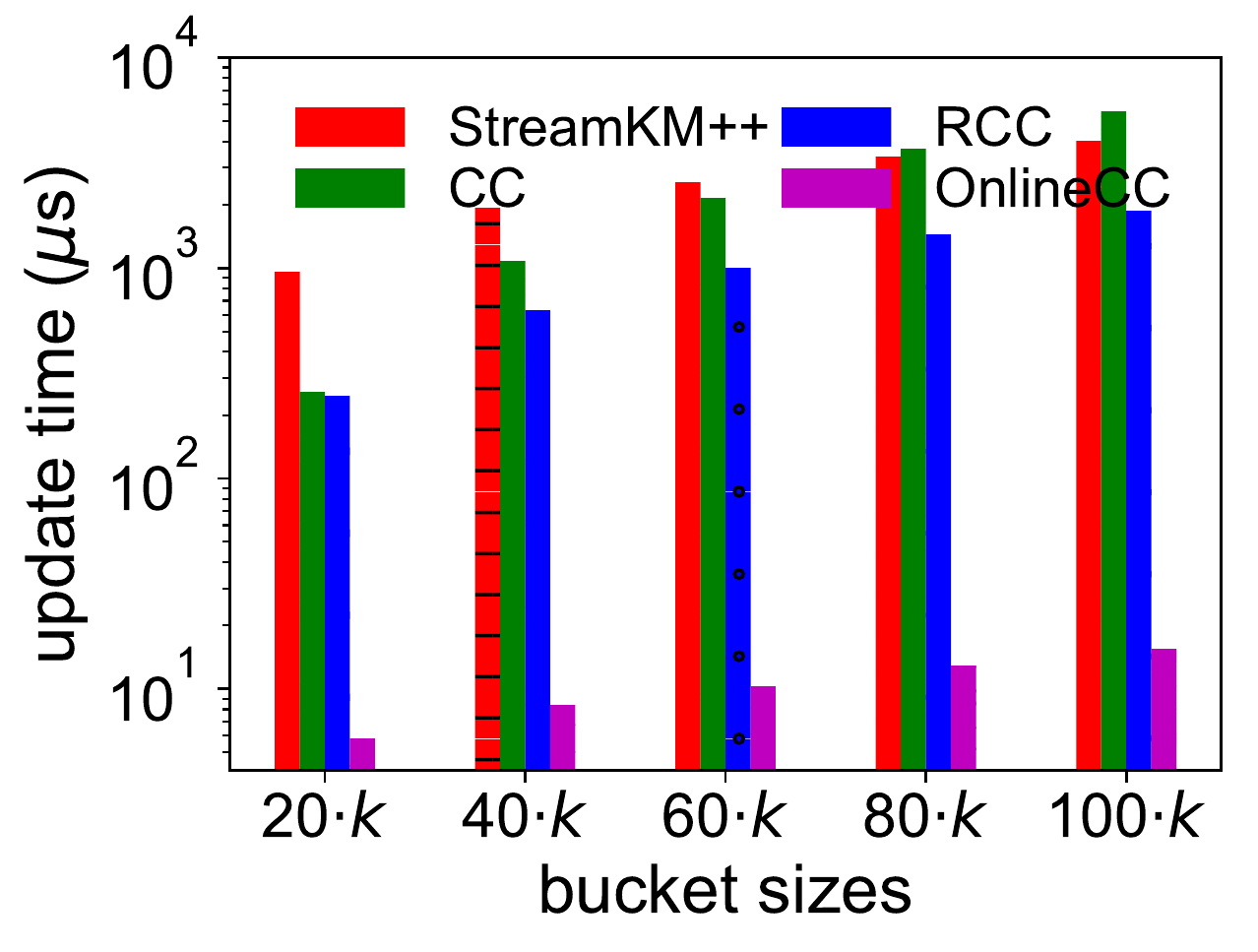}
  }
  \subfloat[\power]{
  \centering \includegraphics[width=0.23\textwidth]{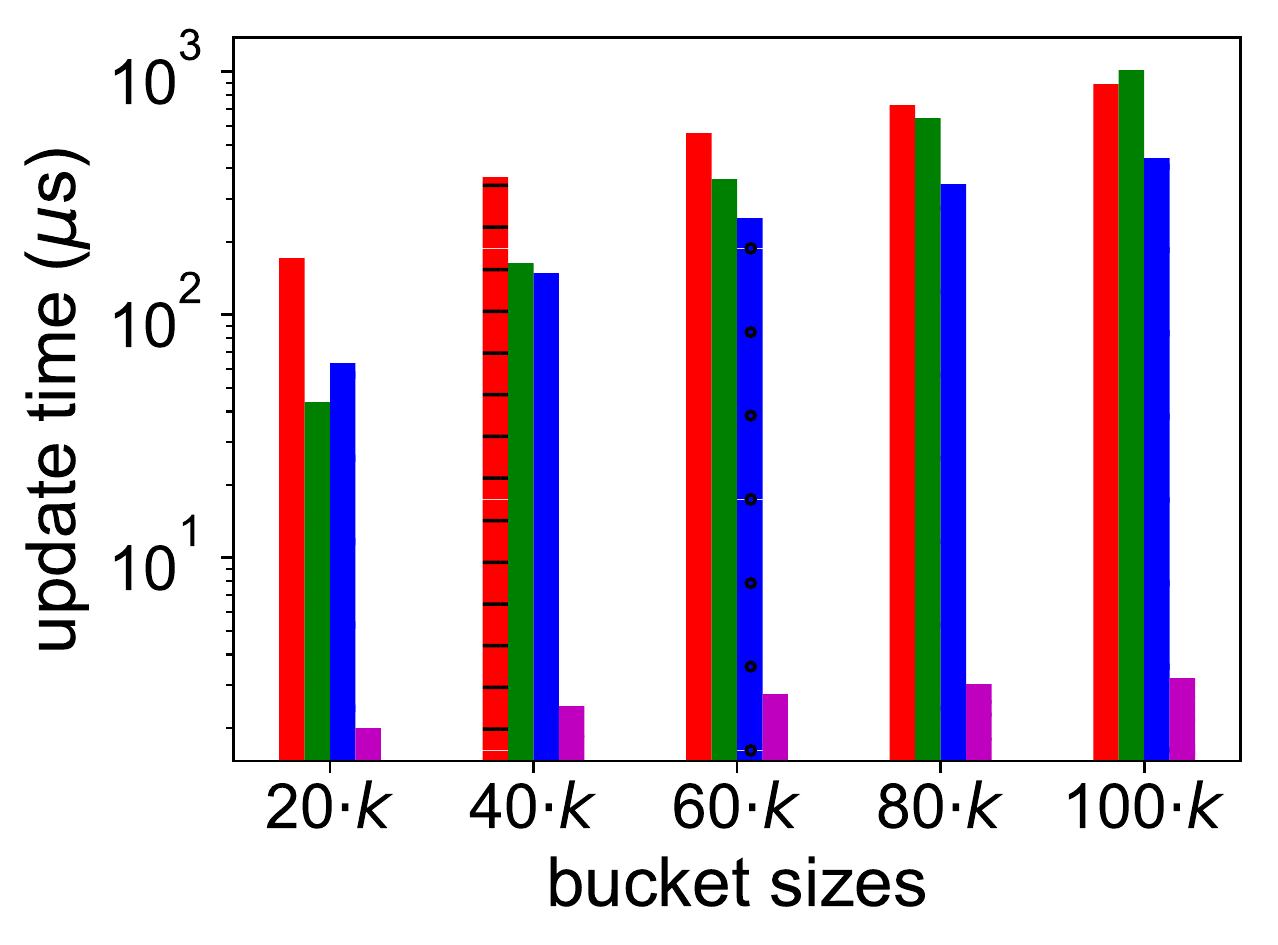}
  }
  \subfloat[\intrusion]{
  \centering \includegraphics[width=0.23\textwidth]{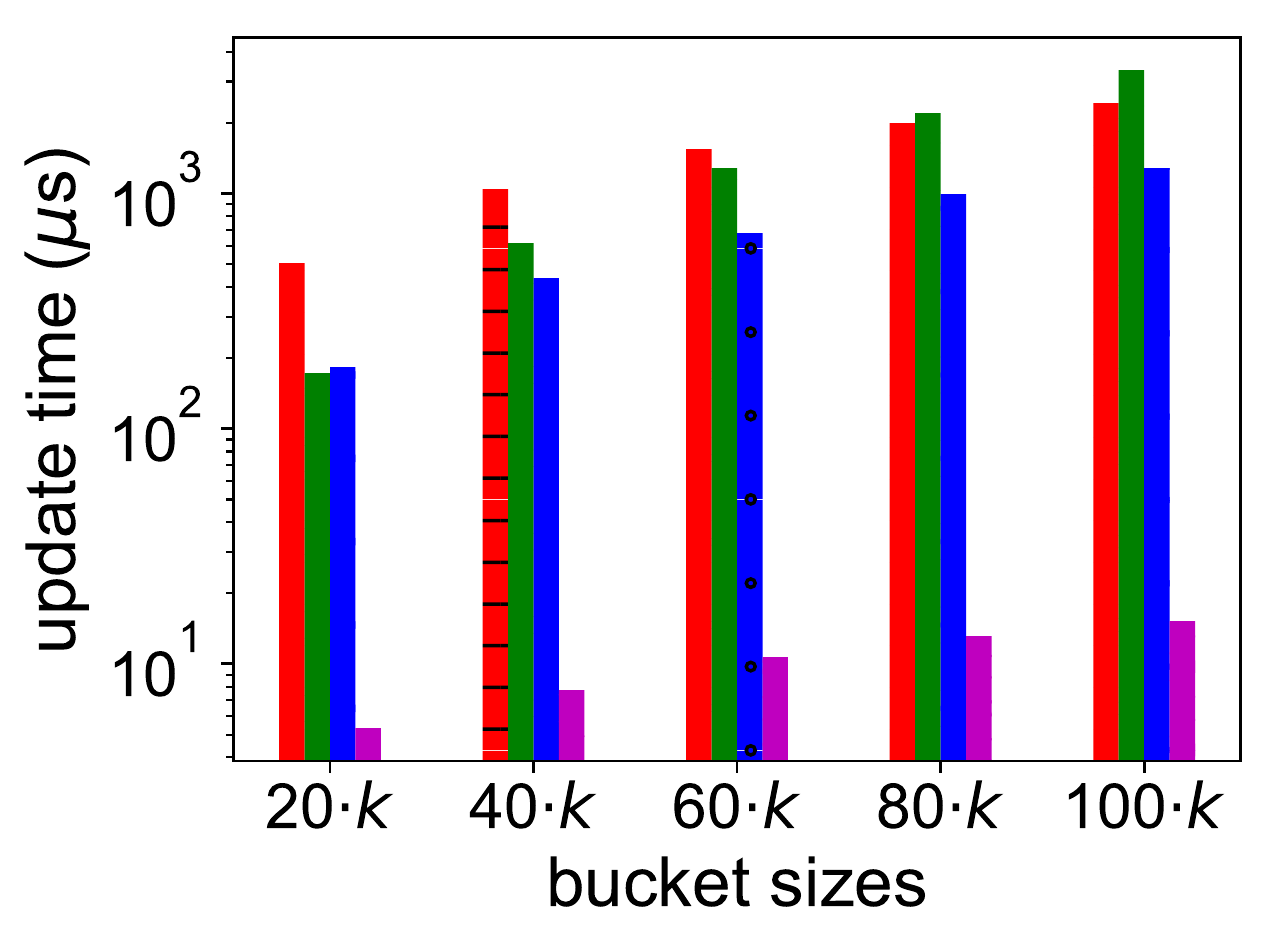}
  }
  \subfloat[\drift]{
  \centering \includegraphics[width=0.23\textwidth]{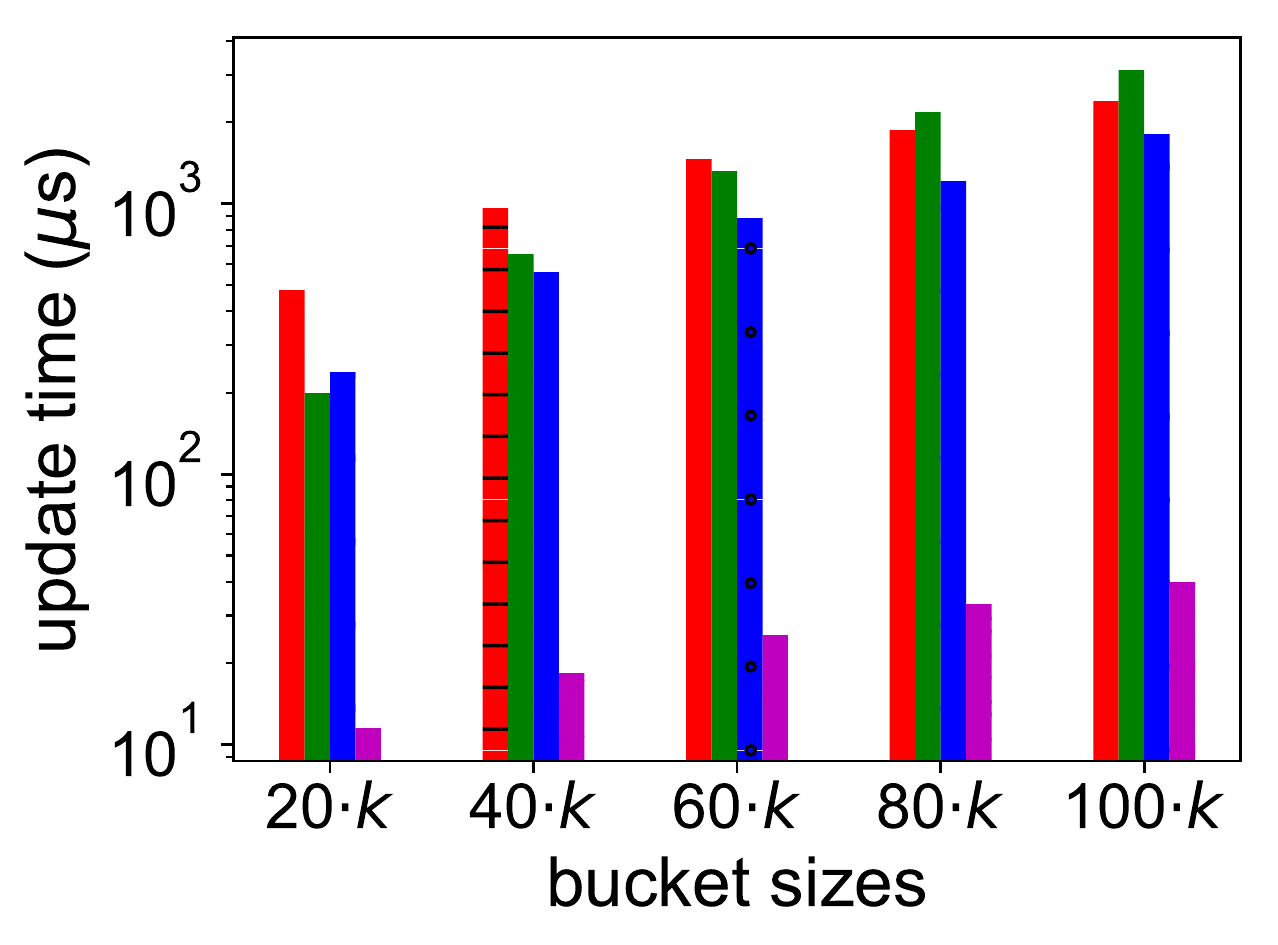}
  }
  \caption{Average runtime per point (microseconds) vs. bucket size $m$. The runtime is sum of the update time and query time, both counted for average per point. The number of clusters $k=30$, query interval $q=100$.  }
 \label{fig:total-versus-m}
\end{figure*}

\begin{figure*}
  \centering
  \subfloat[\covtype]{
  \centering \includegraphics[width=0.23\textwidth]{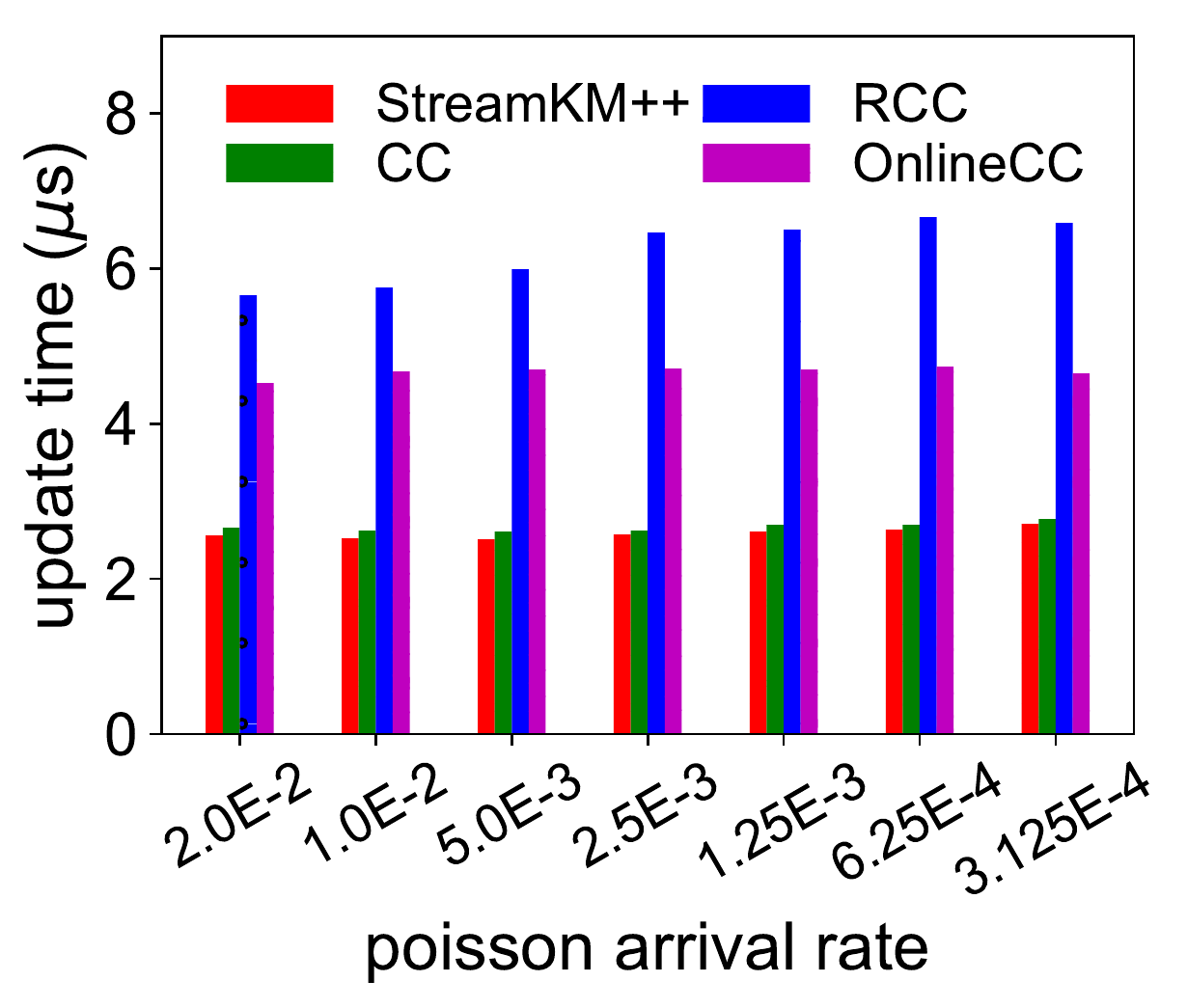}
  }
  \subfloat[\power]{
  \centering \includegraphics[width=0.23\textwidth]{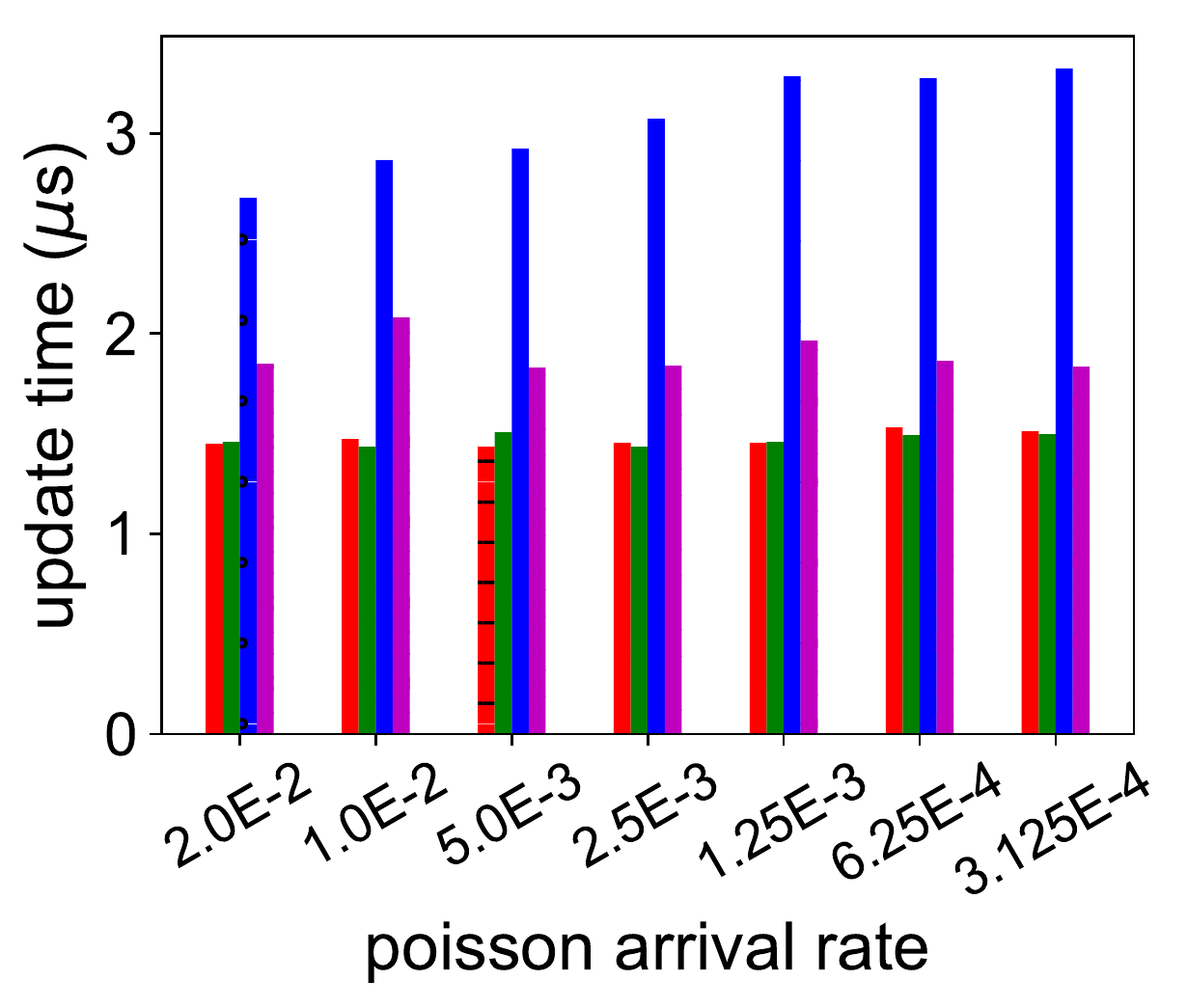}
  }
  \subfloat[\intrusion]{
  \centering \includegraphics[width=0.23\textwidth]{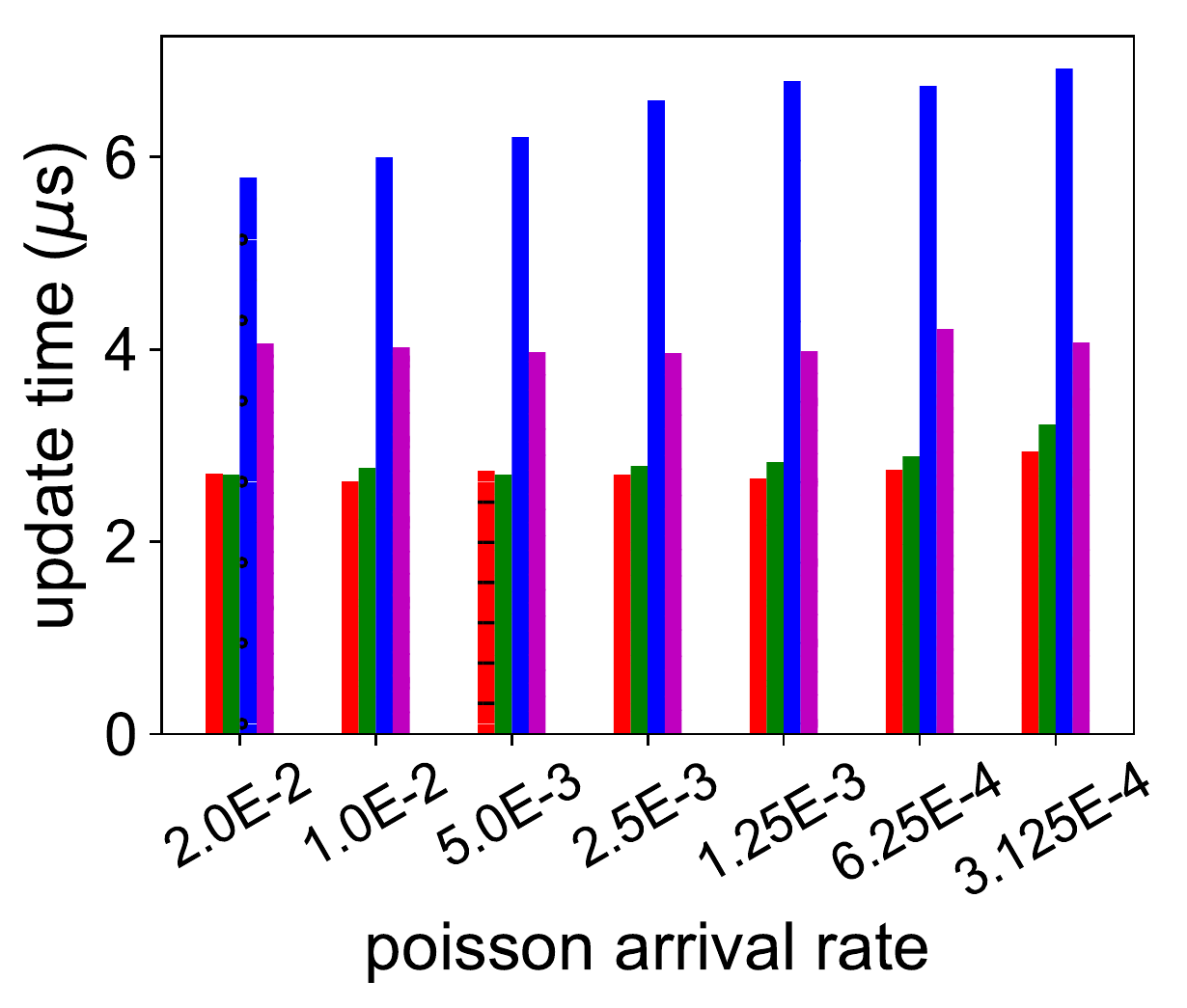}
  }
  \subfloat[\drift]{
  \centering \includegraphics[width=0.23\textwidth]{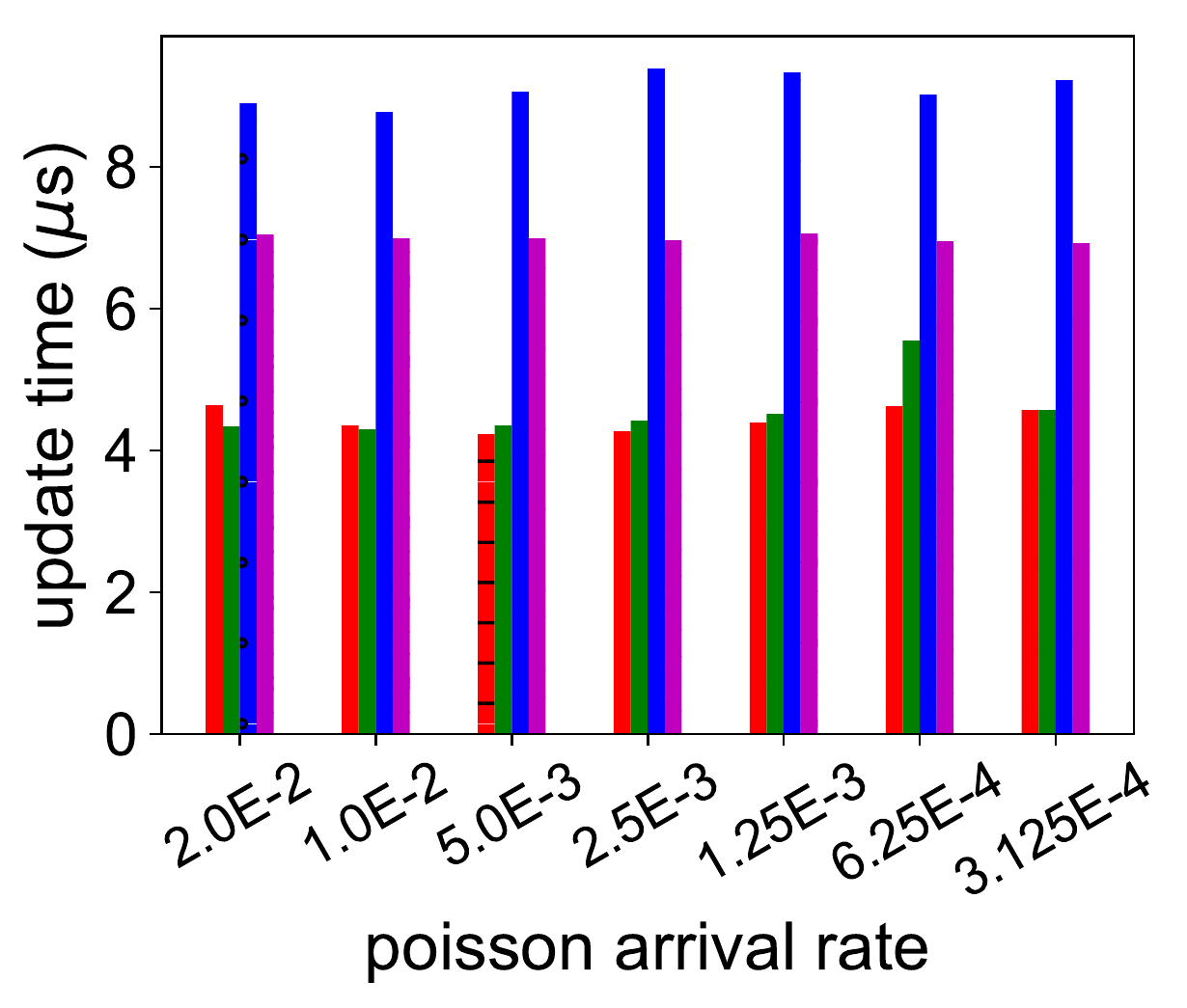}
  }
  \caption{Update time per point (microseconds)  vs. poisson arrival rate $\lambda$. }
 \label{fig:poisson-update}
\end{figure*}

\begin{figure*}
  \centering
  \subfloat[\covtype]{
  \centering \includegraphics[width=0.23\textwidth]{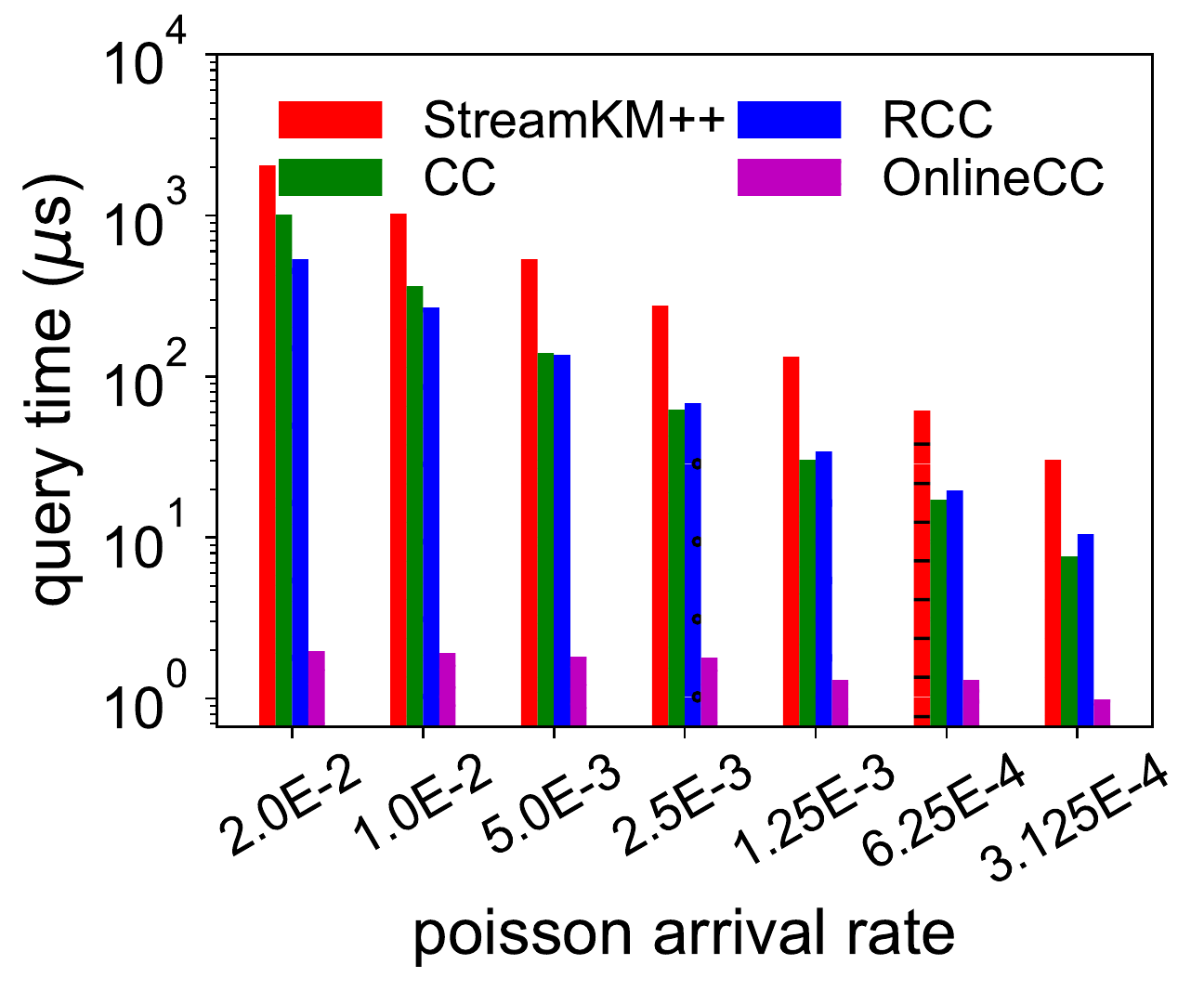}
  }
  \subfloat[\power]{
  \centering \includegraphics[width=0.23\textwidth]{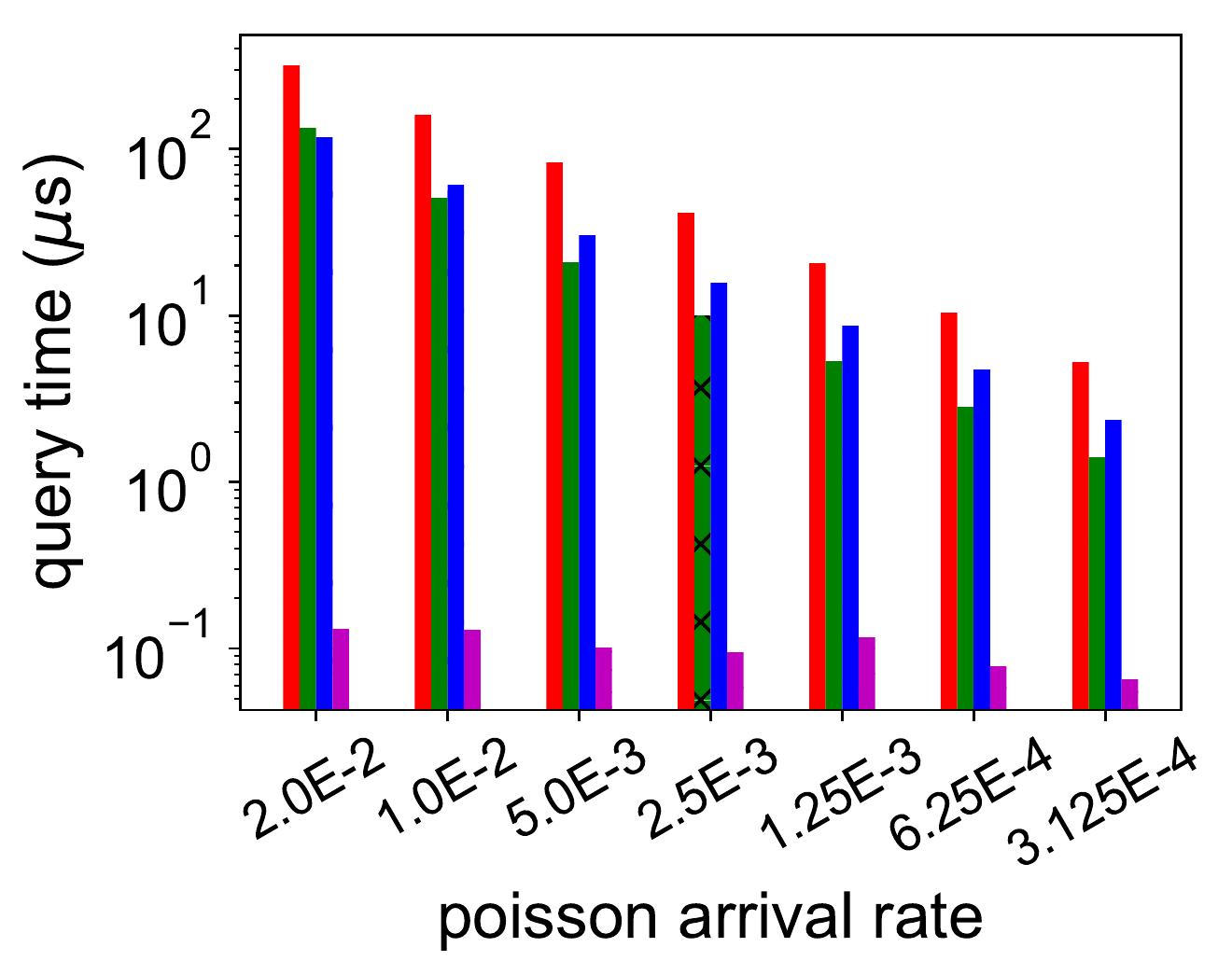}
  }
  \subfloat[\intrusion]{
  \centering \includegraphics[width=0.23\textwidth]{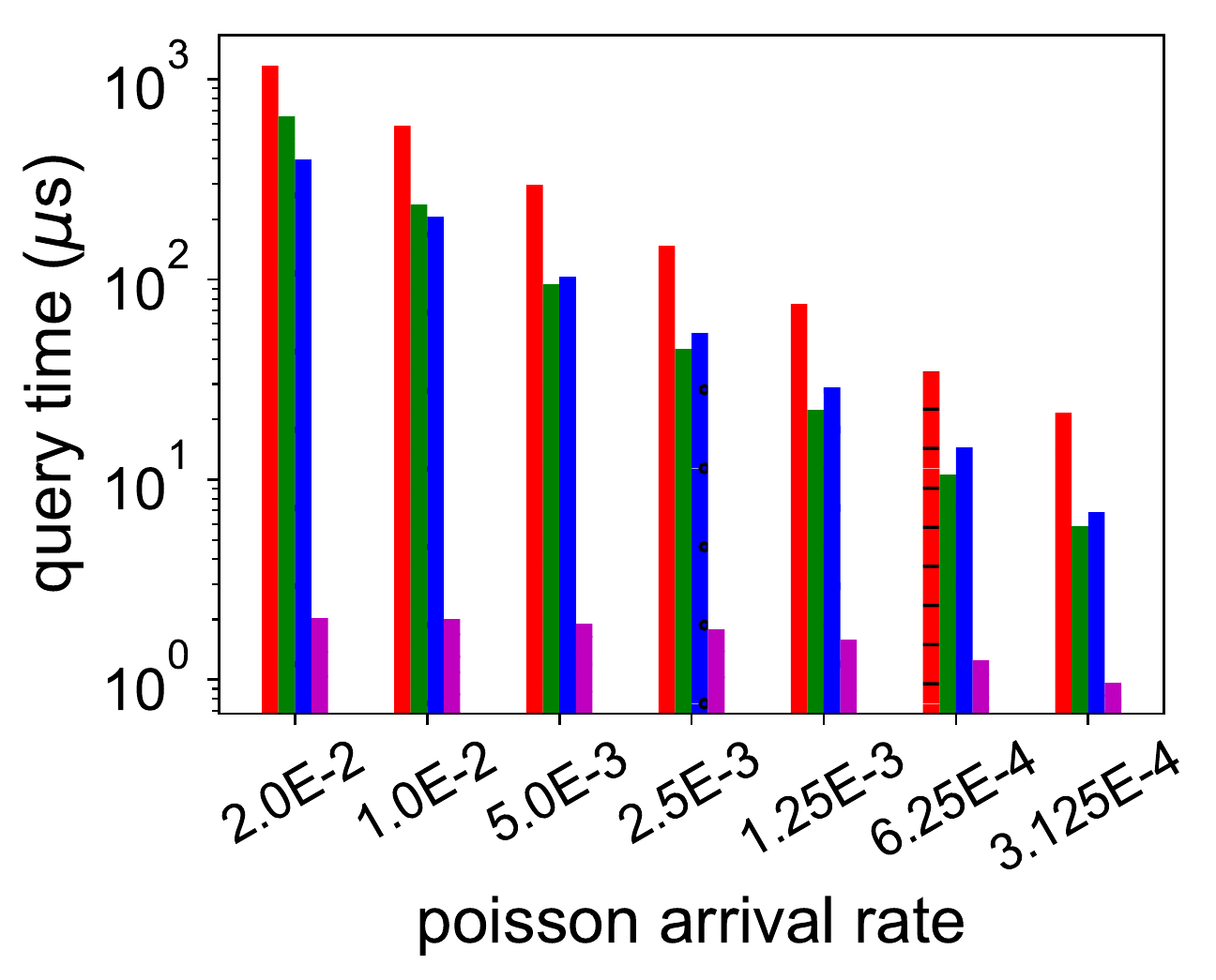}
  }
  \subfloat[\drift]{
  \centering \includegraphics[width=0.23\textwidth]{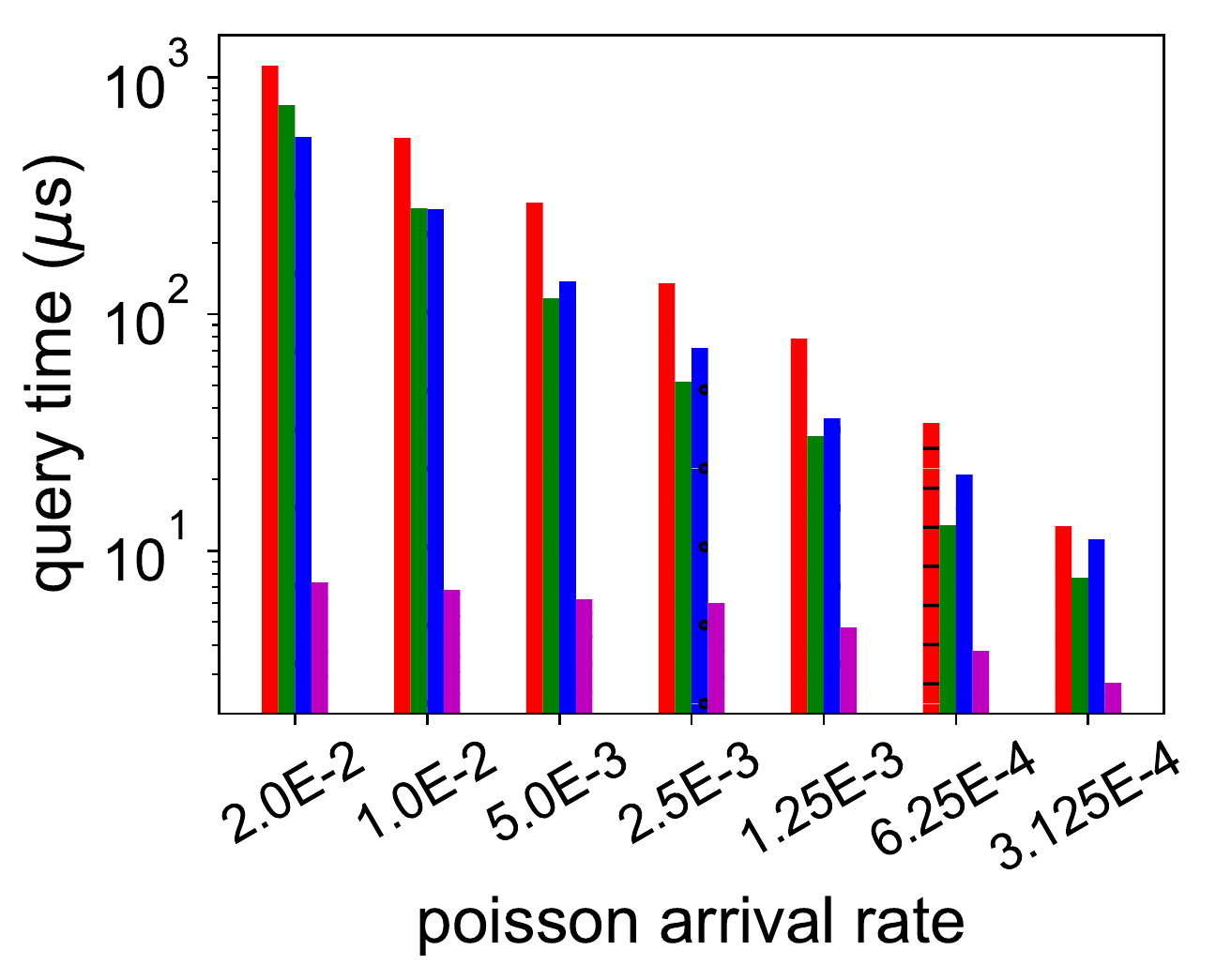}
  }
  \caption{Query time per point (microseconds)  vs. poisson arrival rate $\lambda$.}
 \label{fig:poisson-query}
\end{figure*}

\begin{figure*}
  \centering
  \subfloat[\covtype]{
  \centering \includegraphics[width=0.23\textwidth]{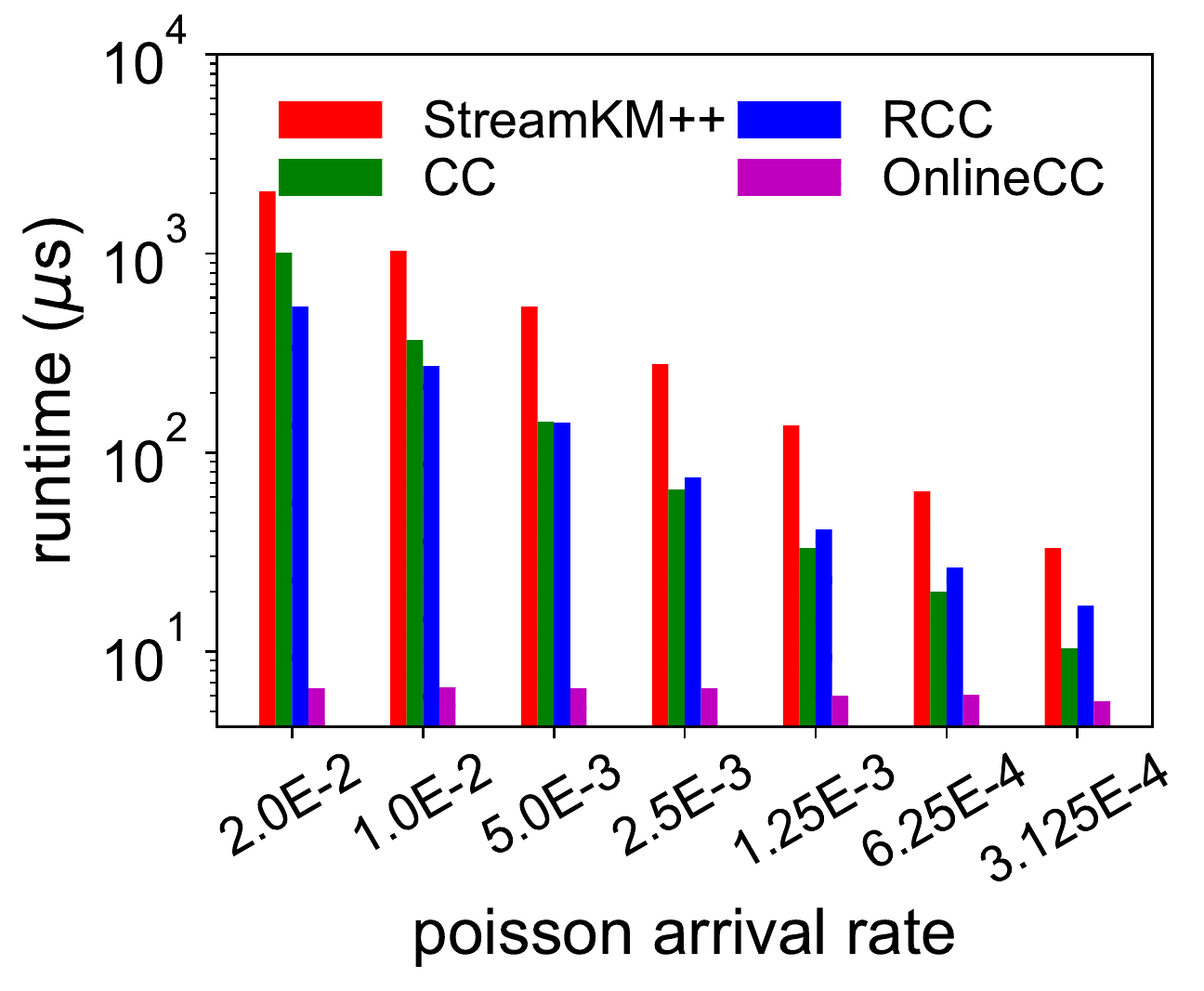}
  }
  \subfloat[\power]{
  \centering \includegraphics[width=0.23\textwidth]{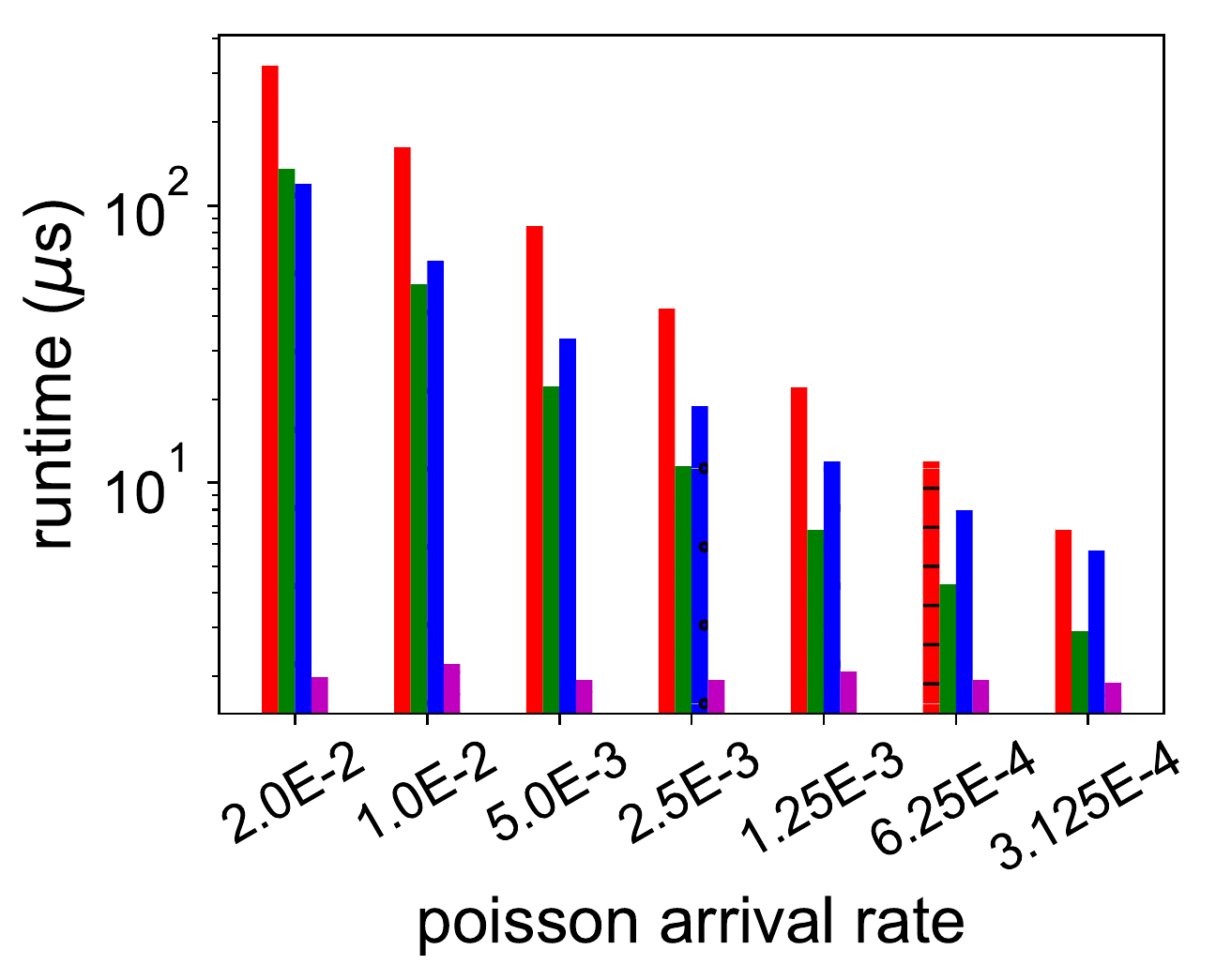}
  }
  \subfloat[\intrusion]{
  \centering \includegraphics[width=0.23\textwidth]{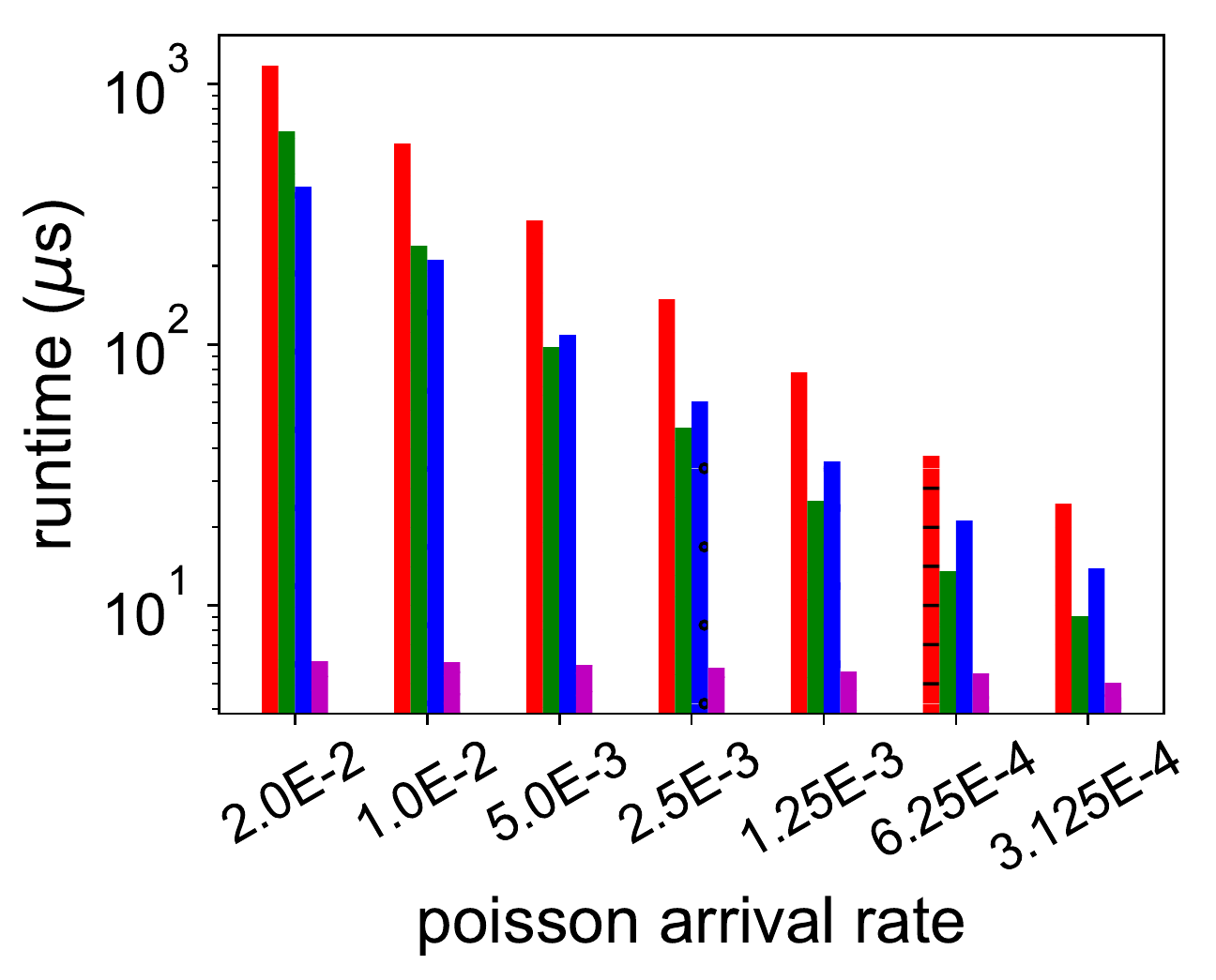}
  }
  \subfloat[\drift]{
  \centering \includegraphics[width=0.23\textwidth]{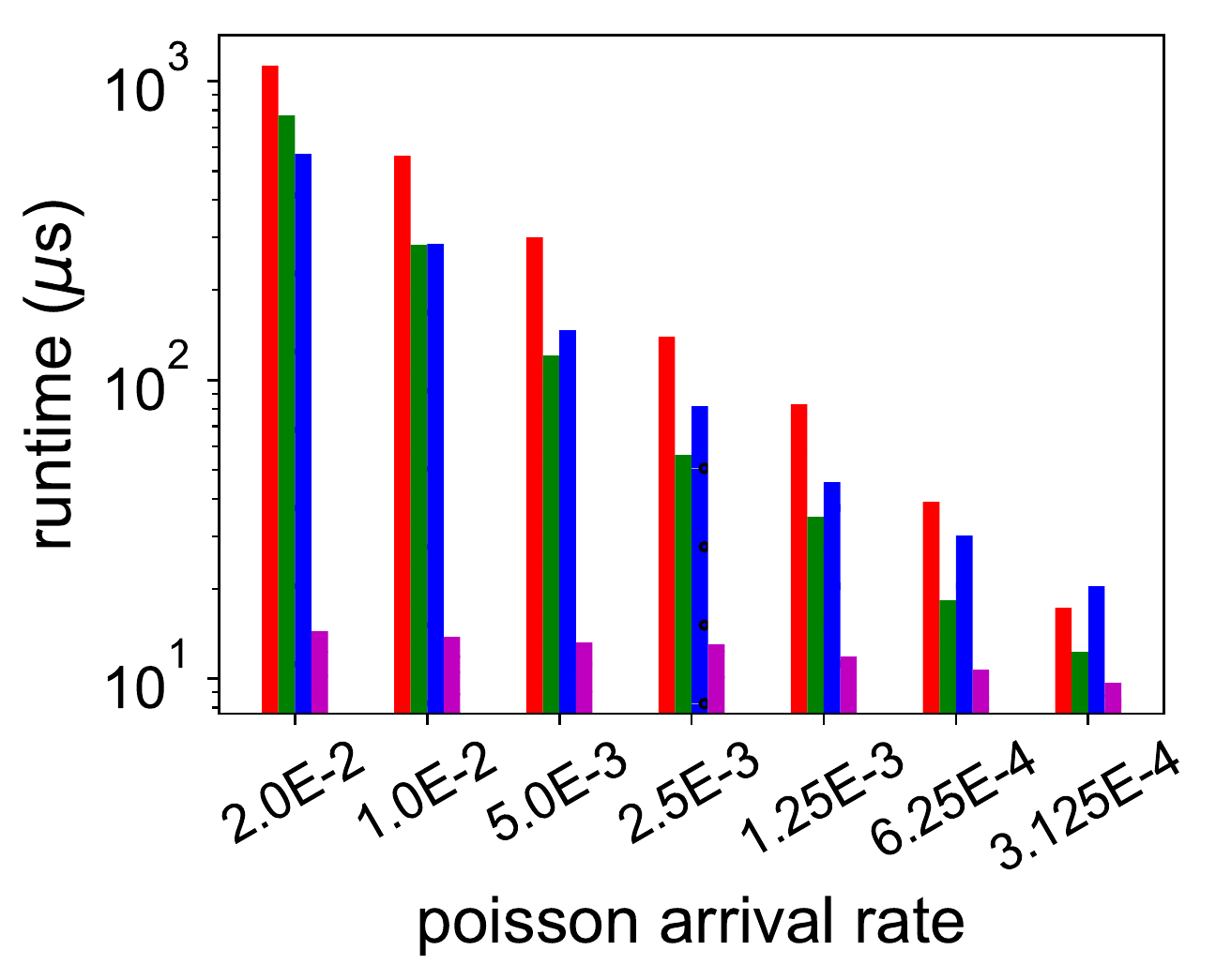}
  }
  \caption{Total time per point (microseconds)  vs. poisson arrival rate $\lambda$.}
 \label{fig:poisson-total}
\end{figure*}


\begin{figure*}
  \centering
  \subfloat[\covtype]{
  \centering \includegraphics[width=0.23\textwidth]{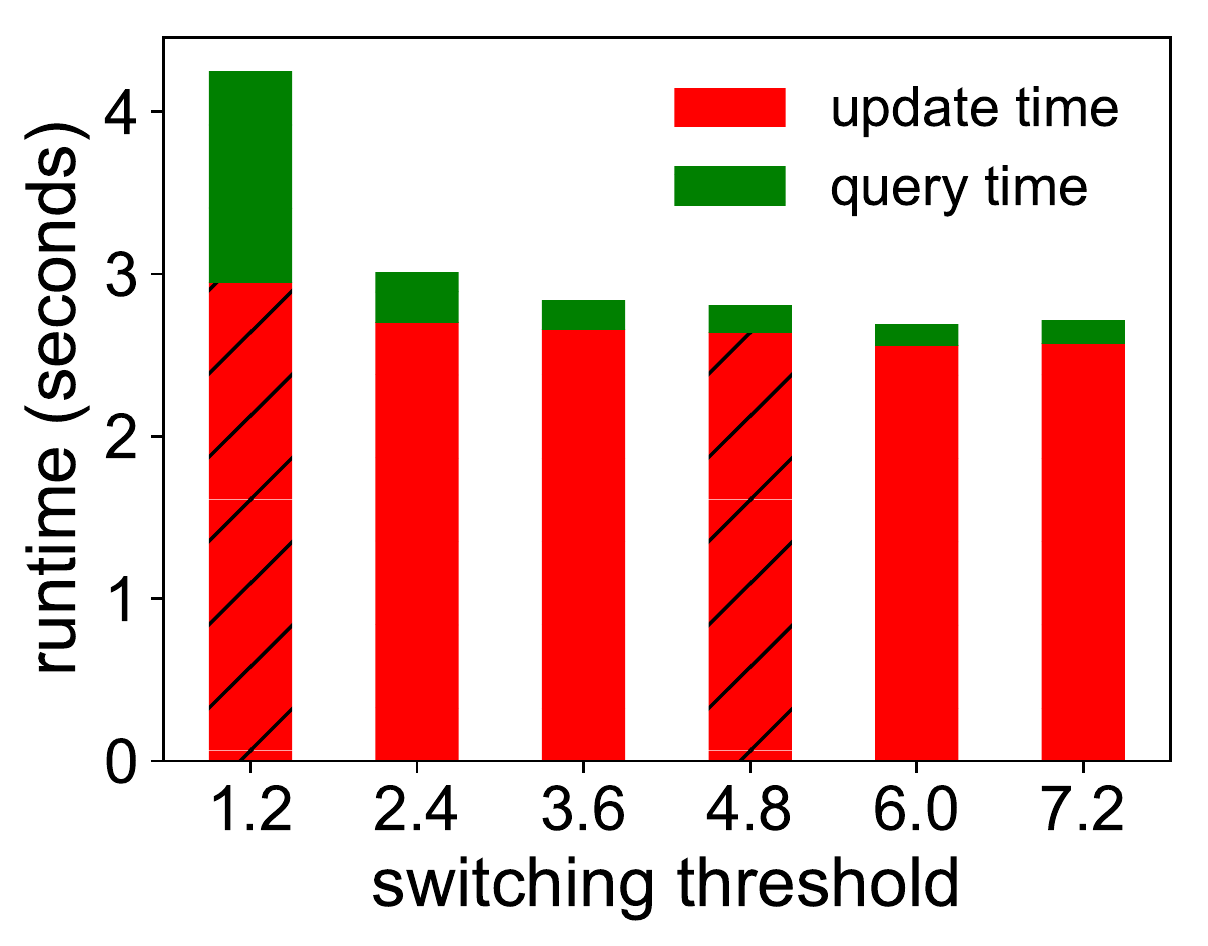}
  }
  \subfloat[\power]{
  \centering \includegraphics[width=0.23\textwidth]{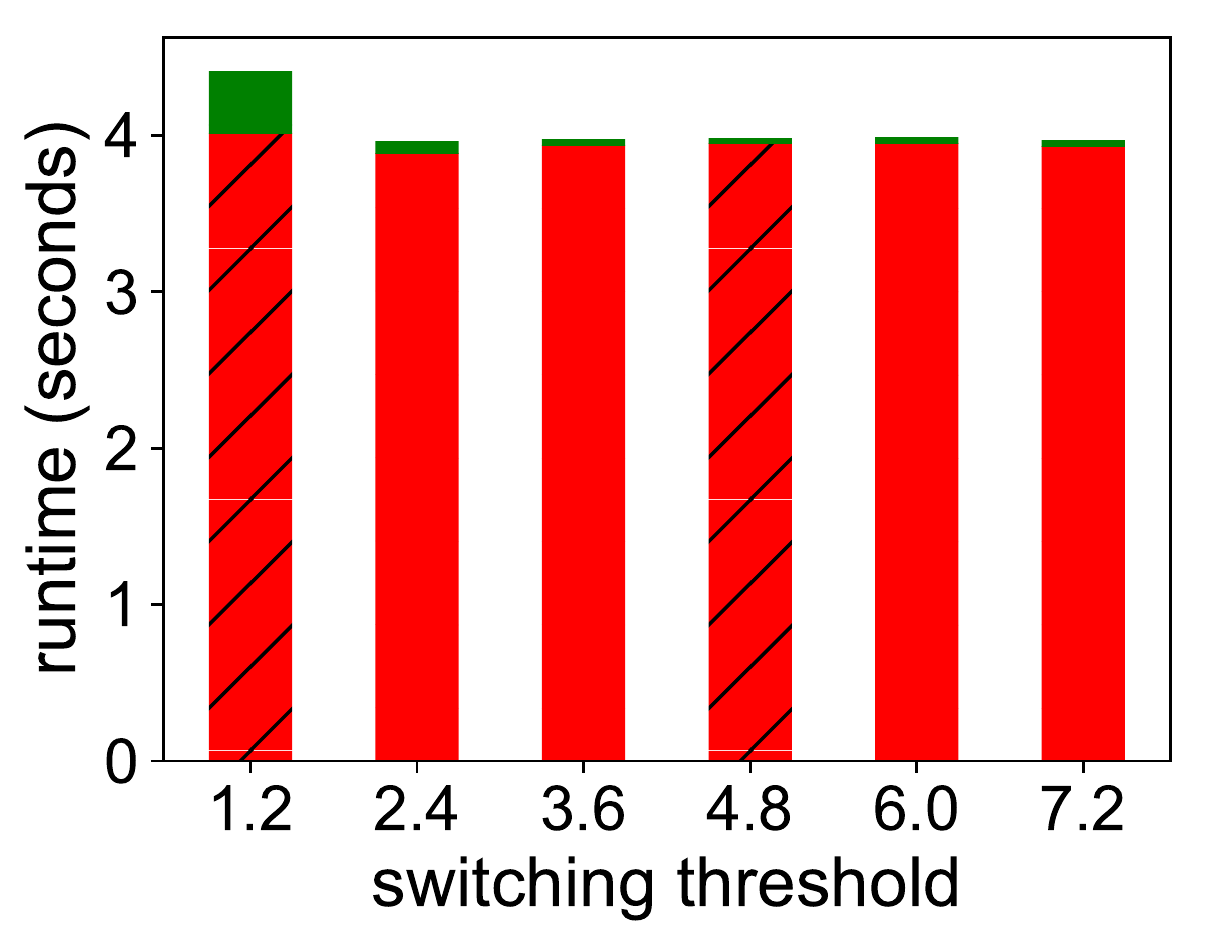}
  }
  \subfloat[\intrusion]{
  \centering \includegraphics[width=0.23\textwidth]{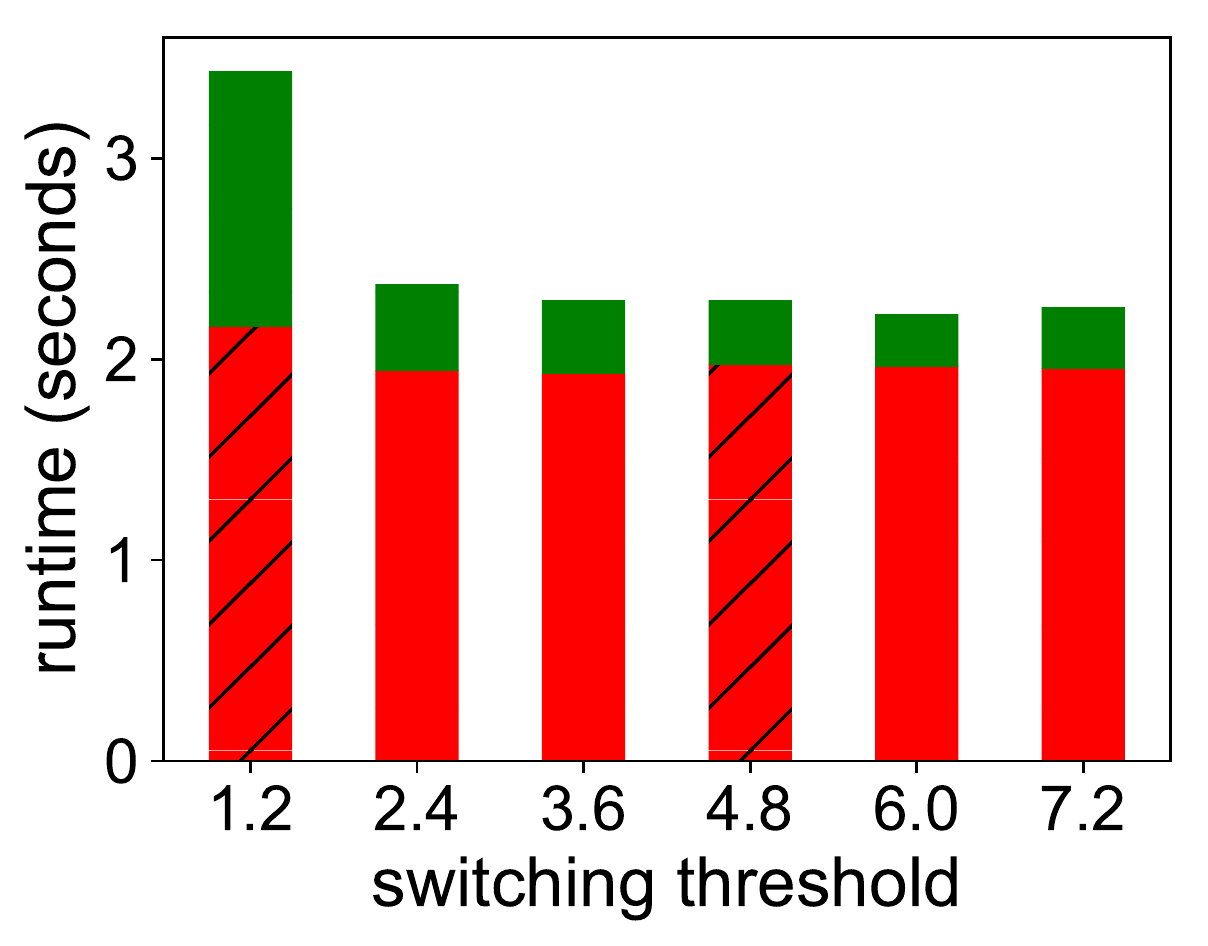}
  }
  \subfloat[\drift]{
  \centering \includegraphics[width=0.23\textwidth]{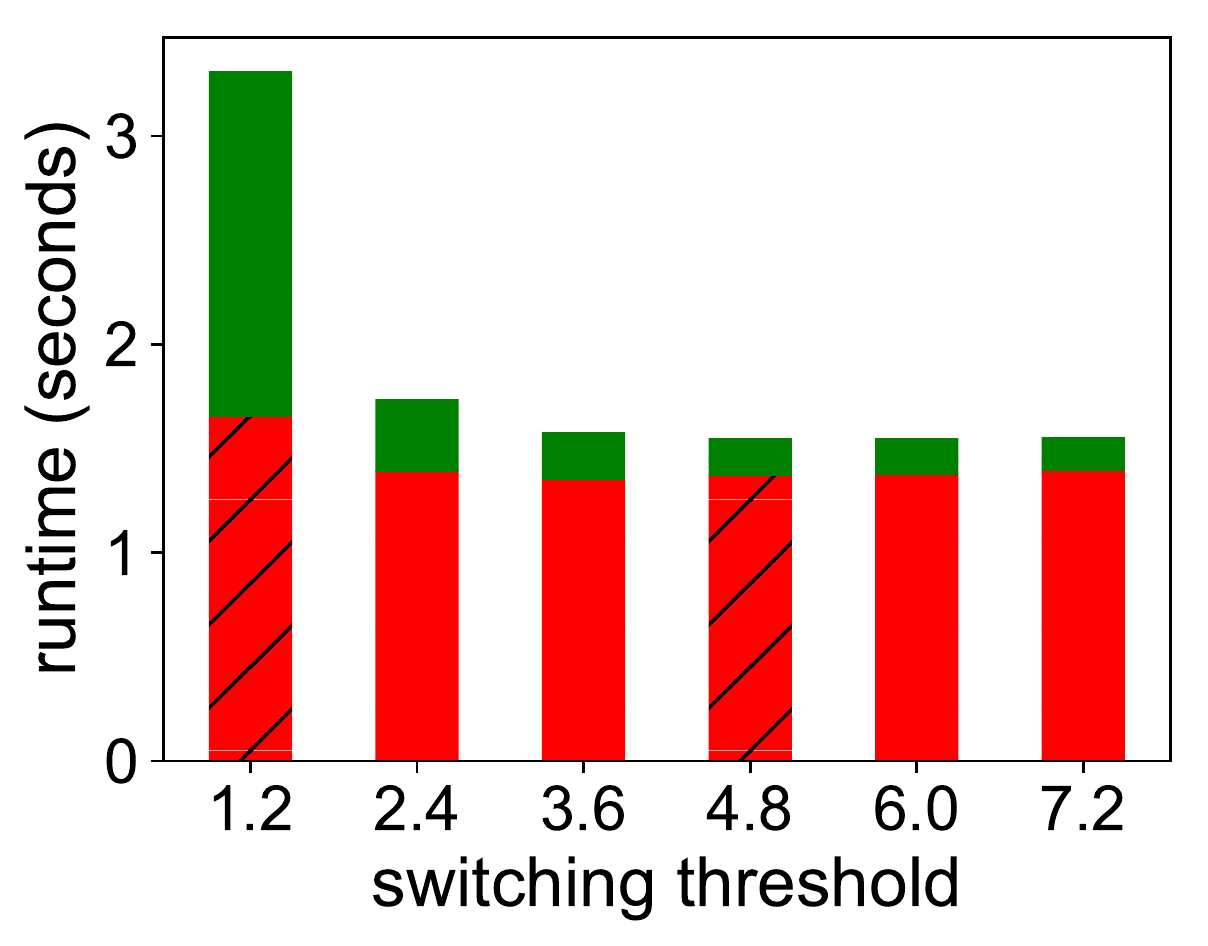}
  }
  \caption{Total runtime (seconds)  vs. switch threshold $\alpha$ in \hybrid algorithm. 
  The number of clusters $k=30$, query interval $q=100$.  
  The update and query time are both counted for the whole stream instead of per point.}
 \label{fig:time-versus-threshold}
\end{figure*}

\begin{table*}
{ 
\small
\centering
\begin{tabular}
{|c|c c c c| c c c c|}
\hline
\multirow{2}{*}{Dataset} & \multicolumn{4}{ c| }{Memory cost in points}  & \multicolumn{4}{ c| }{Memory cost in Megabytes (MB)} \\  \cline{2-9}
& \skmpp & \cc & \rcc & \hybrid  & \skmpp & \cc & \rcc & \hybrid \\  \hline
\covtype  & 5950  & 11350  & 36550  & 11380  & 2.57 & 4.90  & 15.78  & 4.92  \\   \hline
\power  & 7150  & 13750  & 68950  & 13780  & 0.40  & 0.77  & 3.86  & 0.77  \\   \hline
\intrusion  & 5950  & 11350  & 32950  & 11380  & 1.62 & 3.09  & 8.96  & 3.10  \\   \hline
\drift  & 5350  & 10150  & 20950  & 10180 & 2.91 & 5.52  & 11.40  & 5.54   \\   \hline
\end{tabular}
\smallskip
\caption{Memory cost of algorithms, number of clusters $k$ is set to $30$ .}
\label{table:memory-cost}
}
\end{table*}

\subsection{Discussion of Experimental Results}
{\bf Accuracy (\km cost) vs. $k$:} Figures~\ref{fig:cost-versus-k} 
shows the result of \km cost under different number of clusters $k$. 
Note that for the \intrusion data, the result of \seqkm is not shown 
since the cost is much larger (by a factor of about ${10}^4$) than 
the other algorithms. Not surprisingly, for all the algorithms studied, 
the clustering cost decreases with $k$. 
For all the datasets, \seqkm always achieves distinct higher \km cost than other algorithms. 
This shows that \seqkm is consistently worse than the other algorithms, 
when it comes to clustering accuracy---this is as expected, since
unlike the other algorithms, \seqkm does not have a theoretical guarantee on
clustering quality.

The other algorithms, \skmpp, \cc, \rcc, and \hybrid all achieve very similar
clustering cost, on all datasets. In Figure~\ref{fig:cost-versus-k}, we also
show the cost of running a batch algorithm \kmpp (followed by Lloyd's iterations). 
We found that the clustering costs of the streaming algorithms are nearly 
the same as that of running the batch algorithm, which can see the input all at once! 
Indeed, we cannot expect the streaming clustering algorithms to perform any better than this.

Theoretically, it was shown that clustering accuracy improves with the merge
degree.  Hence, \rcc should have the best clustering accuracy (lowest
clustering cost). But we do not observe such behaviors experimentally; 
\rcc and \skmpp show similar accuracy. However their accuracy matches that of batch \kmpp. 
A possible reason for this may be that our theoretical analysis of streaming
clustering methods is too conservative. 

\noindent{\bf Total Runtime vs. Query Interval:} 
We next consider the effect of different query intervals on the runtime. 
Figure~\ref{fig:time-versus-q} shows the total runtime throughout the 
whole stream as a function of the query interval $q$. 
We note that the total time for \hybrid is consistently the smallest, and does not
change with an increase in $q$. This is because \hybrid essentially maintains
the cluster centers on a continuous basis, while occasionally falling back to
\cc to recompute coresets, to improve its accuracy. For the other algorithms,
including \cc, \rcc, and \skmpp, the query time and the total time decrease as
$q$ increases (and queries become less frequent).
\cc and \rcc have similar total runtime and achieve half of the runtime of \skmpp, 
by using the cache. 
All the algorithms converge their total runtime when the queries are very less frequent,
that $q$ is more than $1600$ points.

\noindent\textbf{Metrics vs. Bucket Size:} 
We measure the performance of algorithms under different bucket sizes. 
The bucket size ranges from $20 \cdot k$ to $100 \cdot k$, where $k$ is 
the number of clusters and set to $30$.
Figure~\ref{fig:cost-versus-m} compares the \km cost of different algorithms.
The accuracy is similar with different bucket sizes, even though in theory, it should have better accuracy with larger bucket sizes. This observation matches the results in \cite{AMR+12}, that for most cases in practice, bucket size of $20 \cdot k$ is a good number for \skmpp on clustering accuracy. For our algorithms with coreset caching, the same parameter setting on bucket size applies. 

Figure~\ref{fig:total-versus-m} shows the result of average run time per point, which is the sum of both update and query time. 
Our first observation is that all the timing results are increasing with the bucket size, as both the update and query time are proportional to the bucket size. We also note when bucket size increases over $80 \cdot k$, the query time of \cc exceeds the query time of \skmpp. The reason is when bucket size increases, the number of buckets received in total becomes smaller and in turn the depth of the coreset tree becomes shorter. Thus for \skmpp, the number of buckets to merge during the query is trivially different than using the cache. Comparing to \skmpp, as \cc uses additional time on inserting new coreset to the cache (line $17$ in Algorithm~\ref{algo:cctree-coreset}), the query time of \cc exceeds \skmpp when bucket size is large. 

\noindent\textbf{Queries in Poisson Process:}
We consider the queries arrive in a poisson process instead of the query interval is in the 
fixed number of points. The average update time per point, query time per point and total 
are shown in Figure~\ref{fig:poisson-update},~\ref{fig:poisson-query} 
and~\ref{fig:poisson-total} respectively, with different value of arrival rate. 
Note that the higher value of arrival rate, means the less frequent queries. 
The update time does not have a changing trend with the increasing value of arrival rate, 
as changing query arrival rate only affects the query process. For all the algorithms,
the query time per point drops down with lower arrival rate, as the less frequent queries. 
Comparing different algorithms, \skmpp uses most query time without caching. 
Under high arrival rate $0.02$ such that the average query interval is $50$ points, 
the query time of \rcc is less than \cc. When the arrival rate decreases, \cc has lower 
query time.  The reason is as follows: generally \rcc needs to merge multiple levels of 
coresets comparing to \cc, which only needs to merge one coreset from coreset tree and 
the other coreset in the cache. However, as \rcc applies multiple levels of caches, the chance that 
successfully finding the target coreset in the cache is much higher than \cc. Thus, when queries 
becomes very frequent, with the help of multiple level of caching, the query time of \rcc is faster than \cc.  
Like what we observed in previous experiments, \hybrid achieves the furthest time in query 
due to the nature of online cluster centers maintenance.
As the query time dominates than the update time, the total runtime per point shown in 
Figure~\ref{fig:poisson-total}, which is summation of the two, has similar trend as query time per point.

\noindent\textbf{Switching Threshold of \hybrid:}
We consider the impact of switching threshold parameter to the \hybrid algorithm.
The runtime throughout the whole stream is shown in Figure~\ref{fig:time-versus-threshold}. 
From the plot we first observe that runtime decreases with higher value of switching threshold, 
which indicates the looser requirement on the clustering accuracy. We also notice that the runtime 
drops dramatically when changing from $1.2$ to $2.4$, approximately $3$ to $5$ times. 
But much less decrease when the threshold increases further. Thus, the ideal switching threshold value 
for \hybrid algorithm is $2$ to $4$ if it has already fulfilled the requirement on accuracy.

\noindent\textbf{Memory Usage:} Finally, we report the memory cost in
Table~\ref{table:memory-cost} using $k=30$; the trends are identical for other
values of $k$. Evidently, \skmpp uses the least memory since it only maintains
the coreset tree. Because it also maintains a coreset cache, \cc requires
additional memory. Even then, \cc's memory cost is less than 2x that of
\skmpp. The memory cost of \hybrid is similar to \cc while \rcc has the highest
memory cost.  This shows that the marked improvements in speed requires only a
modest increase in the memory requirement, making the proposed algorithms
practical and appealing.

\section{Conclusion}
\label{sec:concl}
We presented new streaming algorithms for $k$-means clustering.  Compared to
prior methods, our algorithms significantly speedup query processing, while
offering provable guarantees on accuracy and memory cost.  The general framework
of ``coreset caching'' may be applicable to other streaming algorithms built
around the Bentley-Saxe decomposition.  For instance, applying it to streaming
$k$-median seems natural. Many other open questions remain, including
(1)~improved handling of concept drift, through the use of time-decaying
weights, and (2)~clustering on distributed and parallel streams.


\section*{Acknowledgments}
This work is supported in part by the US National Science Foundation through awards 1527541 and 1632116.

\bibliographystyle{IEEEtran}
\bibliography{IEEEtran/IEEEabrv,stream-clustering}

\begin{thebibliography}{10}
\providecommand{\url}[1]{#1}
\csname url@samestyle\endcsname
\providecommand{\newblock}{\relax}
\providecommand{\bibinfo}[2]{#2}
\providecommand{\BIBentrySTDinterwordspacing}{\spaceskip=0pt\relax}
\providecommand{\BIBentryALTinterwordstretchfactor}{4}
\providecommand{\BIBentryALTinterwordspacing}{\spaceskip=\fontdimen2\font plus
\BIBentryALTinterwordstretchfactor\fontdimen3\font minus
  \fontdimen4\font\relax}
\providecommand{\BIBforeignlanguage}[2]{{%
\expandafter\ifx\csname l@#1\endcsname\relax
\typeout{** WARNING: IEEEtran.bst: No hyphenation pattern has been}%
\typeout{** loaded for the language `#1'. Using the pattern for}%
\typeout{** the default language instead.}%
\else
\language=\csname l@#1\endcsname
\fi
#2}}
\providecommand{\BIBdecl}{\relax}
\BIBdecl

\bibitem{AMR+12}
M.~R. Ackermann, M.~M\"{a}rtens, C.~Raupach, K.~Swierkot, C.~Lammersen, and
  C.~Sohler, ``Streamkm++: A clustering algorithm for data streams,'' \emph{J.
  Exp. Algorithmics}, vol.~17, no.~1, pp. 2.4:2.1--2.4:2.30, 2012.

\bibitem{AJM09}
N.~Ailon, R.~Jaiswal, and C.~Monteleoni, ``Streaming k-means approximation,''
  in \emph{NIPS}, 2009, pp. 10--18.

\bibitem{GMM+03}
S.~Guha, A.~Meyerson, N.~Mishra, R.~Motwani, and L.~O'Callaghan, ``Clustering
  data streams: Theory and practice,'' \emph{IEEE TKDE}, vol.~15, no.~3, pp.
  515--528, 2003.

\bibitem{SWM11}
M.~Shindler, A.~Wong, and A.~Meyerson, ``Fast and accurate k-means for large
  datasets,'' in \emph{NIPS}, 2011, pp. 2375--2383.

\bibitem{HM04}
S.~Har-Peled and S.~Mazumdar, ``On coresets for k-means and k-median
  clustering,'' in \emph{STOC}, 2004, pp. 291--300.

\bibitem{AV07}
D.~Arthur and S.~Vassilvitskii, ``k-means++: The advantages of careful
  seeding,'' in \emph{SODA}, 2007, pp. 1027--1035.

\bibitem{MacQueen67}
J.~B. MacQueen, ``Some methods for classification and analysis of multivariate
  observations,'' in \emph{Proc. of the fifth Berkeley Symposium on
  Mathematical Statistics and Probability}, 1967, pp. 281--297.

\bibitem{Lloyd82}
S.~Lloyd, ``Least squares quantization in {PCM},'' \emph{IEEE Trans.
  Information Theory}, vol.~28, no.~2, pp. 129--136, 1982.

\bibitem{MBY+15}
X.~Meng, J.~K. Bradley, B.~Yavuz, and et~al., ``{MLlib: Machine Learning in
  Apache Spark},'' \emph{J. Machine Learning Research}, vol.~17, pp.
  1235--1241, 2016.

\bibitem{KMN+04}
T.~Kanungo, D.~M. Mount, N.~S. Netanyahu, C.~D. Piatko, R.~Silverman, and A.~Y.
  Wu, ``A local search approximation algorithm for k-means clustering,''
  \emph{Computational Geometry}, vol.~28, no. 2 - 3, pp. 89 -- 112, 2004.

\bibitem{ZRL96}
T.~Zhang, R.~Ramakrishnan, and M.~Livny, ``Birch: An efficient data clustering
  method for very large databases,'' in \emph{{SIGMOD}}, 1996, pp. 103--114.

\bibitem{AHW+03}
C.~C. Aggarwal, J.~Han, J.~Wang, and P.~S. Yu, ``A framework for clustering
  evolving data streams,'' in \emph{{PVLDB}}, 2003, pp. 81--92.

\bibitem{BS80}
J.~L. Bentley and J.~B. Saxe, ``Decomposable searching problems i.
  static-to-dynamic transformation,'' \emph{Journal of Algorithms}, vol.~1, pp.
  301 -- 358, 1980.

\bibitem{HK07}
S.~Har-Peled and A.~Kushal, ``Smaller coresets for k-median and k-means
  clustering,'' \emph{Discrete Computational Geometry}, vol.~37, no.~1, pp.
  3--19, 2007.

\bibitem{FL11}
D.~Feldman and M.~Langberg, ``A unified framework for approximating and
  clustering data,'' in \emph{STOC}, 2011, pp. 569--578.

\bibitem{FSS13}
D.~Feldman, M.~Schmidt, and C.~Sohler, ``Turning big data into tiny data:
  Constant-size coresets for k-means, pca and projective clustering,'' in
  \emph{SODA}, 2013, pp. 1434--1453.

\bibitem{Lichman13}
\BIBentryALTinterwordspacing
M.~Lichman, ``{UCI} machine learning repository,'' 2013. [Online]. Available:
  \url{http://archive.ics.uci.edu/ml}
\BIBentrySTDinterwordspacing

\bibitem{BGE+16}
J.~P. Barddal, H.~M. Gomes, F.~Enembreck, and J.~P. Barthes, ``{SNCStream+:
  Extending a high quality true anytime data stream clustering algorithm},''
  \emph{Information Systems}, vol.~62, pp. 60 -- 73, 2016.

\bibitem{BHK+10}
A.~Bifet, G.~Holmes, R.~Kirkby, and B.~Pfahringer, ``{MOA: Massive Online
  Analysis},'' \emph{J. Machine Learning Research}, vol.~11, pp. 1601--1604,
  2010.

\end{thebibliography}

%

\begin{IEEEbiography}[{\includegraphics[width=1in,height=1.25in,clip,keepaspectratio]{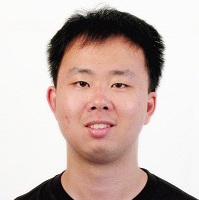}}]{Yu Zhang}
Yu Zhang is a Ph.D. student in the Department of Electrical and Computer Engineering at Iowa State University. He received his M.S. in Computer Engineering from University of Central Florida in 2012, and his B.S. in Electrical Engineering from University of Science and Technology of China in 2011. He is interested in data stream mining and algorithm design for machine learning.
\end{IEEEbiography}

\begin{IEEEbiography}[{\includegraphics[width=1in,height=1.25in,clip,keepaspectratio]{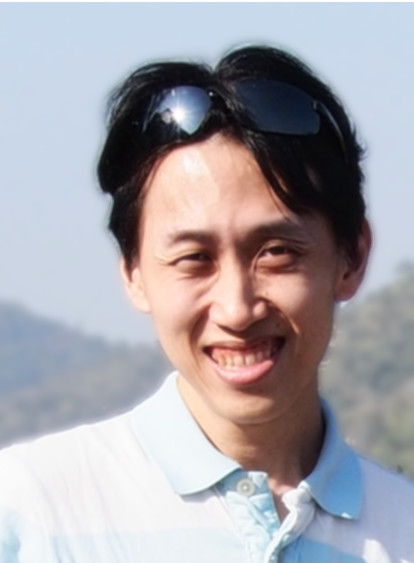}}]{Kanat Tangwongsan}
Kanat Tangwongsan is a computer scientist on the faculty of Mahidol University International College. He received his Ph.D. and B.S. in Computer Science from Carnegie Mellon University in 2010 and 2006. He worked at IBM T.J. Watson Research Center as a research staff member from 2010 to 2015. His current research interests are parallel algorithms design for massive data, both in theory and practice.
\end{IEEEbiography}

\begin{IEEEbiography}[{\includegraphics[width=1in,height=1.25in,clip,keepaspectratio]{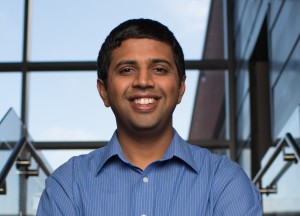}}]{Srikanta Tirthapura}
Srikanta Tirthapura is a Professor in the department of Electrical and Computer Engineering at Iowa State University. He received his Ph.D. in Computer Science from Brown University in 2002, and his B.Tech. in Computer Science and Engineering from IIT Madras in 1996. His research interests are algorithm design for big data, including data stream algorithms and parallel and distributed algorithms. He is a recipient of the IBM Faculty Award, and the Warren Boast Award for excellence in Undergraduate Teaching.
\end{IEEEbiography}

\end{document}